\documentclass[a4paper,11pt]{article}
\usepackage{soul}
\usepackage[usenames,dvipsnames]{color}
\usepackage{jheppub}
\usepackage{esvect}
\usepackage{comment}
\usepackage{amsmath, amssymb, slashed, epsf, color, graphicx, latexsym, bm}
\usepackage{mathrsfs}
\usepackage{tensor}
\usepackage{physics}
\usepackage{pdflscape}
\usepackage{tikz}
\usetikzlibrary{positioning, calc}
\usepackage{dsfont}
\usepackage{epsfig}
\usepackage{graphics}
\usepackage[T1]{fontenc}
\usepackage{mathtools}
\usepackage{fixme}
\fxsetup{status=draft}
\usepackage{caption}
\usepackage{float}
\usepackage{bbm}
\usepackage{subcaption}
\usepackage{enumitem}
\usepackage{verbatim}
\usepackage{blindtext}
\usepackage{tikz}
  \usepackage{amsmath}
\newlength\squareheight
\makeatletter
\renewcommand*\env@matrix[1][*\c@MaxMatrixCols c]{%
  \hskip -\arraycolsep
  \let\@ifnextchar\new@ifnextchar
  \array{#1}}
\makeatother
\DeclareFontEncoding{LS1}{}{}
\DeclareFontSubstitution{LS1}{stix}{m}{n}
\DeclareSymbolFont{symbols2}{LS1}{stixfrak}{m}{n}
\DeclareMathSymbol{\typecolon}{\mathbin}{symbols2}{"25}
\usepackage[mathlines]{lineno}
\newcommand{\C}{\mathbb{C}}
\newcommand{\R}{\mathbb{R}}
\newcommand{\N}{\mathbb{N}}
\newcommand{\Z}{\mathbb{Z}}
\newcommand{\Q}{\mathbb{Q}}

\newcommand{\s}[1]{\mathcal{#1}}
\newcommand{\mr}[1]{\mathrm{#1}}
\newcommand{\mf}[1]{\mathfrak{#1}}
\newcommand{\ov}[1]{\overline{#1}}

\begin{document}
\title{Rationality of Lorentzian Lattice CFTs And The Associated Modular Tensor Category}
\author[a]{Ranveer Kumar Singh,} \author[b]{Madhav Sinha,}\author[a]{Runkai Tao}
\affiliation[a]{NHETC and Department of Physics and Astronomy, Rutgers University, 126
Frelinghuysen Rd., Piscataway NJ 08855, USA}
\affiliation[b]{Department of Physics and Astronomy, Rutgers University, 126
Frelinghuysen Rd., Piscataway NJ 08855, USA}
\emailAdd{ranveersfl@gmail.com}
\emailAdd{ms3066@physics.rutgers.edu}
\emailAdd{runkai.tao@physics.rutgers.edu}
\abstract{We classify the irreducible modules of a rational Lorentzian lattice vertex operator algebra (LLVOA) based on an even, self-dual Lorentzian lattice $\Lambda\subset\R^{m,n}$ of signature $(m,n)$. We show that the set of isomorphism classes of irreducible modules of the LLVOA are in one-to-one correspondence with the equivalence classes $\Lambda_0^\circ/\Lambda_0$ for a certain subset $\Lambda_0^\circ\subset\R^{m,n}$ and a full rank sublattice $\Lambda_0\subset\Lambda$. We also classify the intertwining operators between the modules and calculate the fusion rules. We then
describe the standard construction of modular tensor category (MTC) associated to rational
LLCFTs. We explicitly construct the modular data and braiding and fusing matrices for the MTC. As a concrete example, we show that the LLCFT based on a
certain even, self-dual Lorentzian lattice of signature $(m, n)$, with $m$ even, realizes the
$D(m \bmod 8)$ level 1 Kac-Moody MTC.} 
\maketitle
\section{Introduction}\label{sec:intro}
The theory of non-chiral vertex operator algebra was initiated in \cite{Singh:2023mom} with the first major example provided by Narain conformal field theories (and their generalizations) in string theory. 
%The construction of these examples, called the Lorentzian lattice vertex operator algebra (LLVOA) in \cite{Singh:2023mom}, is similar to the classical Euclidean lattice VOAs in \cite{Frenkel:1988xz} but is based on Lorentzian lattices. 
More precisely, \cite{Singh:2023mom} gives the construction of a non-chiral vertex operator algebra based on an even, self-dual Lorentzian lattice\footnote{See \cite{Ashwinkumar:2021kav,Ashwinkumar:2023ctt} for a construction of \textit{generalized Narain CFTs} based on non self-dual Lorentzian lattices and \cite{Dolan:1994st,Moore:2023zmv} for a construction of chiral (super)conformal field theories associated to Euclidean lattices.}. This example was called the Lorentzian lattice vertex operator algebra (LLVOA) in \cite{Singh:2023mom}. We refer the reader to Section \ref{sec:nchvoadef} for the definition of non-chiral vertex operator algebras and related notions and Section \ref{sec:llvoa} for the construction of the LLVOA. \par The construction of a class of irreducible modules of the LLVOA was discussed in \cite{Singh:2023mom}. Using these modules, one can construct a modular invariant conformal field theory, which was called the Lorentzian lattice conformal field theory (LLCFT) in \cite{Singh:2023mom}. The chiral algebra of LLCFT was also discussed and it corresponded to certain sublattice of the original lattice. 
We now explain the main results of this paper.
\subsection{Classification Of Irreducible Modules of the LLVOA}\label{sec:intro1.1}
One of the main goals of this paper is to give a complete classification of irreducible modules of the LLVOA and justify the restricted class of modules used to construct the LLCFT. To state the main result in this direction, we introduce some notation which will be also be used in the rest of the paper. Let $\langle\cdot,\cdot\rangle$ denote the standard positive-definite, non-degenerate bilinear form on $\R^m,\R^n$. Consider a $d=(m+n)$-dimensional even, integral Lorentzian lattice $\Lambda\subset \mathbb{R}^{m,n}$  with Lorentzian inner product of  signature $(m,n)$, denoted by $\circ$. As in \cite{Singh:2023mom}, we will often write a vector $\lambda\in\Lambda$ as $\lambda=(\alpha^{\lambda},\beta^{\lambda})$, where $\alpha^{\lambda} \in \R^{m}$ and $\beta^{\lambda} \in \R^{n}$. Then we can write 
\begin{equation}
    \lambda_1\circ\lambda_2=\langle\alpha^{\lambda_1},\alpha^{\lambda_2}\rangle-\langle\beta^{\lambda_1},\beta^{\lambda_2}\rangle\in\mathbb{Z}.
\end{equation}
%Note that in general $\langle\alpha^{\lambda_1},\alpha^{\lambda_2}\rangle,\langle\beta^{\lambda_1},\beta^{\lambda_2}\rangle\not\in\mathbb{Z}$. 
Introduce the abelian groups \begin{equation}\label{eq:lamb12fromlamb}
\begin{split}
     &\Lambda_1=\{\alpha^\lambda \,  \lvert \,  \lambda=(\alpha^\lambda,\beta^\lambda)\in\Lambda\text{ for some }\beta^\lambda\in\R^n\}\subset \R^m,\\&\Lambda_2=\{\beta^\lambda \,  \lvert \,  \lambda=(\alpha^\lambda,\beta^\lambda)\in\Lambda\text{ for some }\alpha^\lambda\in\R^m\}\subset \R^n\\&\Lambda':=\Lambda_1\oplus\Lambda_2~.
\end{split}
\end{equation}
% $\Lambda_1$ and $\Lambda_2$ are not lattices, they are just finitely generated $\Z$ modules possibly with non-trivial torsion. For the lattice $\mathrm{II}_{m,n}$ in \eqref{eq:genmatIImn}, it is easy to see that
% \begin{equation}
% \begin{split}
% (\mathrm{II}_{m,n})_1=\Z^m\bigcup\left(\Z+\frac{1}{2}\right)^m \\(\mathrm{II}_{m,n})_2=\Z^n\bigcup\left(\Z+\frac{1}{2}\right)^n. 
% \end{split}
% \end{equation}
We further identify even integral Euclidean sublattices of $\Lambda$ as follows:
\begin{equation}\label{eq:Lambi0}
\begin{split}
    &\Lambda_1^0:=\{\alpha \in \R^{m}  \, \lvert \, (\alpha,0)\in\Lambda \}\\& \Lambda_2^0:=\{\beta \in \R^{n}  \, \lvert \, (0,\beta)\in\Lambda \}\\&\Lambda_{0} :=\Lambda_1^{0} \oplus  \Lambda_2^{0}~.
    \end{split}
\end{equation} 
Also introduce the set\footnote{Note that, in general, $\Lambda_0^\circ$ is not closed under addition and hence is just a set.} 
\begin{equation}
    \Lambda_0^\circ:=\{\lambda\in\Lambda':\lambda\circ\lambda\in2\Z\}\subseteq\Lambda'~.
\end{equation}
% Let $\{\lambda_i \equiv (\alpha^{\lambda_i}, \beta^{\lambda_i})\}_{i=1}^d$ be a basis of $\Lambda$. Then it is easy to see that 
% \begin{equation}
%     \Lambda_1=\mathsf{Span}_{\Z}\{\alpha^{\lambda_i}\}_{i=1}^d,\quad\Lambda_2=\mathsf{Span}_{\Z}\{\beta^{\lambda_i}\}_{i=1}^d.
% \end{equation}
Note that in general 
\begin{equation}
    \Lambda_0\subsetneq \Lambda\subsetneq \Lambda_0^\circ\subsetneq \Lambda'~.
\end{equation} 
Note that $\Lambda_0+\Lambda_0^\circ=\Lambda_0^\circ$.
Let $\Lambda_0^\circ/\Lambda_0$ be the set of equivalence classes under the equivalence relation on $\Lambda_0^\circ$ given by 
\begin{equation}
    \lambda\sim\lambda'\iff \lambda-\lambda'\in\Lambda_0,\quad \lambda,\lambda'\in \Lambda_0^\circ~.
\end{equation}
Then the we prove the following analogue of the Dong's classification result \cite{Dong1993VertexAA,lepowsky2012introduction} for LLVOAs:
\begin{thm}
Assume that $|\Lambda/\Lambda_0|<\infty$. Then the irreducible modules of the LLVOA based on $\Lambda$ are in one-to-one correspondence with the equivalence classes $\Lambda_0^\circ/\Lambda_0$.     
\end{thm}
This theorem can be deduced from Theorem \ref{thm:LLVOAmodconst} and Theorem \ref{thm:classllvoamod} proved in this paper. We further classify all intertwining operators of the irreducible modules and compute the fusion rules, see Section \ref{sec:nchvoadef} for relevant definitions. 
%We show that a necessary condition for the tensor functor to be closed on the category of modules of the LLVOA is that we only include the irreducible modules (and their finite direct sums) corresponding to cosets $\widetilde{\Lambda}/\Lambda_0$ where $\widetilde{\Lambda}\subseteq\Lambda_0^\circ$ is an additive group and $\Lambda_0\subseteq \widetilde{\Lambda}$. This justifies the selection of irreducible modules used to construct the LLCFT in \cite{Singh2024c,Singh:2023mom}. We leave the explicit construction of the tensor bifunctor for the this category to future work. 
\subsection{Modular Tensor Category Of A Rational LLCFT}
The second main goal of this paper is to study several equivalent conditions for when the number of modules of the in the LLCFT are finitely many, that is when the LLCFT is a rational conformal field theory (RCFT). For the case when the signature of Lorentzian lattice is $(m,m)$, this question was extensively studied in \cite{Wendland:2000ye}. We also show how our results are equivalent to the results in \cite{Wendland:2000ye} for this case.
\par
RCFTs have been a subject of immense study due to their rich mathematical structure. Moore and Seiberg \cite{Moore:1988qv,Moore:1988uz,Moore:1989vd} studied various algebraic notions associated to an RCFT and introduced the notion of modular tensor categories (MTC). Roughly speaking the data for a modular tensor category \cite{Fuchs:2002cm,bakalov2001lectures} are the modular matrices, fusion rules, and the braiding and fusing matrices. Braiding and fusing matrices are obtained from the CFT by looking at the transformation of conformal blocks under certain transformation of the crossing ratios \cite{Moore:1988qv,DiFrancesco:1997nk,Blumenhagen:2009zz}. Moore and Seiberg derived polynomial equations, called hexagon and pentagon relations, for the braiding and fusing matrices and gave a proof of the Verlinde formula \cite{Moore:1988uz}. These polynomial equations were later formalised in the language of fusion and modular tensor categories. \par Mathematically, an MTC can be constructed from a vertex operator algebra (VOA) by looking at the category of irreducible modules of the VOA. This category is first endowed with a tensor functor\footnote{Note that the usual tensor product of modules does not work. The correct notion of tensor product of modules  of a VOA, introduced by Huang and Lepowsky in \cite{huang1995theory1,huang1995theory2,huang1995theory3,huang1995theory4}, is quite complicated. The construction also requires the VOA to have some ``nice'' properties, see Footnote \ref{foot:nicevoa}.} which makes it into a tensor category \cite{huang1995theory1,huang1995theory2,huang1995theory3,huang1995theory4}. The construction of this tensor product requires the introduction of \textit{intertwining operators} \cite{Frenkel:1993xz} which are motivated by Moore-Seiberg's \textit{chiral vertex operators} \cite{Moore:1988qv,Moore:1989vd}. Using \textit{intertwining operators}, one defines braiding and fusing matrices. It was shown by Huang, that this data of braiding and fusing matrices along with the tensor category of modules of the VOA possesses the structure of an MTC \cite{huang2005vertex1,huang2008vertex2,huang2008rigidity}. 
\par In this paper, we compute the MTC associated to the left-moving and right-moving sectors of a rational modular invariant LLCFT and show that they are identical. We follow the standard construction of an MTC from a chiral CFT. We take the compact boson at $R=\sqrt{2k},~k\in\N$ and LLCFT based on $\mathrm{II}_{m,n}$, a specific even self-dual Lorentzian lattice, as two concrete examples. In the first case, we obtain an MTC of rank $2k$ while in the second case, the concrete data we get gives the realization of $D(m\!\!\mod 8)$ level $1$ Kac-Moody MTC from LLCFT on $\mathrm{II}_{m,n}$, where $m - n \bmod 24 \equiv 0$. 
\subsection{Organisation Of The Paper}
Our main result and paper organization are as follows: Section \ref{sec:nchvoadef}, we 
 recall the basic definitions in non-chiral VOA.
 In Section \ref{sec:llvoa}, we provide a brief construction of LLVOA and delves into the question of rationality.
 The focal point of this section is Theorem \ref{thm:rcft}, where we explore various equivalent statements regarding the rationality of LLCFT, which generalises Wendland's rational Narain CFT characterization \cite{Wendland:2000ye}. 
 In section \ref{sec:moduleduallattice}, theorem \ref{thm:classllvoamod} gives module classification for any rational LLCFT.
 In section \ref{sec:intertwiners}, all the intertwining operators of an LLCFT are classified.
 In Section \ref{sec:mtc}, we first explained how our modular invariant rational LLCFT can be separated into two VOAs on Euclidean lattice, and
 then derive the $S$ and $T$ matrices 
 as well as the $\mathcal{F}$ and $\mathcal{B}$ matrices for both left and right-moving MTC. We argue that the MTCs for left and right sector are the same. Later in this section, we focus on several examples.

\section{Non-Chiral Vertex Operator Algebra, Modules And Intertwining Operators}\label{sec:nchvoadef}
In this section we briefly recall some basic definitions and results, see \cite{Singh:2023mom} for details. 
\subsection{Non-Chiral VOA}
Let $x,\bar{x}$ be formal variables and $V\{x,\bar{x}\}$ denote the vector space of formal power series in $x,\bar{x}$ indexed over $\R\times\R$ with coefficients in $V$.
\begin{defn}\label{def:nonchiralVOA}
A \textit{non-chiral vertex operator algebra} is an $(\R\times \R)$-graded complex vector space $V$,
\begin{equation}\label{eq:VRxRgrad}
V=\coprod\limits_{(h,\bar{h})\in\R\times\R}V_{(h,\bar{h})} \, ,
\end{equation} 
equipped with the linear map
\begin{equation}\label{eq:veropmap}
\begin{split}
    Y_V: \, &V\longrightarrow \mathsf{End}(V)\{x,\bar{x}\} \, , \\&u\longmapsto Y_V(u,x,\bar{x})=\sum_{\substack{m,n\in\R}}u_{m,n}x^{-m-1}\bar{x}^{-n-1} \, ,
\end{split}
\end{equation}
called the  \textit{vertex operator map}, and distinguished vectors $\omega\in V_{(2,0)}, ~\bar{\omega}\in V_{(0,2)},\mathbf{1}\in V_{(0,0)}$ satisfying the following properties:
\begin{enumerate}
\item\label{item:idprop} \textit{Identity property:} $Y_V(\mathbf{1}, x, \bar{x})=\mathds{1}_V$.
\item\label{item:gradres} \textit{Grading-restriction property:} $\text{dim}(V_{(h,\bar{h})})<\infty$ for every $(h,\Bar{h})\in\R\times \R$ and  
there exists $M\in\R$, such that 
$V_{(h,\bar{h})}=0$, for $h<M$ or $\Bar{h}<M$.    
\item\label{item:singvalprop} \textit{Single-valuedness property:} $h-\bar{h}\in\mathbb{Z}$ for every homogenous subspace $V_{(h,\bar{h})}$. 
\item \label{item:creatprop}\textit{Creation property:} For any $v\in V$, $\lim_{x,\bar{x}\to 0}Y_V(v,x,\bar{x})\mathbf{1}=v.$ 
\item\label{item:virprop} \textit{Virasoro property:} The \textit{conformal vertex operators} $Y_V(\omega,x,\bar{x})$ and $Y_V(\bar{\omega},x,\bar{x})$ have Laurent series in $x,\bar{x}$ of the form 
    \begin{equation}\label{eq:sttenvoa}
     \begin{split}
&Y_V(\omega,x,\bar{x})=\sum_{n\in\Z}L(n)x^{-n-2}\\
&Y_V(\bar{\omega},x,\bar{x})=\sum_{n\in\Z}\bar{L}(n)\bar{x}^{-n-2}    
     \end{split}   
    \end{equation}
where $L(n),\bar{L}(n)$ are operators which satisfy the \textit{Virasoro algebra} with \textit{central charge} $c,\bar{c}$ respectively:
\begin{equation}
\begin{split}
    &\left[L(m),L(n)\right]=(m-n) L(m+n)+\frac{c}{12} m\left(m^{2}-1\right) \delta_{m+n,0}~,\\&\left[\bar{L}(m),\bar{L}(n)\right]=(m-n) \bar{L}(m+n)+\frac{\bar{c}}{12} m\left(m^{2}-1\right) \delta_{m+n,0}~,\\&\left[L(m),\bar{L}(n)\right]=0~.
\end{split}
\label{eq:viraalg}
\end{equation}
\item\label{item:gradprop} \textit{Grading property:} The operator $(L(0),\bar{L}(0))$ is the grading  operator on $V$, that is for $v\in V_{(h,\bar{h})}$
\begin{equation}
    L(0)v=hv,\quad \Bar{L}(0)v=\bar{h}v.
\end{equation}
\item\label{item:l0prop} \textit{$L(0)$-property} : 
    \begin{equation}
    \begin{split}
        &[L(0), Y_{V}(u , x, \bar{x})]=x \frac{\partial}{\partial x}Y_{V}(u ,x, \bar{x})+Y_{V}(L(0) u , x, \bar{x})\\&[\bar{L}(0), Y_{V}(u , x, \bar{x})]=\bar{x} \frac{\partial}{\partial \bar{x}}Y_{V}(u ,x, \bar{x})+Y_{V}(\bar{L}(0) u , x, \bar{x})
    \end{split}
\end{equation}
\item\label{item:transprop} \textit{Translation property:} For any $u\in V$

\begin{equation}
\begin{split}
& {\left[L(-1), Y_{V}(u , x, \bar{x})\right]=Y_{V}\left(L(-1) u , x, \bar{x}\right)=\frac{\partial}{\partial x} Y_{V}(u , x, \bar{x}),} \\
& {\left[\bar{L}(-1), Y_{V}(u , x, \bar{x})\right]=Y_{V}\left(\bar{L}(-1) u , x, \bar{x}\right)=\frac{\partial}{\partial \bar{x}} Y_{V}(u , x, \bar{x})} \\
&
\end{split}
\end{equation}
\item \label{item:locprop}\textit{Locality property:} 
For $u_1, \dots, u_n\in V$, there is an operator-valued  function 
$$\mf{M}^V_n (u_1, ... , u_n; z_1, \bar{z}_1, ... , z_n, \bar{z}_n)~,$$ defined on 
\begin{equation}
    \{(z_1,\dots,z_n,\bar{z}_1,\dots, \bar{z}_n) \in \C^{2n} \, \lvert \, z_i,\bar{z}_i\neq 0, z_i \neq z_j,\bar{z}_i\neq\bar{z}_j\} \, ,
\end{equation} which is multi-valued and analytic
when $\bar{z}_1, ..., \bar{z}_n$ are viewed as  independent complex variables and is single-valued when $\bar{z}_1, ..., \bar{z}_n$ are equal to the complex conjugates of $z_1, ... , z_n$. Moreover, for any permutation $\sigma\in S_n$, the product of vertex operators 
\begin{equation}\label{eq:veropprodvoa}
Y_V\left(u_{\sigma(1)} , z_{\sigma(1)}, \bar{z}_{\sigma(1)}\right)\cdots  Y_V\left(u_{\sigma(n)} , z_{\sigma(n)}, \bar{z}_{\sigma(n)}\right)~,     
\end{equation}
is the expansion of $\mf{M}^V_n\left(u_1,\dots,u_n;z_1, \bar{z}_1,\dots, z_n, \bar{z}_n \right)$
in the domain
$\left|z_{\sigma(1)}\right|>\left|z_{\sigma(2)}\right|>\dots>|z_{\sigma(n)}|>0$. Here, $\bar{z}_{\sigma(1)}, ..., \bar{z}_{\sigma(n)}$ are complex conjugates of $z_{\sigma(1)}, ... , z_{\sigma(n)}$ respectively. If a function $\mf{M}^V_n$ satisfying above properties exists, we say that the vertex operators
$$Y_V\left(u_1 , z_1, \bar{z}_1\right),\dots, Y_V\left(u_n , z_n, \bar{z}_n\right)$$ are mutually \textit{local} with respect to each other. 
\end{enumerate}
We will often denote the non-chiral VOA by $(V,Y_V,\textbf{1},\omega,\bar{\omega})$ or simply by $V$. 
\end{defn}
% be an $\R\times\R$-graded complex vector space vector space. Let 
% \begin{equation}
%     \overline{V}=\prod\limits_{(h,\bar{h})\in\R\times\R}V_{(h,\bar{h})} \, ,
% \end{equation}
% denote the algebraic completion of $V$. Let 
% \begin{equation}
%     V'=\coprod\limits_{(h,\bar{h})\in\R\times\R}V'_{(h,\bar{h})} \, ,
% \end{equation}
% be the contragradient of $V$ where $V'_{(h,\bar{h})}$ is the dual of $V_{(h,\bar{h})}$. A series $\sum f_n$ in $\overline{V}$ is said to be absolutely convergent if for every $f'\in V'$ the series $\sum |\langle f',f_n\rangle|$ is convergent. Here $\langle f',f_n\rangle=f'(f_n)\in\C$ is just the action of the linear functional on $f'$ on $f_n$.
When $u\in V_{(h,\bar{h})}$, then single-valuedness implies that the vertex operator $Y_V(u,x,\bar{x})$ has an expansion of the form 
\begin{equation}\label{eq:genveropcorrsum}
\begin{split}
    Y_V(u,x,\bar{x})&=\sum_{\substack{m,n\in\R\\(m-n)\in\Z}}u_{m,n}x^{-m-1}\bar{x}^{-n-1}~,\\&=\sum_{\substack{m,n\in\R\\(m-n)\in\Z}}x_{m,n}(u)x^{-m-h}\bar{x}^{-n-\bar{h}}~,
\end{split}    
\end{equation}
so that 
\begin{equation}
    x_{m,n}(u)=u_{m+h-1,n+\Bar{h}-1},\quad m,n\in\R~.
\end{equation}
In particular, if the vertex operator corresponding to some homogeneous vector $u\in V_{h,\bar{h}}$ is independent of $x$ ($u$ is called a \textit{chiral vector} in this case) or $\bar{x}$ ($u$ is called a \textit{anti-chiral vector} in this case) then the corresponding \textit{(anti-)chiral vertex operator} has the expansion
\begin{equation}\label{eq:chiantchiveropexp}
    \begin{split}
        &Y_V(u,x)=\sum_{n\in\Z}x_n(u)x^{-n-(h-\bar{h})},\quad u\text{ chiral}~, \\&Y_V(u,\bar{x})=\sum_{n\in\Z}\bar{x}_n(u)\bar{x}^{-n-(\bar{h}-h)},\quad u\text{ anti-chiral}~.   
    \end{split}
    \end{equation}
The coefficients in the expansion of (anti-)chiral vertex operators can thus be determined using contour integral after replacing the formal variable by a nonzero complex variable:
\begin{equation}\label{eq:modeschirantchiint}
\begin{split}
    &x_n(u)= \frac{1}{2\pi i}\oint dz~Y_V(u,z)z^{n+(h-\bar{h})-1},\quad u\text{ chiral}~,
    \\&\bar{x}_n(u)= \frac{1}{2\pi i}\oint d\bar{z}~Y_V(u,\bar{z})\bar{z}^{n+(\bar{h}-h)-1},\quad u\text{ anti-chiral}~,
\end{split}  
\end{equation}
where the contour of integration is a circle around $z=0,\bar{z}=0$ respectively.
For $u\in V_{(h,\bar{h})}$, we write 
\begin{equation}
\mathsf{wt}(u)=h,\quad\overline{\mathsf{wt}}(u)=\bar{h} \, ,    
\end{equation}
and call $\mathsf{wt}(u),\overline{\mathsf{wt}}(u)$ the \textit{conformal weights} of $u$. 
One can show that the degree of the operators $x_{m,n}(u),u_{m,n}$ is given by \cite{Singh:2023mom}
\begin{equation}
\begin{split}
    \mathsf{wt}~x_{m,n}(u)=-m,\quad &\overline{\mathsf{wt}}~x_{m,n}(u)=-n,\\\mathsf{wt}~u_{m,n}=h-m-1,\quad &\overline{\mathsf{wt}}~u_{m,n}=\bar{h}-n-1~.
\end{split}    
\end{equation}
\subsection{Modules Of A Non-Chiral VOA}
Let us now define the notion of a module of a non-chiral VOA. Let $(V,Y_V,\textbf{1},\omega,\bar{\omega})$ be a non-chiral VOA. 
\begin{defn}\label{def:modofncft}
A module for $V$ is a $(\C\times \C)$-graded complex vector space $W$ equipped with a linear map $Y_W$ 
\begin{equation}
\begin{split}
    Y_W: \, &V\longrightarrow \mathsf{End}(W)\{x,\bar{x}\}\\&u\longmapsto Y_W(u,x,\bar{x})=\sum_{\substack{m,n\in\C}}u^W_{m,n}x^{-m-1}\bar{x}^{-n-1}~,
\end{split}
\end{equation}
% or equivalently a map
% \begin{equation}
% \begin{split}
%     Y_W: \, &\C^\times \longrightarrow \text{Hom}(V\otimes W,\overline{W})\\&z\longmapsto Y_W(\cdot,z,\Bar{z}):u\otimes w\longmapsto Y_W(u,z,\bar{z})w,
% \end{split}
% \end{equation}
% which is analytic in $z$ and anti-analytic in $\Bar{z}$,
called the \textit{module vertex operator map}, satisfying the following properties:
The following properties must be satisfied:
\begin{enumerate}
\item\label{item:M_idprop} \textit{Identity property:} $Y_W(\mathbf{1}, x, \bar{x})=\mathds{1}_W$.
% The vertex operator corresponding to the vacuum vector acts as identity, i.e.
% \begin{equation}
%      w = w, \quad \forall~ w \in W
% \end{equation}
\item\label{item:M_gradres} \textit{Grading-restriction property:} Every homogeneous subspace $W_{(h,\bar{h})}$ is finite-dimensional and there exists $M\in\R$, such that $W_{(h,\bar{h})}=0$ for $\mathsf{Re}\,h<M$ or $\mathsf{Re}\,\Bar{h}<M$.
% \begin{equation}\label{eq:M_lowerM}
% ,\quad\text{ for }  \text{ or }.    
% \end{equation}
\item\label{item:M_singvalprop} \textit{Single-valuedness property:} For every homogeneous subspace $W_{(h,\bar{h})}$, $h-\bar{h}\in\mathbb{Z}.$ 
% \begin{equation}
    
% \end{equation}
\begin{comment}
\item \label{item:M_creatprop}\textit{Creation property:} For any $v\in V$ 
\begin{equation}\label{eq:M_creatprop}
    \lim_{x,\bar{x}\to 0}Y_V(v,x,\bar{x})\mathbf{1}=v.
\end{equation}
\end{comment}
\item\label{item:M_virprop} \textit{Virasoro property:} The vertex operators $Y_W(\omega,x,\bar{x})$ and $Y_W(\bar{\omega},x,\bar{x})$ have expansion of the form
    \begin{equation}
     \begin{split}
&Y_W(\omega,x,\bar{x})=\sum_{n\in\Z}L^W(n)x^{-n-2}\\
&Y_W(\bar{\omega},x,\bar{x})=\sum_{n\in\Z}\bar{L}^W(n)\bar{x}^{-n-2}    
     \end{split}   
    \end{equation}
where $L^W(n),\bar{L}^W(n)$ are operators which satisfy the Virasoro algebra \eqref{eq:viraalg} with central charge $c,\bar{c}$ respectively.
\item\label{item:M_gradprop} \textit{Grading property:} For $w\in W_{(h,\bar{h})}$, $L^{W}(0)w=hw,\quad \Bar{L}^{W}(0)w=\bar{h}w.$
% \begin{equation}
    
% \end{equation}
\item\label{item:M_l0prop} \textit{$L^W(0)$-property} : 
    \begin{equation}
    \begin{split}
        &[L^W(0), Y_W(u , x, \bar{x})]=x \frac{\partial}{\partial x}Y_W(u ,x, \bar{x})+Y_W(L(0) u , x, \bar{x})\\&[\bar{L}^W(0), Y_W(u , x, \bar{x})]=\bar{x} \frac{\partial}{\partial \bar{x}}Y_W(u ,x, \bar{x})+Y_W(\bar{L}(0) u , x, \bar{x})
    \end{split}
\end{equation}
\item\label{item:M_transprop} \textit{Translation property:} For any $u\in V$

\begin{equation}\label{eq:translation_module}
\begin{split}
& {\left[L^W(-1), Y_W(u , x, \bar{x})\right]=Y_W\left(L(-1) u , x, \bar{x}\right)=\frac{\partial}{\partial x} Y_W(u , x, \bar{x}),} \\
& {\left[\bar{L}^W(-1), Y_W(u ; x, \bar{x})\right]=Y_W\left(\bar{L}(-1) u , x, \bar{x}\right)=\frac{\partial}{\partial \bar{x}} Y_W(u , x, \bar{x})} \\
&
\end{split}
\end{equation}
\item \label{item:M_locprop}\textit{Locality and Duality property:} 
The module vertex operators must be local, that is given $n$ module vertex operators $Y_W(u_i,z_i,\Bar{z}_i),~i=1,\dots,n$, there exists an operator-valued function $\mf{M}^W_n(u_1,\dots,u_n,z_1,\dots,z_n,\Bar{z}_1,\dots,\Bar{z}_n)$ satisfying the requirements in Property \ref{item:locprop} of Definition \ref{def:nonchiralVOA}. Moreover,  
for $u_1, u_2\in V$,
\begin{equation}\label{eq:M_locality_fields}
\begin{split}
& Y_W\left(u_1 , z_1, \bar{z}_1\right) Y_W\left(u_2 , z_2, \bar{z}_2\right) , \\
&Y_W\left(u_2 , z_2, \bar{z}_2\right) Y_W\left(u_1 , z_1, \bar{z}_1\right) ,\\
& Y_W\left(Y_V\left(u_1 , z_1-z_2, \bar{z}_1-\bar{z}_2\right)u_2 , z_2, \bar{z}_2\right), 
\end{split}
\end{equation}
are the expansions of a function
$\mf{M}_2^W\left(u_1,u_2,z_1, \bar{z}_1, z_2, \bar{z}_2 \right)   $ 
in the sets given by $\left|z_1\right|>\left|z_2\right|>0$, $\left|z_2\right|>\left|z_1\right|>0$, and $|z_2|>|z_1-z_2|>0$, respectively, where $\bar{z}_1$, $\bar{z}_2$ are the complex conjugates of $z_1$ and $z_2$ respectively. Also $\mf{M}_2^W$ is an $\mathsf{End}(W)$-valued function, linear in $u_1,u_2$, defined on 
\begin{equation}
    \{(z_1,z_2) \in \C^2 \, \lvert \, z_1, z_2 \neq 0, z_1 \neq z_2\} \, ,
\end{equation}
multi-valued and analytic
when $\bar{z}_1,\bar{z}_2$ are viewed as independent variables and is single-valued when $\bar{z}_1, \bar{z}_2$ are equal to the complex conjugates of $z_1,z_2$ respectively. We say that the module vertex operators
$Y_W\left(u_1 , z_1, \bar{z}_1\right)$ and $ Y_W\left(u_2 , z_2, \bar{z}_2\right)$ satisfy locality and duality with respect to each other if they satisfy \eqref{eq:M_locality_fields}.

\end{enumerate}
We will denote a module by $(W,Y_W)$ or simply by $W$. 
\end{defn}
As before, one can show that the module vertex operator $Y_W(u,x,\bar{x})$ for $u\in V_{(h,\bar{h})}$ can be expanded as a formal power series of the form
\begin{equation}\label{eq:mode_expansion_module}
\begin{split}    Y_W(u,x,\bar{x})&=\sum_{\substack{m,n\in\C\\(m-n)\in\Z}}u^W_{m,n}x^{-m-1}\bar{x}^{-n-1}\\&=\sum_{\substack{m,n\in\C\\(m-n)\in\Z}}x^W_{m,n}(u)x^{-m-h}\bar{x}^{-n-\bar{h}}\in\mathsf{End}(W)\{x,\bar{x}\}.
\end{split}
\end{equation} 
Similarly, the module vertex operators for (anti-)chiral vectors have the form 
\begin{equation}\label{eq:chiantchiveropexp-mod}
    \begin{split}
        &Y_W(u,x)=\sum_{n\in\Z}x^W_n(u)x^{-n-(h-\bar{h})},\quad u\text{ chiral}~, \\&Y_W(u,\bar{x})=\sum_{n\in\Z}\bar{x}^W_n(u)\bar{x}^{-n-(\bar{h}-h)},\quad u\text{ anti-chiral}~,   
    \end{split}
    \end{equation}
    and the coefficients can be written in terms of appropriate contour integrals. 
Clearly $V$ is a module for itself. Let $(W,Y_W)$ be a module of a non-chiral VOA $V$. A $V$-\textit{submodule} of $W$ is a vector subspace $W_1\subset W$ such that the vertex operator map restricts to a map on $W_1$:
\begin{equation}
\begin{split}
    Y_W: \, &V\otimes W_1\longrightarrow W_1\{x,\bar{x}\}\\&u\otimes w\longmapsto Y_W(u,x,\bar{x})w
\end{split}
\end{equation}
and is a $V$-module in its own right. A $V$-module is called \textit{irreducible} if it has no non-zero proper submodules. Irreducible modules are also called \textit{simple} modules. Direct sum of two $V$-modules is another $V$-module with the obvious definition of vertex operator map.
A homomorphism between two $V$-modules $(W_1,Y_{W_1})$ and $(W_2,Y_{W_2})$ is a grading preserving linear map $f:W_1\longrightarrow W_2$ satisfying 
\begin{equation}
    f(Y_{W_1}(v,x,\bar{x})w)=Y_{W_2}(v,x,\bar{x})f(w),\quad\forall~ v\in V,w\in W_1.
\end{equation}
The notion of isomorphisms and automorphisms are defined analogous to the non-chiral VOA. 
A \textit{semi-simple} $V$-module is a $V$-module isomorphic to the direct sum of finitely many simple $V$-modules.  The following proposition is a simplification of \cite[Proposition 11.9]{dong:1993}.
\begin{prop}\label{prop:veropnonzero}
Let $(W,Y_W)$ be an irreducible module of a non-chiral VOA $(V,Y_V)$. Then for any non-zero vectors $v\in V$ and $w\in W$, 
\begin{equation}
    Y_W(v,x,\bar{x})w\neq 0.
\end{equation}
Or equivalently there exists $m,n\in\Z$ such that $x^W_{m,n}(v)\cdot w\neq 0$.
\begin{proof}
Since $W$ is irreducible, we have 
\begin{equation}
\begin{split}
    W=\mathsf{Span}_\C\left\{x^W_{m_1,n_1}(v_1)x^W_{m_2,n_2}(v_2)\cdots x^W_{m_k,n_k}\right.&(v_k)\cdot w : v_i\in V_{(h_i,\bar{h}_i)},h_i,\bar{h}_i\in\R,\\&\left.n_i,m_i\in\Z,i=1,\dots,k; k\in\N_0\right\}.
    \end{split}
\end{equation}
If not, the RHS will define an invariant subspace of $W$ hence contradicting the irreducibility of $W$. Suppose now that $Y_W(v,x,\bar{x})w=0$. Then 
by the locality property \ref{item:M_locprop} we see that 
\begin{equation}
    Y_W(v,z_1,\bar{z}_1)Y_W(u,z_2,\bar{z}_2)w=0 \, ,
\end{equation}
for an arbitrary $u\in V$. But this implies that $Y_W(v,x,\bar{x})=0$. Moreover, since $v\neq 0$ by duality it implies that $Y_W(\cdot,x,\bar{x})\equiv 0$ which is a contradiction.  
\end{proof}
\end{prop}
\subsection{Intertwining Operators}
\begin{defn}\label{defn:intertwiners}
Let $(V,Y_V,\omega,\Bar{\omega},\mathbf{1})$ be a non-chiral vertex operator algebra and let $\left(W_i, Y_i\right)$, $\left(W_j, Y_j\right)$ and $\left(W_k, Y_k\right)$ be three $V$-modules. An \textit{intertwining operator} of type ${i\choose j~k}$ is a linear map
\begin{equation}
\begin{split}
\mathcal{Y}:W_j&\longrightarrow \mathsf{Hom}(W_k,W_i)\{x,\Bar{x}\}\\
 w_{(j)}&\longmapsto \mathcal{Y}(w_{(j)}, x,\bar{x})=\sum_{n,m\in\R}(w_{(j)})_{n,m}x^{-n-1}\bar{x}^{-m-1}~,
\end{split}
\end{equation}
satisfying the following properties:
\begin{enumerate}
\item\label{item:intL0} $L(0)$\textit{-property:} For any $w_{(j)}\in W_j$ 
\begin{equation}
\begin{split}
        &[L(0), \mathcal{Y}(w_{(j)} , x, \bar{x})]=x \frac{\partial}{\partial x}\mathcal{Y}(w_{(j)} ,x, \bar{x})+
        \mathcal{Y}
        (L^{W_j}(0) w_{(j)} , x, \bar{x})\\&[\bar{L}(0), \mathcal{Y}(w_{(j)} , x, \bar{x})]=\bar{x} \frac{\partial}{\partial \bar{x}}\mathcal{Y}(w_{(j)} ,x, \bar{x})+\mathcal{Y}(\bar{L}^{W_j}(0) w_{(j)} , x, \bar{x})
    \end{split}    
\end{equation}
where the commutator on the LHS is understood to be 
\begin{equation}
\begin{split}
&[L(0), \mathcal{Y}(w_{(j)} , x, \bar{x})]=L^{W_i}(0)\mathcal{Y}(w_{(j)} , x, \bar{x})-\mathcal{Y}(w_{(j)} , x, \bar{x})L^{W_k}(0) \\&[\bar{L}(0), \mathcal{Y}(w_{(j)} , x, \bar{x})]=\bar{L}^{W_i}(0)\mathcal{Y}(w_{(j)} , x, \bar{x})-\mathcal{Y}(w_{(j)} , x, \bar{x})\bar{L}^{W_k}(0).
\end{split}
\end{equation}
\item\label{item:intL-1} \textit{Translation property:} For any $w_{(j)}\in W_j$
\begin{equation}
\begin{split}
& {[L(-1), \mathcal{Y}(w_{(j)} , x, \bar{x})]=\mathcal{Y}\left(L^{W_j}(-1) w_{(j)} , x, \bar{x}\right)=\frac{\partial}{\partial x} \mathcal{Y}(w_{(j)} , x, \bar{x}),} \\
& {[\bar{L}(-1), \mathcal{Y}(w_{(j)} , x, \bar{x})]=\mathcal{Y}\left(\bar{L}^{W_j}(-1) w_{(j)} , x, \bar{x}\right)=\frac{\partial}{\partial \bar{x}} \mathcal{Y}(w_{(j)} , x, \bar{x})}
\end{split}
\end{equation}
where the commutativity is understood as above.
\item\label{item:intloc} \textit{Locality and Duality property:} The module vertex operators and the intertwiner must be local in the following sense: given vectors $u_1,\dots,u_{n-1}\in V, w_{(j)}\in W_j$, there exists an operator-valued function $\mf{M}_n(u_1,\dots,u_{n-1},w_{(j)},z_1,\dots,z_n,\Bar{z}_1,\dots,\Bar{z}_n)$ satisfying the requirements in Property \ref{item:locprop} of Definition \ref{def:nonchiralVOA}. Here, the product of vertex operators in \eqref{eq:veropprodvoa} is replaced by
\begin{equation}
\begin{split}
Y_i\left(u_{\sigma(1)} , z_{\sigma(1)}, \bar{z}_{\sigma(1)}\right)\cdots Y_i\left(u_{\sigma(a-1)} , z_{\sigma(a-1)}, \bar{z}_{\sigma(a-1)}\right)\mathcal{Y}\left(w_{(j)} , z_{n}, \bar{z}_{n}\right)\\Y_k\left(u_{\sigma(a)} , z_{\sigma(a)}, \bar{z}_{\sigma(a)}\right)\cdots  Y_k\left(u_{\sigma(n-1)} , z_{\sigma(n-1)}, \bar{z}_{\sigma(n-1)}\right)~. 
\end{split}
\end{equation}
Moreover,  
for $w_{(j)}\in W_j, u\in V$, 
\begin{equation}\label{eq:M_locality_fields_int}
\begin{split}
& Y_i\left(u, z_1, \bar{z}_1\right) \mathcal{Y}\left(w_{(j)}, z_2, \bar{z}_2\right) , \\
&\mathcal{Y}\left(w_{(j)}, z_2, \bar{z}_2\right)Y_k\left(u , z_1, \bar{z}_1\right) ,\\
& \mathcal{Y}\left(Y_j\left(u , z_1-z_2, \bar{z}_1-\bar{z}_2\right)w_{(j)} , z_2, \bar{z}_2\right), 
\end{split}
\end{equation}
are the expansions of a function
$\mf{M}_2\left(u,w_{(j)},z_1, \bar{z}_1, z_2, \bar{z}_2 \right)   $ 
in the sets given by $\left|z_1\right|>\left|z_2\right|>0$, $\left|z_2\right|>\left|z_1\right|>0$, and $|z_2|>|z_1-z_2|>0$, respectively, where $\bar{z}_1$, $\bar{z}_2$ are the complex conjugates of $z_1$ and $z_2$ respectively. Also $\mf{M}_2$ is  linear in $w_{(j)},u$, defined on 
\begin{equation}
    \{(z_1,z_2) \in \C^2 \, \lvert \, z_1, z_2 \neq 0, z_1 \neq z_2\} \, ,
\end{equation}
multi-valued and analytic
when $\bar{z}_1,\bar{z}_2$ are viewed as independent variables and is single-valued when $\bar{z}_1, \bar{z}_2$ are equal to the complex conjugates of $z_1,z_2$ respectively.
\end{enumerate}
\end{defn}   
Clearly the space $\mathcal{V}_{jk}^i$ of intertwiners of type ${i\choose j~k}$ forms a vector space. We define the fusion rules to be 
\begin{equation}
    \mathcal{N}_{jk}^i:=\mathsf{dim}(\mathcal{V}_{jk}^i)~.
\end{equation}
\section{Lorentzian Lattice Vertex Operator Algebra (LLVOA) And Its Modules}\label{sec:llvoa}
In this section, we will review the construction of the Lorentzian lattice vertex operator algebra (LLVOA) based on an even, integral Lorentzian lattice
$\Lambda \subset \R^{m,n}$ of signature $(m,n)$. 
\par Let $\Lambda_1,\Lambda_2,\Lambda',\Lambda_1^0,\Lambda_2^0,\Lambda_0$ and $\Lambda_0^\circ$ be as in Section \ref{sec:intro1.1}. 
We denote by $\C [\Lambda']$ the group algebra of the abelian group $\Lambda'$ and denote the element $\lambda\in\Lambda'$ embedded in $\C[\Lambda']$ by ${\rm e}^{\lambda}$. The multiplication in $\C[\Lambda']$ is defined by 
\begin{equation}
    {\rm e}^{\lambda_1}\cdot {\rm e}^{\lambda_2}={\rm e}^{\lambda_1+\lambda_2}~.
\end{equation}
Define the vector space 
\begin{equation}\label{eq:defmfhi}
    \mf{h}_i:=\Lambda_i\otimes_{\Z}\C,\quad i=1,2,\quad \mf{h}=\mf{h}_1\oplus\mf{h}_2~, 
\end{equation}
and $\C$-linearly extend the bilinear form on $\Lambda_i$ to  $\mf{h}_i$. Here $\Lambda_i$ is as defined in \eqref{eq:lamb12fromlamb}. Note that 
\begin{equation}
    \mathsf{dim}(\mf{h}_1)=m,\quad \mathsf{dim}(\mf{h}_2)=n~.
\end{equation}
\begin{comment}
and 
\begin{equation}
\mf{h}:=\mf{h}_1\oplus\mf{h}_2.    
\end{equation}
Note that 
\begin{equation}
\begin{split}
    \mf{h}_1=\text{Span}_{\C}\{\alpha_1,\dots,\alpha_{d}\} \\
    \mf{h}_2=\text{Span}_{\C}\{\beta_1,\dots,\beta_{d}\}.
\end{split}
\end{equation}
\end{comment}
We define the Lie algebra 
\begin{equation}\label{eq:defhatmfhi}
\begin{split}
\hat{\mf{h}}&:=\left( \bigoplus_{\substack{r, s\in\Z}}(\mf{h}_1\otimes t^r)\oplus(\mf{h}_2\otimes \bar{t}^{\,s})\right)\oplus(\C \textbf{k} \oplus\C\bar{\textbf{k}}).
\end{split}      
\end{equation}
%\fixme{Should we change $m,n$ here to something else, it is confusing since $\Lambda\subset\R^{m,n}$, or is it clear from the context?}
Introduce the notation 
\begin{equation}
    \alpha(r):=\alpha\otimes t^r,\quad \beta(s):=\beta\otimes\bar{t}^{\,s},\quad \alpha\in\mf{h}_1, \beta\in\mf{h}_2.
\end{equation}
The non-zero Lie bracket on $\hat{\mf{h}}$ is below 
\begin{equation}\label{liebracketmodes}
\begin{split}
[\, \alpha(r_1), \alpha'(r_2) \, ] &=  r_1 \, \left\langle \alpha, \alpha' \right\rangle  \, \delta_{r_1+r_2,0} \, \textbf{k} 
\\
[ \, \beta(s_1), \beta'(s_2) \, ]  &=  s_1 \, \left\langle \beta, \beta' \right\rangle  \, \delta_{s_1+s_2,0} \, \bar{\textbf{k}}  \  .    
\end{split}
\end{equation}
\begin{comment}
Further, for the basis $\{\lambda_i\}_{i=1}^d$ of $\Lambda$, we define
\begin{equation}
    \alpha_{i}(m) := \alpha^{\lambda_i}(m) , \quad \beta_{i}(m) := \beta^{\lambda_i}(m)
\end{equation}
Then a general element $\alpha^{\lambda}(m)$ of $\mf{h}_{1}\otimes t^m$ can be written as  
\begin{equation}
    \alpha^{\lambda}(m) = \sum_{i=1}^d c_i \, \alpha_{i}(m), \quad \text{where} \quad \lambda = \sum_{i=1}^d c_i \,  \lambda_{i}.
\end{equation}
\end{comment}
Note that 
\begin{equation}
\hat{\mf{h}}=\hat{\mf{h}}_1^\star\oplus\hat{\mf{h}}_2^\star\oplus\hat{\mf{h}}_1^0\oplus\hat{\mf{h}}_2^0     \, , 
\end{equation}
where $\hat{\mf{h}}_1^\star,\hat{\mf{h}}_2^\star$ are the standard Heisenberg algebras associated to the abelian Lie algebras $\mf{h}_1,\mf{h}_2$ respectively \cite[Chapter 1]{Frenkel:1988xz} and 
\begin{equation}
\hat{\mf{h}}_1^0:=\mf{h}_1\otimes t^0\cong \mf{h}_1,\quad \hat{\mf{h}}_2^0:=\mf{h}_2\otimes \bar{t}^0\cong \mf{h}_2.     
\end{equation}
Define 
\begin{equation}
\begin{split} \hat{\mf{h}}^{-}:=\left(\bigoplus_{r,s < 0}(\mf{h}_1\otimes t^r)\oplus(\mf{h}_2\otimes\bar{t}^{\,s})\right),& \quad \hat{\mf{h}}^{0}:=(\mf{h}_1\otimes t^0)\oplus (\mf{h}_2\otimes \bar{t}^0) \oplus\C \textbf{k}\oplus\C\bar{\textbf{k}} \, , \\
\hat{\mf{h}}^{+}:=\Bigg{(} \bigoplus_{r,s > 0}(\mf{h}_1  &  \otimes t^r)    \oplus(\mf{h}_2\otimes\bar{t}^{\,s}) \Bigg{)}.
\end{split}
\end{equation}
% Note that 
% \begin{equation}
%     \hat{\mf{h}}=\hat{\mf{h}}^{-}\oplus\hat{\mf{h}}^{0}\oplus\hat{\mf{h}}^{+}
% \end{equation} 
Define the space (see Remark \ref{rem:modsinglevaluedness})
\begin{equation}\label{FockSpace}
\begin{split}
    V_{\Lambda'}&:=   S\left(\hat{\mf{h}}^{-}\right) \otimes \C\left[\Lambda'\right]~,
\end{split}
\end{equation}
where $S(\hat{\mf{h}}^{-})$ is the symmetric algebra for $\hat{\mf{h}}^{-}$. 
%and $\C[\Lambda_0^\circ]$ is identified with the subspace\footnote{Note that $\C[\Lambda_0^\circ]$ is not a \textit{subalgebra} of $\C[\Lambda']$ in general.} of $\C[\Lambda']$. 
The space $V_{\Lambda'}$ is generated by elements of the form 
\begin{equation}\label{eq:genvect}
\begin{split}
  & \left( \alpha_{1}(-m_1)\cdot\alpha_{2}(-m_2)\cdots\alpha_{k}(-m_k) \cdot \beta_{1}(-\bar{m}_1)\cdot\beta_{2}(-\bar{m}_2)\cdots\beta_{\bar{k}}(-\bar{m}_{\bar{k}}) \right)  \otimes {\rm e}^{\lambda} 
\end{split}
\end{equation}
for $m_i, \, \bar{m}_i> 0,\, k,\bar{k}\geq 0,\lambda=(\alpha^\lambda,\beta^\lambda)\in\Lambda'$, $\alpha_i \in \mf{h}_1$, and $\beta_i \in \mf{h}_2$.
The subalgebra $\hat{\mf{h}}^{-}$ has a natural action on $S(\hat{\mf{h}}^{-})$ while $\hat{\mf{h}}^{0}$ acts on $\C [\Lambda']$ as
\begin{equation}\label{eq:alp0act}
\begin{split}
    \alpha' (0) \, {\rm e}^\lambda   &=  \langle \alpha' , \alpha^\lambda \rangle   \, {\rm e}^{\lambda} \\ 
    \beta'(0) \, {\rm e}^\lambda  &=   \langle \beta ' , \beta^\lambda \rangle  \, {\rm e}^{\lambda}
\end{split}
\end{equation}
where $\alpha'(0) \in \hat{\mf{h}}_1^0,~\beta'(0) \in \hat{\mf{h}}_2^0$. The central elements $\textbf{k}$ and $\bar{\textbf{k}}$ act on $\C [\Lambda'] $ as identity. 
We also let $\hat{\mf{h}}^{+}$ act on $\C [\Lambda']$ by  0. We can extend the action of these subspaces of $\hat{\mf{h}}$ to $V_{\Lambda'}$ by using the Lie bracket given in \eqref{liebracketmodes}.
Similarly, we can define\footnote{Note that $\C[\Lambda_0^\circ]$ is not a \textit{subalgebra} of $\C[\Lambda']$ in general.} 
\begin{equation}
    V_0^\circ:=S\left(\hat{\mathfrak{h}}^{-}\right) \otimes \mathbb{C}\left[\Lambda_0^\circ\right] , \quad 
    V_\Lambda:=S\left(\hat{\mathfrak{h}}^{-}\right) \otimes \mathbb{C}\left[\Lambda\right] \, . 
\end{equation}
Similarly $V^\circ_0$ and $V_\Lambda$ are $\mathfrak{\hat h}$- modules.

We focus on $V_0^\circ$ for the reason of single-valuedness. 
Using the equivalence relation on $\Lambda_0^\circ$ given by 
\begin{equation}
    \lambda\sim\lambda'\iff \lambda-\lambda'\in\Lambda_0,\quad \lambda,\lambda'\in \Lambda_0^\circ~,
\end{equation}
we can decompose the space $V_0^\circ$ as sum over equivalence classes $\Lambda^\circ_0/\Lambda_0$:
\begin{equation}
    V_0^\circ=\bigoplus_{[\mu]\in\Lambda_0^\circ/\Lambda_0}V(\mu)~,
\end{equation}
where 
\begin{equation}
    \label{eq:equivalent class}
    V(\mu):=S(\hat{\mf{h}}^{-}) \otimes \C[\mu+\Lambda_0]\subset V_0^\circ~.
\end{equation}
 This makes $V_{\Lambda'}$  into  an $\hat{\mf{h}}$-module. Similarly $V^\circ_0$ and $V_\Lambda$ are $\mathfrak{\hat h}$- modules.

\bigskip 

% We define a $\mathbb{Z}$-bilinear map $\epsilon: \Lambda' \times \Lambda' \rightarrow \mathbb{Z}$, which acts on the basis  as 

% \begin{equation}\label{epsilonaction}
% \epsilon\left(\lambda_i, \lambda_j\right)= \begin{cases} \lambda_i \circ \lambda_j  & i>j \\ 0 & i \leq j,\end{cases}  \  \    \\\end{equation}
% where $\{\lambda_i\}_{i=1}^{m+n}$ is an integral basis of $\Lambda$.
% The action of $\epsilon$ on general vectors is defined by the $\Z$-bilinearity of $\epsilon$.
% Consider  $\hat{\Lambda}=\Z_2\times\Lambda$, with the multiplication on it given by 
% \begin{equation}\label{cocyclelattice}
% \begin{split}
%     (\theta, \lambda) \cdot (\tau, \lambda') = & \left(  \theta \tau  \,  (-1)^{ \epsilon(\lambda, \lambda')}, \lambda + \lambda' \right).  
% \end{split}
% \end{equation}
Consider the central extension of the abelian group $\Lambda'$:
\begin{equation}\label{eq:centextoflamb}
    0\longrightarrow\langle\omega_p\rangle\longrightarrow\hat{\Lambda}'\longrightarrow\Lambda'\longrightarrow 0 \ ,
\end{equation}
where $\langle\omega_p\rangle\cong\Z/p\Z$ is the group of $p^{{\rm th}}$ roots of unity for some $p\in2\Z_{\geq 0}$. We assume that the commutator map $c:\Lambda'\times \Lambda'\to \langle\omega_p\rangle$ when restricted to $\Lambda$ is given by 
\begin{equation}
    c(\lambda_1,\lambda_2)=(-1)^{\lambda_1\circ\lambda_2},\quad \lambda_1,\lambda_2\in\Lambda~.
\end{equation}
%  We denote elements $(1,\lambda),(\theta,0)\in \hat{\Lambda}$ by ${\rm e}_{\lambda}=(1,\lambda)$ and $\theta=(\theta,0)$ respectively. Then it is easy to check that  
% \begin{equation}\label{eq:elambdatheta}
%     (\theta,\lambda)=\theta \, {\rm e}_{\lambda}={\rm e}_{\lambda}\theta,
% \end{equation}
% and 
% \begin{equation}\label{eq:elamemu}
%     {\rm e}_{\lambda}{\rm e}_{\mu}=(-1)^{\epsilon(\lambda,\mu)} \, {\rm e}_{\lambda+\mu} \ .
% \end{equation}
% Using the above relation, it can be shown that (see Lemma \ref{lemma:commfunccenext} for proof) 
% \begin{equation}\label{eq:commutelm}{\rm e}_{\lambda}{\rm e}_{\mu}=(-1)^{\lambda\circ\mu} \, {\rm e}_{\mu} \, {\rm e}_{\lambda} \,.
% \end{equation}
% This property requires that the lattice be even and integral. Note that we chose an integral basis of $\Lambda$ to define the central extension. In Appendix \ref{app:centext} we show that a cocycle $\tilde{\epsilon}$ defined analogous to \eqref{epsilonaction} for a different choice of basis is cohomologous to $\epsilon$, and hence gives rise to an isomorphic central extension.  
Choose a section $\lambda\mapsto {\rm e}_\lambda$ of the extension. Let $\epsilon:\Lambda'\times\Lambda'\longrightarrow \Q$ be the cocycle for the extension:
\begin{equation}\label{eq:elambmu}
    {\rm e}_\lambda{\rm e}_\mu=c(\lambda,\mu){\rm e}_\mu{\rm e}_\lambda=(-1)^{\epsilon(\lambda,\mu)}{\rm e}_{\lambda+\mu}~.
\end{equation}
Then we define the action of ${\rm e}_\lambda$ on $\C[\Lambda']$ as follows
% \begin{equation}\label{elambdaaction}
%     (\theta,\lambda') \, {\rm e}^{\lambda}=\theta(-1)^{\epsilon(\lambda',\lambda)} \, {\rm e}^{\lambda + \lambda'}  \, ,
% \end{equation}
% in particular for $ (\theta,\lambda') = (1, \lambda' ) = {\rm e}_{\lambda'}$, we have 
\begin{equation}\label{eq:eloweraction}
  \omega_p\cdot {\rm e}^\lambda= \omega_p{\rm e}^\lambda,\quad  {\rm e}_{\lambda'} \cdot {\rm e}^{\lambda}= (-1)^{\epsilon(\lambda',\lambda)} \, {\rm e}^{\lambda + \lambda'}  \, ,
\end{equation}
where $\epsilon$ is the 2-cocycle corresponding to the central extension $\hat{\Lambda}'$. This makes $V_{\Lambda'}$ into a $\hat{\Lambda}'$-module where $\hat{\Lambda}'$ acts only on $\C[\Lambda']$. In a similar mannar, $V^\circ_0$ and $V(\mu)$ are $\hat \Lambda_0$-module, and 
$V_\Lambda$ is  $\hat \Lambda_0$-module as well as $\hat\Lambda$-module, 
where $\hat\Lambda_0$ and
$\hat\Lambda$ are defined 
as follows.
\par Since $\Lambda$ is integral, restricting the cocycle to $\Lambda$ and $\Lambda_0$ gives a central extension of $\Lambda$ and $\Lambda_0$ by $\Z_2=\{\pm 1\}$:
\begin{equation}
\begin{split}
     0\longrightarrow\Z_2\longrightarrow\hat{\Lambda}\longrightarrow\Lambda\longrightarrow 0 \ ,    \\ 0\longrightarrow\Z_2\longrightarrow\hat{\Lambda}_0\longrightarrow\Lambda_0\longrightarrow 0 \ .    
\end{split}
\end{equation}
We can choose an isomorphism 
\begin{equation}\label{eq:thetalambhat}
    \hat{\Lambda}=\{(\theta,\lambda):\theta\in\Z_2,\lambda\in\Lambda\},\quad \hat{\Lambda}_0=\{(\theta,\lambda):\theta\in\Z_2,\lambda\in\Lambda_0\}~.
\end{equation}
We can identify the section 
\begin{equation}
    {\rm e}_\lambda=(1,\lambda)\in\hat{\Lambda},\quad \lambda\in\Lambda~,
\end{equation}
and the cocycle restricted to $\Lambda$ can be taken to be 
\begin{equation}\label{epsilonaction}
\epsilon\left(\lambda_i, \lambda_j\right)= \begin{cases} \lambda_i \circ \lambda_j  & i>j \, ,  \\ 0 & i \leq j \, ,\end{cases}  \  \    \\\end{equation}
where $\{\lambda_i\}_{i=1}^{m+n}$ is an integral basis of $\Lambda$ and $\epsilon$ is bilinearly extended to $\Lambda$.
% Note that the same cocycle $\epsilon$ restricted to $\Lambda_0$
% \begin{equation}
%     \epsilon:\Lambda_0\times\Lambda_0\longrightarrow\Z \, , 
% \end{equation}
% defines a central extension $\hat{\Lambda}_0:=\Z_2\times\Lambda_0\subset\hat{\Lambda}$:
% \begin{equation}
% 0\longrightarrow\Z_2\longrightarrow\hat{\Lambda}_0\longrightarrow\Lambda_0\longrightarrow 0 \ .
% \end{equation}
% Moreover, the action \eqref{elambdaaction} restricted to $\hat{\Lambda}_0$ makes $\C[\Lambda_0]$ into a $\hat{\Lambda}_0$-module. 
Let $x,\bar{x}$ be formal variables.
For any vector $\lambda = (\alpha^\lambda, \beta^\lambda)$, define the operators $x^{\alpha^\lambda} , \bar{x}^{\beta^\lambda} $ by the following actions 
\begin{equation}\label{eq:xpoweraction}
\begin{split}
    x^{\alpha^\lambda} (u \otimes {\rm e}^{\lambda ' })  = x^{\langle \alpha^\lambda, \, \alpha^{\lambda'} \rangle} (u \otimes {\rm e}^{\lambda'}) \, ,  \\ 
    \bar{x}^{\beta^\lambda} (u \otimes {\rm e}^{\lambda'})  =  \bar{x}^{\langle \beta^\lambda, \, \beta^{\lambda'} \rangle}  (u \otimes {\rm e}^{\lambda'}) \, ,  
\end{split} 
\end{equation}
where $u\in S(\hat{\mf{h}}^-),\lambda'\in\Lambda' $. 
\begin{comment}
\begin{remark}\label{rem:exttofullL}
We can define the central extension $\hat{\Lambda}$ of $\Lambda$ by exactly the same construction as above. The action of $\hat{\Lambda}$ on $\C[\Lambda]$ can also be defined in exactly the same way.      
\end{remark}
\end{comment}
For a general vector $v$ of the form \eqref{eq:genvect} with $\lambda\in\Lambda_0$, the vertex operator is defined as 
\begin{equation}\label{eq:genvertopdef}
Y(v,x,\bar{x})=\typecolon \prod_{r=1}^{k}\prod_{s=1}^{\bar{k}}\left(\frac{1}{(m_r-1)!}\frac{d^{m_r - 1 }\alpha_{r}(x)}{dx^{m_r-1}}\right)\left(\frac{1}{(\bar{m}_s-1)!}\frac{d^{\bar{m}_s-1}\beta_{s}(\bar{x})}{d\bar{x}^{\bar{m}_s-1}}\right) Y({\rm e}^\lambda,x,\bar{x})\typecolon~.
\end{equation}
Various ingredients in the definition are summarized below:
\begin{enumerate}
\item The notation $\alpha(x)$ and $\beta(\bar{x})$ is defined as
\begin{equation}
    \alpha(x) =  \underbrace{ \sum_{r > 0} \alpha(r) x ^{-r-1}}_{:=\alpha(x)^{+}} + \underbrace{\sum_{r < 0}\alpha(r) x ^{-r-1}}_{:=\alpha(x)^{-}}   + \,    \alpha(0) x^{-1}  \ , 
\end{equation}
Similarly, we can also define $\beta(\bar{x})$. 
\item Formal differentiation and integration is defined as
\begin{align}
\frac{dx^r}{dx}&=rx^{r-1},& \frac{d\bar{x}^r}{dx}&=r\bar{x}^{r-1}~,& r\in\R~,\\
\int dx~x^r&=\frac{x^{r+1}}{r+1},&\int d\bar{x}~\bar{x}^r&=\frac{\bar{x}^{r+1}}{r+1}~,& r\neq -1.
\end{align}
\item For $\lambda=(\alpha^\lambda,\beta^\lambda) \in \Lambda_0$, the vertex operators $ Y({\rm e}^{\lambda},x,\bar{x})$ are defined as  
\begin{equation}\label{eq:veropelamb}
\begin{split}
    Y({\rm e}^{\lambda},x,\Bar{x})&=\exp\left(\int dx~\alpha^{\lambda}(x)^-\right)\exp\left(\int dx~\alpha^{\lambda}(x)^+\right)  \\ &\times \exp\left(\int d\bar{x}~\beta^{\lambda}(\Bar{x})^-\right)\exp\left(\int d\bar{x}~\beta^{\lambda}(\Bar{x})^+\right){\rm e}_{\lambda} \, x^{\alpha^{\lambda}}\Bar{x}^{\beta^{\lambda}}.
    \end{split}
\end{equation}
\item Normal ordering $\typecolon\typecolon$ is defined as 
\begin{equation}\label{eq:normord}
\begin{split}
&\typecolon \alpha^\lambda(p) \,  \alpha^{\lambda'}(q)\typecolon = \typecolon \alpha^{\lambda'}(q) \,  \alpha^\lambda(p)\typecolon =\begin{cases}
\alpha^\lambda(p) \, \alpha^{\lambda'}(q)& p\leq q\\\alpha^{\lambda'}(q) \,  \alpha^\lambda(p)& p\geq q,
\end{cases}
\\&\typecolon \alpha^\lambda(p)\mathrm{e}_{\lambda'}\typecolon =\typecolon \mathrm{e}_{\lambda'} \,  \alpha^\lambda(p)\typecolon = \mathrm{e}_{\lambda'} \,  \alpha^\lambda(p),\\&\typecolon x^{\alpha^{\lambda}} \, {\rm e}_{\lambda'} \typecolon = \typecolon {\rm e}_{\lambda'} \, x^{\alpha^{\lambda}}\typecolon = {\rm e}_{\lambda'}  \, x^{\alpha^{\lambda}}
\end{split}
\end{equation} 
and similarly for $\beta^{\lambda}$ and $\bar{x}^{\beta^\lambda}$. 
\end{enumerate}
% \begin{remark}\label{rem:genvertop}
% Using the central extension \eqref{eq:centextoflamb},  \eqref{eq:genvertopdef} can be used to define vertex operators even if ${\rm e}^\lambda\in\C[\Lambda]$. These vertex operators will act on vectors of the form \eqref{eq:genvect} with ${\rm e}^\lambda\in\C[\Lambda]$ rather than $\C[\Lambda_0]$. This will be crucial when we construct module vertex operators and intertwining operators on the modules of $V_{\Lambda}$.    
% \end{remark}
Define the vacuum vector by $\mathbf{1}=\mathrm{e}^0$ and the conformal vector is given by
\begin{equation}\label{eq:confvectLambda}
    \omega :=   \frac{1}{2}\sum_{i = 1}^{m} \left( u_{i}(-1)^2 \right)\otimes \mathbf{1} \quad \overline{\omega}:= \frac{1}{2}\sum_{i = 1}^{n}  \left( v_{i}(-1)^2\right)\otimes \mathbf{1}~,
\end{equation}
where $\{u_i\}_{i=1}^m\subset\mf{h}_1, \{v_i\}_{i=1}^n\subset\mf{h}_2$ are orthonormal basis of $\mf{h}_1$ and $\mf{h}_2$ respectively:
\begin{equation}
    \langle u_i,u_j\rangle=\delta_{i,j},\quad  \langle v_i,v_j\rangle=\delta_{i,j}~.
\end{equation}
One can check that the conformal vertex operator is given by 
\begin{equation}
\begin{split}
Y(\omega,x)=\sum_{n\in\Z}L(n)x^{-n-2},\quad Y(\overline{\omega},\bar{x})=\sum_{n\in\Z}\Bar{L}(n)\bar{x}^{-n-2}~,
\end{split}
\end{equation}
where the Virasoro generators are given by
\begin{equation}\label{eq:Virgenlamb}
\begin{split}
    L(p)&=\frac{1}{2}\sum_{i=1}^{m}\sum_{k\in\Z}\typecolon u_i(k)u_i(p-k)\typecolon 
    \\
    \bar{L}(p)&=\frac{1}{2}\sum_{i=1}^{n}\sum_{k\in\Z}\typecolon v_i(k)v_i(p-k)\typecolon~.
\end{split}
\end{equation}
For $v\in V_0^\circ$ of the form \eqref{eq:genvect}, direct calculation gives 
\begin{equation}\label{eq:M_gradingformula}
L(0)v=\left[\sum_{j=1}^km_j+\frac{\langle\alpha^\lambda,\alpha^\lambda\rangle}{2}\right]v,\quad \bar{L}(0)v=\left[\sum_{\bar{j}=1}^{\bar{k}}\bar{m}_{\bar{j}}+\frac{\langle\beta^\lambda,\beta^\lambda\rangle}{2}\right]v \, .     
\end{equation}
Therefore, the modules we defined are single-valued.

Following the calculations in \cite{Singh:2023mom}, we can prove the following theorem.
\begin{thm}\label{thm:LLVOAmodconst}
The tuple $(V(0),Y,\omega,\overline{\omega},\mathbf{1})$ is a non-chiral vertex operator algebra with central charge $(c,\ov{c})=(m,n)$. Moreover $(V(\mu),Y)$ are irreducible modules of the non-chiral VOA $(V(0),Y,\omega,\overline{\omega},\mathbf{1})$. We will call $(V(0),Y,\omega,\overline{\omega},\mathbf{1})$ the LLVOA based on $\Lambda$.     
\end{thm}
\begin{remark}\label{rem:modsinglevaluedness}
As we have noted above, the subset $\Lambda_0^\circ$ satisfies 
\begin{equation}
    \Lambda_0+\Lambda_0^\circ=\Lambda_0^\circ~.
\end{equation}
This makes sure that $V^\circ_0$ and therefore $V(\mu)$ are $\hat \Lambda_0$-module
% the modules $V(\mu)$ satisfy single-valuedness property.This is the reason for restricting to the subspace $\C[\Lambda_0^\circ]\subset\C[\Lambda']$ in the definition of $V_0^\circ$.
\end{remark}
%\fixme{This makes sure that $V^\circ_0$ and therefore $V(\mu)$ are $\hat \Lambda_0$-module}

\begin{thm}
Given the LLVOA $(V(0),Y, \omega,\bar\omega , \boldsymbol{1})$ and 
a subset of irreducible modules $\{(V(\mu),Y)\}$, such that for any $\mu,\nu \in \Lambda_0^\circ$, $\mu + \nu \in \Lambda_0^\circ$,
there exists an intertwining operator of type $\binom{V(\mu+\nu)}{V(\mu)~ V(\nu)}$ as follows
\begin{equation}
\begin{split}
\mathcal{Y}^{V(\mu+\nu)}_{V(\mu) V(\nu)}(w,x,\bar{x})=\typecolon \prod_{r=1}^{k}\prod_{s=1}^{\bar{k}}\left(\frac{1}{(m_r-1)!}\frac{d^{m_r - 1 }\alpha_{r}(x)}{dx^{m_r-1}}\right)&\left(\frac{1}{(\bar{m}_s-1)!}\frac{d^{\bar{m}_s-1}\beta_{s}(\bar{x})}{d\bar{x}^{\bar{m}_s-1}}\right)\\& \mathcal{Y}^{V(\mu+\nu)}_{V(\mu) V(\nu)}(\mathrm{e}^{\mu+(\alpha,\beta)},x,\bar{x})\typecolon ,  
\end{split}
\end{equation} 
where $w \in V(\mu)$ is of the form \eqref{eq:genvect} with 
$\lambda = \mu + (\alpha ,\beta) \in \mu + \Lambda_0$ and 

\begin{equation}
\mathcal{Y}^{V(\mu+\nu)}_{V(\mu) V(\nu)}(\mathrm{e}^{\mu+(\alpha,\beta)},x,\bar{x})\equiv Y(\mathrm{e}^{\mu+(\alpha,\beta)},x,\bar{x}),  \end{equation}
where the RHS is defined in \eqref{eq:veropelamb}.
\end{thm}
\subsection{Rationality of LLCFT}\label{sec:rationality_LLCFT}
Recall that the LLCFT is defined to be the non-chiral VOA $(V(0),Y_{0})$ along with the modules $\{(V(\mu),Y_{\mu})\}_{\mu\in\Lambda/\Lambda_0}$ \cite{Singh:2023mom}. Notice that we are restricting to the modules corresponding to the cosets $\Lambda/\Lambda_0$ to ensure the modular invariance of partition function.  
The characters for the modules are given by   \begin{equation}\label{eq:chichar}
\chi_\mu(\tau,\ov{\tau})=\frac{1}{\eta(\tau)^m\ov{\eta(\tau)}^n}\sum_{(\alpha,\beta)\in\Lambda_0+\mu}q^{\frac{\langle\alpha,\alpha\rangle}{2}}\ov{q}^{\frac{\langle\beta,\beta\rangle}{2}},\quad q={\rm e}^{2\pi i\tau},~~\ov{q}={\rm e}^{-2\pi i\ov{\tau}}\,.
\end{equation}
The partition function of the LLCFT is given by 
\begin{equation}\label{eq:partZlamb}
    Z_\Lambda(\tau,\ov{\tau})=\sum_{[\mu]\in\Lambda/\Lambda_0}\chi_\mu(\tau,\ov{\tau})~,
\end{equation}
and is modular invariant if and only if $\Lambda$ is self-dual and $m-n\equiv 0\bmod 24$. 
Recall also that a CFT is rational if it contains finitely many irreducible modules of the chiral algebra. 
%Let $V_\Lambda$ denote the LLVOA and $\{W_\mu\}_{[\mu]\in\Lambda/\Lambda_0}$ the irreducible modules constructing the LLCFT. 
It is clear that an LLCFT is rational if and only if $|\Lambda/\Lambda_0|<\infty$. The results of Appendix \ref{app:sublattices} then implies the following equivalent conditions for the rationality of the LLCFT.
\begin{thm}\label{thm:rcft}
Suppose the LLCFT based on $\Lambda$ is modular invariant. The following statements are equivalent:  
\begin{enumerate}
    \item The LLCFT based on $\Lambda$ is rational.
    \item $\Lambda_1^0$ is full rank.
    \item $\Lambda_2^0$ is full rank.
    \item $|\Lambda/\Lambda_0|=
|(\Lambda_1^0)^\star/\Lambda_1^0| = 
|(\Lambda_2^0)^\star/\Lambda_2^0|$.
\end{enumerate}
\end{thm}
In Appendix \ref{app:wend}, we show that when the signature of $\Lambda$ is $(m,m)$ then this is equivalent to the conditions for rationality obtained by Wendland \cite{Wendland:2000ye}. 
\section{Classification Of Irreducible Modules Of The LLVOA}\label{sec:moduleduallattice}
Following \cite{Dong1993VertexAA}, we classify all irreducible modules of the LLVOA. 
We will assume that 
\begin{equation}
    \mathsf{rank}(\Lambda)=\mathsf{rank}(\Lambda_0)=m+n~,
\end{equation}
or equivalently 
\begin{equation}
|\Lambda/\Lambda_0|<\infty~.    
\end{equation}
We record the following result, proved in Appendix \ref{app:sublattices}, which will be useful later.
\begin{lemma}\label{thm:rankL0mn}
Suppose $\mathsf{rank}(\Lambda_0)=m+n$. 
If $\Lambda$ is self-dual then $\Lambda'=(\Lambda_1^0)^\star\oplus (\Lambda_2^0)^\star$, where $(\Lambda_i^0)^\star$ is the dual of $\Lambda_i^0$.
\end{lemma}
% Let us introduce the notation 
% \begin{equation}
%  \begin{split}
% \mf{h}_0^i&:=\Lambda_0^i\otimes_\Z\C,\\\mf{h}_0&:=\mf{h}_0^1\oplus\mf{h}_0^2\\\mathsf{Root}(\mf{h}_0)&:=\{\lambda\in \mf{h}_0:\lambda\circ\lambda\in2\Z\}~.
%  \end{split}   
% \end{equation}
% Define the following set
% \begin{equation}
% \Lambda_0^\circ:=\left\{\lambda\in \mf{h} : \frac{(\lambda+\lambda')\circ(\lambda+\lambda')}{2} \in \Z, \text{ for all }\lambda'\in\Lambda_0\right\}=\Lambda_0^\star\cap\mathsf{Root}(\mf{h}_0)~,
% \end{equation}
% where $\Lambda^\star_0$ is the dual of $\Lambda_0$. 

% \\\\
We now want to prove the following analogue of Dong's classification theorem \cite{Dong1993VertexAA} for modules of vertex operator algebras:
\begin{thm}\label{thm:classllvoamod}
Let $(W,Y_W)$ be an irreducible module of the LLVOA $(V(0),Y)$.Then $(W,Y_W)$ is isomorphic to $(V(\mu),Y)$ for some $[\mu]\in\Lambda_0^\circ/\Lambda_0$.     
\end{thm}
We will prove this theorem in a series of lemmas and propositions.\\     Let $(W,Y_W)$ be a module for a non-chiral VOA $(V,Y_V)$. For $u\in V_{(h,\bar{h})},v\in V_{(h',\bar{h}')}$, chiral and anti-chiral  vectors respectively, introduce the notation\footnote{Recall that for chiral (anti-chiral) vectors $h - \bar{h} \; (\bar{h}' - h')$ is an integer and hence the sum below is well-defined.} 
\begin{equation}
\begin{split}
&Y_W(u,x)^+=\sum_{m\geq -(h-\bar{h})+1}x^W_{m}(u)x^{-m-(h-\bar{h})},\quad Y_W(u,x)^-=\sum_{m\leq -(h-\bar{h})}x^W_{m}(u)x^{-m-(h-\bar{h})},\\&Y_W(v,\Bar{x})^+=\sum_{m\geq -(\Bar{h}'-h')+1}\bar{x}^W_{m}(v)\Bar{x}^{-m-(\Bar{h}'-h')},\quad Y_W(v,\Bar{x})^-=\sum_{m\leq -(\Bar{h}'-h')}\bar{x}^W_{m}(v)\Bar{x}^{-m-(\Bar{h}'-h')}.
\end{split}    
\end{equation}
\begin{lemma}\label{lemma:modwwithverop}
For chiral and anti-chiral vectors $u\in V_{(h,\bar{h})},v\in V_{(h',\bar{h}')}$ and any $w\in V$, we have
\begin{equation}
\begin{split}
&\left[x_{\bar{h} -h+1}^W(u), Y_W\left(w, x, \bar{x}\right)\right]=Y_W\left(x_{\bar{h}-h+1}(u) \cdot w, x, \bar{x}\right),\\
&\left[\bar{x}_{h'-\bar{h}^{\prime}+1}^W(v), Y_W\left(w, x, \bar{x}\right)\right]=Y_W\left(\bar{x}_{h'-\bar{h}^{\prime}+1}(v) \cdot w, x, \bar{x}\right),
\end{split}
\end{equation}
and
\begin{equation}\label{eq:YWpmYWfull}
\begin{split}
&Y_W\left(u, x\right)^{-} Y_W\left(w, x, \bar{x}\right)+Y_W\left(w, x, \bar{x}\right) Y_W\left(u, x\right)^{+}= Y_W\left(x_{\bar{h}-h}(u)\cdot w, x, \bar{x}\right)\, ,\\
&Y_W\left(v, \bar{x}\right)^{-} Y_W\left(w, x, \bar{x}\right)+Y_W\left(w, x, \bar{x}\right) Y_W\left(v, \bar{x}\right)^{+}= Y_W\left(x_{h'-\bar{h}^{\prime}}(v)\cdot w, x, \bar{x}\right).
\end{split}
\end{equation}
\begin{proof}
Let $C^a_i(z)$ denote a contour in the variable $z_i$, in counterclockwise direction, of radius $a$ and centered around $z$. Further,  $C_{i}^{r} := C_{i}^{r}(0)$. Choosing $r_1>|z_2|>r_2>0$ and using locality of module vertex operators \eqref{eq:M_locality_fields}, we obtain
\begin{equation}
\begin{gathered}
\oint_{C_{1}^{r_1}} d z_1 Y_W\left(u, z_1\right) Y_W\left(w, z_2, \bar{z}_2\right)-\oint_{C_{1}^{r_2}} d z_1 Y_W\left(w, z_2, \bar{z}_2\right) Y_W\left(u, z_1\right) \\
=\oint_{C_{1}^\delta(z_2) } d z_1 \sum_{p \in \mathbb{Z}} Y_W\left(x_p(u) \cdot w, z_2, \bar{z}_2\right)\left(z_1-z_2\right)^{-p-(h-\bar{h})}
\end{gathered}
\end{equation}
This gives
\begin{equation}
\left[x_{\bar{h}-h+1}^W(u), Y_W\left(w, z_2, \bar{z}_2\right)\right]=Y_W\left(x_{\bar{h}-h+1}(u) \cdot w, z_2, \bar{z}_2\right) .
\end{equation}
Similarly
\begin{equation}
\left[\bar{x}_{h'-\bar{h}^{\prime}+1}^W(v), Y_W\left(w, z_2, \bar{z}_2\right)\right]=Y_W\left(\bar{x}_{h'-\bar{h}^{\prime}+1}(v) \cdot w, z_2, \bar{z}_2\right).
\end{equation}
% Next choose $r_1, r_2$ appropriately such that $\left|z_1\right|>\left|z_2\right|$ on $C_{1}^{r_1}$ and $\left|z_1\right|<\left|z_2\right|$ on $C_{2}^{r_2}$ and $0 \in \operatorname{Int}\left(C_{1}^{r_i}\right), i=1,2$. 
% \fixme{These sets of contours are redundant.}
Next we have 
\begin{equation}
\begin{aligned}
\oint_{C_{1}^{r_1}} d{z_1} Y_W\left(u, z_1\right)\left(z_1-z_2\right)^{-1} & =\oint_{C_{1}^{r_1}} d z_1 Y_W\left(u, z_1\right) \sum_{n=0}^{\infty} z_2^n z_1^{-n-1} \\
& =\sum_{n=0}^{\infty} \oint_{C_{1}^{r_1}} d z_1 \sum_{m \in \mathbb{Z}} x_m^W(u) z_1^{-m-(h-\bar{h})} z_2^n z_1^{-n-1} \\
& =\sum_{n=0}^{\infty} x_{\bar{h}-h-n}^W(u) z_2^n \\
& =\sum_{n \leq-(h-\bar{h})} x_n^W(u) z_2^{-n-(h-\bar{h})} \\
& =Y_W\left(u, z_2\right)^{-}
\end{aligned}
\end{equation}
Similarly
\begin{equation}
\begin{aligned}
\oint_{C_{1}^{r_2}} 
d z_1 Y_W\left(u, z_1\right)\left(z_1-z_2\right)^{-1}&=-\oint_{C_{1}^{r_2}}
d z_1
Y_W\left(u, z_1\right)\left(z_2-z_1\right)^{-1} \\&=  -Y_W\left(u, z_2\right)^{+} .
\end{aligned}
\end{equation}
Then from
\begin{equation}\label{eq:contintYvonw}
\begin{aligned}
& \oint_{C_{1}^{r_1}} d z_1 Y_W\left(u, z_1\right) Y_W\left(w, z_2, \bar{z}_2\right)\left(z_1-z_2\right)^{-1}-\oint_{C_{1}^{r_2}} d z_1 Y_W\left(w, z_2, \bar{z}_2\right) Y_W\left(u, z_1\right)\left(z_1-z_2\right)^{-1} \\
& =\oint_{C_{1}^\delta(z_2)} d z_1 \sum_{p \in \mathbb{Z}} Y_W\left(x_p(u) \cdot w, z_2, \bar{z}_2\right)\left(z_1-z_2\right)^{-p-(h-\bar{h})}\left(z_1-z_2\right)^{-1},
\end{aligned}
\end{equation}
we obtain
\begin{equation}
\label{eq:contour-relation1}
Y_W\left(u, z_2\right)^{-} Y_W\left(w, z_2, \bar{z}_2\right)+Y_W\left(w, z_2, \bar{z}_2\right) Y_W\left(u, z_2\right)^{+}=Y_W\left(x_{\bar{h}-h}(u)\cdot w, z_2, \bar{z}_2\right).
\end{equation}
Similarly
\begin{equation}\label{eq:contour-relation2}
Y_W\left(v, \bar{z}_2\right)^{-} Y_W\left(w, z_2, \bar{z}_2\right)+Y_W\left(w, z_2, \bar{z}_2\right) Y_W\left(v, \bar{z}_2\right)^{+}= Y_W\left(x_{h'-\bar{h}^{\prime}}(v)\cdot w, z_2, \bar{z}_2\right).
\end{equation}
\end{proof}
\end{lemma}
The following vectors
\begin{equation}
\begin{split}
&\alpha(-1)\otimes\mathbf{1}\in S(\hat{\mf{h}}_1^-)\otimes \C[\Lambda_0]\subset S(\hat{\mf{h}}^-)\otimes \C[\Lambda_0] \, ,\\&\beta(-1)\otimes\mathbf{1}\in S(\hat{\mf{h}}_2^-)\otimes \C[\Lambda_0]\subset S(\hat{\mf{h}}^-)\otimes \C[\Lambda_0] \, , 
\end{split}
\end{equation}
are chiral and anti-chiral respectively, as a consequence of \eqref{eq:translation_module} and the relation $\bar{L}(-1) (\alpha(-1) \otimes \mathbf{1})= 0$ and $L(-1) (\beta(-1) \otimes \mathbf{1}) = 0$.  We denote the modes by  
\begin{equation}\label{eq:alpxWexp}
\begin{split}
&Y_W(\alpha(-1)\otimes\mathbf{1},x)=\sum_{m\in\Z}\alpha^W(m)x^{-m-1}=:\alpha^W(x)\, , \\&Y_W(\beta(-1)\otimes\mathbf{1},\bar x)=\sum_{m\in\Z}\beta^W(m)\bar{x}^{-m-1}=:\beta^W(\bar{x}) \, .
\end{split}      
\end{equation}
\begin{lemma}\label{lemma:Wisoheisrep}
Let $(W,Y_W)$ be an irreducible $V(0)$-module. Then 
\begin{equation}
    W\cong S(\hat{\mf{h}}^-)\otimes \Omega_W \, ,
\end{equation}
where $\Omega_W$ is the vacuum space of $W$ defined by 
\begin{equation}
\Omega_W= \{w\in W : \alpha^W(n)\cdot w=\beta^W(m)\cdot w=0,\alpha\in\mf{h}_1,\beta\in\mf{h}_2,n,m>0 \}.   
\end{equation}
\begin{proof}
Note that 
\begin{equation}
\alpha(-1) \otimes \mathbf{1} \in V_{(1,0)},\quad \beta(-1) \otimes \mathbf{1} \in V_{(0,1)}.     
\end{equation}
Hence 
\begin{equation}
x_n^W(\alpha_1(-1)\otimes\mathbf{1}) = \alpha_1^W(n),\quad    \bar{x}_n^W(\beta_1(-1)\otimes\mathbf{1}) = \beta_1^W(n).    
\end{equation}
Then using \cite[Theorem 4.1]{Singh:2023mom} and the commutation relations in \eqref{liebracketmodes}, one can show that
\begin{equation}\label{eq:heisalgW}
\begin{split}    &[\alpha_1^W(m),\alpha_2^W(n)]=m\langle\alpha_1,\alpha_2\rangle\delta_{m+n,0} \, ,\\
&[\beta_1^W(m),\beta_2^W(n)]=m\langle\beta_1,\beta_2\rangle\delta_{m+n,0} \, , \\&[\alpha^W(n),\beta^W(n)]=0.
\end{split}    
\end{equation}
Thus these operators form an algebra isomorphic to $\hat{\mf{h}}$.
This means that $W$ is also an $\hat{\mf{h}}$-module. Thus from \cite[Theorem C.2]{Singh:2023mom} we conclude the result.
\end{proof}
\end{lemma}
To proceed further, we define the normal ordering for the operators $\alpha_1^W(n),\alpha_2^W(n) $ for $\alpha_1,\alpha_2 \in\mf{h}_1$ :
\begin{equation}
    \typecolon \alpha_1^W(m) \,  \alpha_2^{W}(n)\typecolon = \typecolon \alpha_2^{W}(n) \,  \alpha_1^W(m)\typecolon =\begin{cases}
\alpha_1^W(m) \, \alpha_2^{W}(n)& m\leq n\\\alpha_2^{W}(n) \,  \alpha_1^W(m)& m\geq n
\end{cases}
\end{equation}
and similarly for $\beta_1^W(n), \beta_2^W(n)$ for  $\beta_1, \beta_2\in\mf{h}_2$. We then define the normal ordering of $\alpha_1^W(x),\alpha_2^W(x)$:
\begin{equation}
    \typecolon\alpha_1^W(x_1)\alpha_2^W(x_2)\typecolon=\sum_{m,n\in\Z}\typecolon\alpha_1^W(m)\alpha_2^W(n)\typecolon x_1^{-m-1}x_2^{-n-1},
\end{equation}
and similarly for $\beta_1^W(\bar{x}), \beta_2^W(\bar{x})$. Then one can easily check that 
\begin{equation}
\begin{split}
    &\typecolon\alpha_1^W(x_1)\alpha_2^W(x_2)\typecolon=\alpha_1^W(x_1)^-\alpha_2^W(x_2)+\alpha_2^W(x_2)\alpha_1^W(x_1)^+\, ,\\&\typecolon\beta_1^W(\bar{x}_1)\beta_2^W(\bar{x}_2)\typecolon=\beta_1^W(\bar{x}_1)^-\beta_2^W(\bar{x}_2)+\beta_2^W(\bar{x}_2)\beta_1^W(\bar{x}_1)^+ \, ,
\end{split}
\end{equation}
where 
\begin{equation}
    \label{eq: alphabetapm}
    \alpha^W(x)^\pm:=Y_W(\alpha(-1)\otimes \mathbf{1},x)^\pm,\quad \beta^W(\bar{x})^\pm:=Y_W(\beta(-1)\otimes \mathbf{1},\bar{x})^\pm.
\end{equation}
Using this along with Lemma \ref{lemma:modwwithverop} we obtain the following
\begin{lemma}\label{lemma:veropWal1al2}
For $\alpha_i\in \mf{h}_1,\beta_i\in\mf{h}_2,~i=1,2$ and $\lambda\in\Lambda_0$, we have 
\begin{equation}
\begin{split}    &Y_W(\alpha_1(-1)\alpha_2(-1)\otimes\mathbf{1},x)=\typecolon\alpha_1^W(x)\alpha_2^W(x)\typecolon\, ,\\&Y_W(\beta_1(-1)\beta_2(-1)\otimes\mathbf{1},\bar{x})=\typecolon\beta_1^W(\bar{x})\beta_2^W(\bar{x})\typecolon\, ,\\&Y_W(\alpha_i(-1)\beta_j(-1)\otimes\mathbf{1},x,\bar{x})=\alpha_i^W(x)\beta_j^W(\bar{x}) \, ,\\&Y_W(\alpha_i(-1)\otimes\mathrm{e}^\lambda,x,\bar{x})=\alpha_i^W(x)^-Y_W(\mathrm{e}^\lambda,x,\bar{x})+Y_W(\mathrm{e}^\lambda,x,\bar{x})\alpha_i^W(x)^+ \, , \\&Y_W(\beta_i(-1)\otimes\mathrm{e}^\lambda,x,\bar{x})=\beta_i^W(\bar{x})^-Y_W(\mathrm{e}^\lambda,x,\bar{x})+Y_W(\mathrm{e}^\lambda,x,\bar{x})\beta_i^W(\bar{x})^+.
\end{split}
\end{equation}
\begin{proof}
%\fixme{How is $\alpha_1^W$ in any way different from $\alpha_1$, see our claim in 5.4, $Y_\mu = Y_{V_{\Lambda}}$}
These relations can be proved using \eqref{eq:contour-relation1}. For example
\begin{equation}
\begin{aligned}
&Y_W(x_{-1}(\alpha_1(-1)\otimes \boldsymbol{1})\cdot \alpha_2(-1)\otimes \boldsymbol{1})\\ =& Y_W(\alpha_1(-1)\otimes \boldsymbol{1}, x)^{-} Y_W(\alpha_2(-1)\otimes\boldsymbol{1},x) +Y_W(\alpha_2(-1)\otimes \boldsymbol{1},x)Y_W(\alpha_1(-1)\otimes\boldsymbol{1})^+ \, , \\=&\alpha_1^W(x)^- \alpha_2^W(x) + \alpha_2^W(x)\alpha_1^W(x)^+ \, , \\=&
\typecolon\alpha_1^W\left(x_1\right) \alpha_2^W\left(x_2\right)\typecolon \, .
\end{aligned}
\end{equation}
For the third relation, we need to use the fact that $[\alpha_i(x),\beta_j(\bar{x})]=0$.   

\end{proof}
\end{lemma}
\begin{cor}
We have 
\begin{equation}\label{eq:VirgenW}
\begin{split}
    &L^W(p)=\frac{1}{2}\sum_{i=1}^{m}\sum_{q\in\Z}\typecolon u_i^W(q)u_i^W(p-q)\typecolon \, , \\&\bar{L}^W(p)=\frac{1}{2}\sum_{i=1}^{n}\sum_{q\in\Z}\typecolon v_i^W(q)v_i^W(p-q)\typecolon \, , 
\end{split}    
\end{equation}
where $\{u_i\},\{v_i\}$ are orthonormal basis of $\mf{h}_1,\mf{h}_2$  respectively as in \eqref{eq:Virgenlamb}.
\begin{proof}
Since 
\begin{equation}
    \omega=\frac{1}{2}\sum_{i=1}^{m}u_i(-1)^2\otimes\mathbf{1},\quad \overline{\omega}=\frac{1}{2}\sum_{i=1}^{n}v_i(-1)^2\otimes\mathbf{1},
\end{equation}
using Lemma \ref{lemma:veropWal1al2} we get 
\begin{equation}
\begin{split}
    &Y_W(\omega,x)=\frac{1}{2}\sum_{i=1}^{m}\typecolon u_i^W(x)u_i^W(x)\typecolon \, , \\&Y_W(\overline{\omega},\bar{x})=\frac{1}{2}\sum_{i=1}^{n}\typecolon v_i^W(\bar{x})v_i^W(\bar{x})\typecolon.
\end{split}    
\end{equation}
Expanding $u^W_i(x),v^W_i(\Bar{x})$ as in \eqref{eq:alpxWexp} and using the Cauchy product formula, we obtain the desired result.
\end{proof}
\end{cor}
Following the work of Dong \cite{Dong1993VertexAA}, we now define the $Z$-operators. For $\lambda=(\alpha_\lambda,\beta_\lambda)\in\Lambda_0$ define the $Z$-operator by
\begin{equation}
\begin{split}
    \label{eq:defZopera}Z(\mathrm{e}^\lambda,x,\bar{x})&:=\left[\exp\left(\sum_{n<0}\frac{\alpha^W_\lambda(n)}{n}x^{-n}\right)\exp\left(\sum_{n<0}\frac{\beta^W_\lambda(n)}{n}\bar{x}^{-n}\right)Y_W(\mathrm{e}^\lambda,x,\bar{x})\right.\\&\quad~~\left.\exp\left(\sum_{n>0}\frac{\alpha^W_\lambda(n)}{n}x^{-n}\right)\exp\left(\sum_{n>0}\frac{\beta^W_\lambda(n)}{n}\bar{x}^{-n}\right)\right]\\&=:\sum_{\substack{m,n\in\C\\m-n\in\Z}}Z_{m,n}(\lambda)x^{-m-h_\lambda}\bar{x}^{-n-\bar{h}_\lambda},
\end{split}   
\end{equation}
where %\fixme{type, set $m \to n$ in the power of $\bar{x}$ in last equal}
\begin{equation}\label{eq:hlambda}
h_\lambda=\frac{\langle\alpha_\lambda,\alpha_\lambda\rangle}{2}\in\Z,\quad \bar{h}_\lambda=\frac{\langle\beta_\lambda,\beta_\lambda\rangle}{2}\in\Z \, .   
\end{equation}
$Z(\mathrm{e}^\lambda,x,\Bar{x}) : W \to W \{x, \bar{x} \}$ is a well-defined
\footnote{It is important to check that once $Z(\mathrm{e}^\lambda,x,\Bar{x})$ acts on $w \in W$, for each $x^m \bar{x}^n$ only finitely many terms survive, so that the sum lies in $W$. } map, due to the Property \ref{item:M_gradres} of modules of non-chiral VOA.   
From Lemma \ref{lemma:modwwithverop}, it is easy to see that for $\alpha\in\mf{h}_1,\beta\in\mf{h}_2$, and $u^\star\in S(\hat{\mf{h}}^-)$ we have 
\begin{equation}\label{eq:commalph0YW}
\begin{split}    &[\alpha^W(0),Y_W(u^\star\otimes\mathrm{e}^\lambda,x,\bar{x})]=\langle\alpha,\alpha_\lambda\rangle Y_W(u^\star\otimes\mathrm{e}^\lambda,x,\bar{x}) \, , \\&[\beta^W(0),Y_W(u^\star\otimes\mathrm{e}^\lambda,x,\bar{x})]=\langle\beta,\beta_\lambda\rangle Y_W(u^\star\otimes\mathrm{e}^\lambda,x,\bar{x})\, .
\end{split}
\end{equation}
Since $\alpha^W(0)$ commutes with $\alpha_\lambda^W(n),\beta^W_\lambda(n)$ for $n\neq 0$ we have 
\begin{equation}\label{eq:commalph0Z}
    \begin{split}
        &[\alpha^W(0),Z(\mathrm{e}^\lambda,x,\Bar{x})]=\langle\alpha,\alpha_\lambda\rangle Z(\mathrm{e}^\lambda,x,\Bar{x}),\\&[\beta^W(0),Z(\mathrm{e}^\lambda,x,\Bar{x})]=\langle\beta,\beta_\lambda\rangle Z(\mathrm{e}^\lambda,x,\Bar{x}),
    \end{split}
\end{equation}
which implies 
\begin{equation}\label{eq:commalp0Zmn}
\begin{split}
        &[\alpha^W(0),Z_{m,n}(\lambda)]=\langle\alpha,\alpha_\lambda\rangle Z_{m,n}(\lambda),\\&[\beta^W(0),Z_{m,n}(\lambda)]=\langle\beta,\beta_\lambda\rangle Z_{m,n}(\lambda) .
    \end{split}    
\end{equation}
In addition we have the following lemma. 
\begin{lemma}
We have 
\begin{equation}\label{eq:commalphnZmn}
\begin{split}
        &[\alpha^W(p),Z_{m,n}(\lambda)]=0,\quad[\beta^W(p),Z_{m,n}(\lambda)]=0.
    \end{split}        
\end{equation}
for $p\neq 0$. In particular, $\Omega_W$ is invariant under the action of $Z_{m,n}(\lambda)$ for all $m,n\in\C$ with $m-n\in\Z$.
\begin{proof}
Using the contour integral argument (2.85) in \cite{Singh:2023mom}, we get

\begin{equation}
    [\alpha^W(n),Y_W(\mathrm{e}^\lambda,x,\bar{x})]=\sum_{p\geq 0}{n\choose p}Y_W(\alpha(p)\cdot\mathrm{e}^\lambda,x,\Bar{x})x^{n-p}.
\end{equation}
Since $\alpha(p)\cdot \mathrm{e}^\lambda=0$ for $p>0$ and $\alpha(p)\cdot \mathrm{e}^\lambda=\langle\alpha,\alpha_\lambda\rangle \mathrm{e}^\lambda$ for $p = 0$, we obtain 
\begin{equation}\label{eq:alphaWcomWverop}
[\alpha^W(n),Y_W(\mathrm{e}^\lambda,x,\bar{x})]=\langle\alpha,\alpha_\lambda\rangle x^n Y_W(\mathrm{e}^\lambda,x,\bar{x}).    
\end{equation}
Similarly we get 
\begin{equation}
[\beta^W(n),Y_W(\mathrm{e}^\lambda,x,\bar{x})]=\langle\beta,\beta_\lambda\rangle \Bar{x}^nY_W(\mathrm{e}^\lambda,x,\bar{x}).    
\end{equation}
Next, using the Heisenberg algebra \eqref{eq:heisalgW} one can easily show that 
\begin{equation}\label{eq:exp_modes_W_Hei_Alg}
\begin{split}    &\left[\alpha^W(m),\exp\left(\sum_{n<0}\frac{\alpha^W_\lambda(n)}{n}x^{-n}\right)\right]=-\langle\alpha,\alpha_\lambda\rangle x^m\exp\left(\sum_{n<0}\frac{\alpha^W_\lambda(n)}{n}x^{-n}\right),\quad m>0\\&\left[\alpha^W(m),\exp\left(\sum_{n>0}\frac{\alpha^W_\lambda(n)}{n}x^{-n}\right)\right]=-\langle\alpha,\alpha_\lambda\rangle x^m\exp\left(\sum_{n>0}\frac{\alpha^W_\lambda(n)}{n}x^{-n}\right),\quad m<0\\&\left[\beta^W(m),\exp\left(\sum_{n<0}\frac{\beta^W_\lambda(n)}{n}\bar{x}^{-n}\right)\right]=-\langle\beta,\beta_\lambda\rangle \bar{x}^m\exp\left(\sum_{n<0}\frac{\beta^W_\lambda(n)}{n}\bar{x}^{-n}\right),\quad m>0\\&\left[\beta^W(m),\exp\left(\sum_{n>0}\frac{\beta^W_\lambda(n)}{n}\bar{x}^{-n}\right)\right]=-\langle\beta,\beta_\lambda\rangle \bar{x}^m\exp\left(\sum_{n>0}\frac{\beta^W_\lambda(n)}{n}\bar{x}^{-n}\right),\quad m<0.
\end{split}
\end{equation}
The commutators clearly vanish for other ranges of $m$. 
We then have for $m<0$
\begin{equation}
\begin{split}
[\alpha^W(m),Z(\mathrm{e}^\lambda,x,\Bar{x})]&= \exp\left(\sum_{n<0}\frac{\alpha^W_\lambda(n)}{n}x^{-n}\right)\exp\left(\sum_{n<0}\frac{\beta^W_\lambda(n)}{n}\bar{x}^{-n}\right)Y_W(\mathrm{e}^\lambda,x,\bar{x})\\&\quad\,\left[\alpha^W(m),\exp\left(\sum_{n>0}\frac{\alpha^W_\lambda(n)}{n}x^{-n}\right)\right]\exp\left(\sum_{n>0}\frac{\beta^W_\lambda(n)}{n}\bar{x}^{-n}\right)\\&+ \,\exp\left(\sum_{n<0}\frac{\alpha^W_\lambda(n)}{n}x^{-n}\right)\exp\left(\sum_{n<0}\frac{\beta^W_\lambda(n)}{n}\bar{x}^{-n}\right)\left[\alpha^W(m),Y_W(\mathrm{e}^\lambda,x,\bar{x})\right]\\&\quad~\exp\left(\sum_{n>0}\frac{\alpha^W_\lambda(n)}{n}x^{-n}\right)\exp\left(\sum_{n>0}\frac{\beta^W_\lambda(n)}{n}\bar{x}^{-n}\right)\\&=0,
\end{split}    
\end{equation}
where we used \eqref{eq:alphaWcomWverop} and the above equation. Similarly for $m>0$, the commutator vanishes. Thus we get  
\begin{equation}\label{eq:Zcommalphabeta}
\begin{split}
&[\alpha^W(m),Z(\mathrm{e}^\lambda,x,\Bar{x})]=0 \, , 
\\&[\beta^W(m),Z(\mathrm{e}^\lambda,x,\Bar{x})]=0 \, ,
\end{split}    
\end{equation}
for all $m\neq 0$. This proves the required result. \\
As $Z_{m,n}(\lambda)$ commutes with $a^{W}(p)$ when $p >0$, it is obvious that $\Omega_W$ is invariant under the action of $Z_{m,n}(\lambda)$.

\end{proof}
\end{lemma}
We next have the following lemma.
\begin{lemma}\label{lemma:derofZ}
For $\lambda=(\alpha_\lambda,\beta_\lambda)\in\Lambda_0$ we have
\begin{equation}
\begin{split}
    \frac{\partial}{\partial x}Z(\mathrm{e}^\lambda,x,\Bar{x})=Z(\mathrm{e}^\lambda,x,\Bar{x})\alpha^W_\lambda(0)x^{-1} \, , \\\frac{\partial}{\partial \bar{x}}Z(\mathrm{e}^\lambda,x,\Bar{x})=Z(\mathrm{e}^\lambda,x,\Bar{x})\beta^W_\lambda(0)\bar{x}^{-1}.
\end{split}    
\end{equation}
Or equivalently $($see \eqref{eq:defZope}, \eqref{eq:hlambda}$)$
\begin{equation}\label{eq:Zmnalplamb}
\begin{split}
    Z_{m,n}(\lambda)(m+h_\lambda+\alpha^W_\lambda(0))=0\\Z_{m,n}(\lambda)(n+\bar{h}_\lambda+\beta^W_\lambda(0))=0.
\end{split}
\end{equation}
\begin{proof}
The proof is essentially the same as the proof of \cite[Lemma 3.3]{Dong1993VertexAA}.
\end{proof}
\end{lemma}
\begin{lemma}
There exists a non-zero $w\in\Omega_W$ and linear functionals $\varphi_i\in (\mf{h}_{i})^{\star},~i=1,2$ such that 
\begin{equation}
    \alpha^W(0)\cdot w=\varphi_1(\alpha)w,\quad \beta^W(0)\cdot w=\varphi_2(\beta)w,
\end{equation}
for every $\alpha\in\mf{h}_1,\beta\in\mf{h}_2$.
\begin{proof}
Since 
\begin{equation}
\begin{split}
    Y_W(\mathrm{e}^\lambda,x,\bar{x})&:=\left[\exp\left(-\sum_{n<0}\frac{\alpha^W_\lambda(n)}{n}x^{-n}\right)\exp\left(-\sum_{n<0}\frac{\beta^W_\lambda(n)}{n}\bar{x}^{-n}\right)Z(\mathrm{e}^\lambda,x,\bar{x})\right.\\&\quad~~\left.\exp\left(-\sum_{n>0}\frac{\alpha^W_\lambda(n)}{n}x^{-n}\right)\exp\left(-\sum_{n>0}\frac{\beta^W_\lambda(n)}{n}\bar{x}^{-n}\right)\right],
\end{split}       
\end{equation}
by Proposition \ref{prop:veropnonzero}, 
\begin{equation}\label{eq:Zw'neq0}
Z(\mathrm{e}^\lambda,x,\bar{x})w'\neq 0    
\end{equation}
for every $\lambda\in\Lambda_0$ and $0\neq w'\in W$. By assumption, $d:=\mathsf{rank}(\Lambda_0)=m+n$. Choose an integral basis $\{\lambda_i=(\alpha_i,\beta_i)\}_{i=1}^{d}$ of $\Lambda_0$. Since $\Lambda_0$ is full rank, we have 
\begin{equation}
\mf{h}_1=\mathsf{Span}_\C\{\alpha_1,\dots\alpha_{d}\},\quad \mf{h}_2=\mathsf{Span}_\C\{\beta_1,\dots\beta_{d}\}~.    
\end{equation}
%Then $\{\alpha_i+\beta_j : i=1,\dots,d,j=1,\dots d_2\}$ is a basis of $\mf{h}=\mf{h}_1\oplus\mf{h}_2$. 
Choose a non-zero $w'\in \Omega_W$. 
Then by \eqref{eq:Zw'neq0}, there exists $m_1,n_1\in\C$ with $m_1-n_1\in\Z$ such that $Z_{m_1,n_1}(\lambda_1) w'\neq 0$. 
Then using \eqref{eq:commalp0Zmn} and \eqref{eq:Zmnalplamb} we get 
\begin{equation}
\begin{split}
    &\alpha_1^W(0)Z_{m_1,n_1}((\alpha_1,\beta_1))w'=\left(\frac{\langle\alpha_1,\alpha_1\rangle}{2}-m_1\right)Z_{m_1,n_1}((\alpha_1,\beta_1))w' \, , \\&\beta_1^W(0)Z_{m_1,n_1}((\alpha_1,\beta_1))w'=\left(\frac{\langle\beta_1,\beta_1\rangle}{2}-n_1\right)Z_{m_1,n_1}((\alpha_1,\beta_1))w' \, .
\end{split}    
\end{equation}
From the above equation we note that $w_{1}:=Z_{m_1,n_1}(\lambda_1) w'\in \Omega_W$ is a simultaneous eigenvector of $\alpha^W_1(0)$ and $\beta^W_1(0)$. Following the same procedure with $w'$ replaced by $w_{1}$ we obtain another simultaneous eigenvector $w_{2}\neq 0$ of $\alpha^W_2(0),\beta^W_2(0)$. By \eqref{eq:commalp0Zmn}, $w_{2}$ is also a simultaneous eigenvector of $\alpha^W_1(0),\beta^W_1(0)$. 
%Order the set \begin{equation}\{(i,j):i=1,\dots, d_1,j=1,\dots, d_2\} \end{equation}by \begin{equation}\label{eq:orderonij}(i,j)<(i',j')\iff i<i' \text{ if $i\neq i'$} \text{ and } j<j' \text{ if $i=i'$}.\end{equation}
Following the above procedure, we obtain non-zero vectors $\{w_{i}:i=1,\dots, d\}$ such that $w_{d}$ is a simultaneous eigenvector of $\alpha^W_i(0),\beta^W_i(0)$ for all $i=1,\dots, d$. The linear functionals $\varphi_i,~i=1,2$ is defined on the $\C$-spanning set $\{\alpha_i : i=1,\dots,d\},\{\beta_i : i=1,\dots,d\}$ by the eigenvalues of $w_{d}$ under $\alpha_j^W(0),\beta^W_j(0)$:
\begin{equation}
    \varphi_1(\alpha_i)  w_d =\alpha^W_i(0)w_{d},\quad \varphi_2(\beta_i)  w_d =\beta^W_i(0)w_{d} \, .
\end{equation}
\end{proof}
\end{lemma}
Let $\alpha_{\varphi}\in \mf{h}_1,\beta_{\varphi}\in \mf{h}_2$ be the vectors corresponding to $\varphi_1\in (\mf{h}_1)^\star,\varphi_2\in (\mf{h}_2)^\star$ under the isomorphism 
\begin{equation}
    \alpha\in\mf{h}_1\mapsto\langle\cdot,\alpha\rangle\in(\mf{h}_1)^\star,\quad \beta\in\mf{h}_2\mapsto\langle\cdot,\beta\rangle\in(\mf{h}_2)^\star \, , 
\end{equation}
respectively.  
Therefore for any irreducible module $W$,  there exist 
$w_d \in \Omega_W$ and 
$\alpha_\varphi \in \mathfrak{h}_1$ and 
$\beta_\varphi \in \mathfrak{h}_2$ such that 
\begin{equation}
    \alpha^W(0)\cdot w_d 
    =
    \langle \alpha, \alpha_\varphi \rangle w_d , \quad 
    \beta^W(0)\cdot w_d 
    =
    \langle \beta, \beta_\varphi \rangle w_d \, .
\end{equation}
Also put 
\begin{equation}\label{eq:lambphi}
\lambda_\varphi=(\alpha_{\varphi},\beta_{\varphi})\in\mf{h}=\mf{h}_1\oplus \mf{h}_2 \, .    
\end{equation}
%\fixme{\textcolor{red}{Requires that $\Lambda_0$ be full rank since $<.,.>$ is non-degenerate only on $\R^m,\R^n$ not on subspaces.  }}
We then have the following proposition.
\begin{prop}
We have that, for any irreducible module $W$ 
\begin{equation}
    \Omega_W=\bigoplus_{\lambda=(\alpha_\lambda,\beta_\lambda)\in\Lambda_0}\Omega_W(\lambda_\varphi+\lambda) \, , 
\end{equation}
where $\Omega_W(\lambda_\varphi+\lambda)$ is defined by 
\begin{equation}
\begin{split}    \Omega_W(\lambda_\varphi+\lambda):=\left\{ w\in\Omega_W \middle \vert~\forall~(\alpha,\beta)\in\mf{h},~\begin{array}{l}
    \alpha^W(0)\cdot w=\langle\alpha,\alpha_\varphi + \alpha_\lambda\rangle  w \\
    \beta^W(0)\cdot w=\langle\beta,\beta_\varphi +  \beta_\lambda \rangle w
  \end{array}\right\}.
\end{split}
\end{equation}
Moreover, $w_d \in \Omega_W(\lambda_\varphi)$ . 
\begin{proof}
By \eqref{eq:commalph0YW}, we see that the subspace 
% \begin{equation}
%     \bigoplus_{\lambda=(\alpha_\lambda,\beta_\lambda)\in\Lambda_0}\Omega_W(\lambda_\varphi+\lambda)
% \end{equation}
% is invariant under $Y_W(\cdot,x,\bar{x})$. Thus 
\begin{equation}
S(\hat{\mf{h}}^-)\otimes\bigoplus_{\lambda=(\alpha_\lambda,\beta_\lambda)\in\Lambda_0}\Omega_W(\lambda_\varphi+\lambda) \, ,    
\end{equation}
with the vertex operator $Y_W$ is a submodule of $(W,Y_W)$. Since $W$ is an irreducible representation, we obtain the result.  
\end{proof}
\end{prop}
\begin{prop}
We have that $\lambda_\varphi\in\Lambda_0^\circ$ $($see \eqref{eq:lambphi}$)$.    
\end{prop}
\begin{proof}
Let $\lambda\in\Lambda_0$ be arbitrary. 
For $w\in \Omega_W(\lambda_\varphi+\lambda) $, using \eqref{eq:VirgenW} we have 
\begin{equation}
    L^W(0)\cdot w=\frac{1}{2}\sum_{i=1}^{m}\left(\langle u_i,\alpha_\lambda+\alpha_\varphi\rangle\right)^2 w=\frac{\langle\alpha_\lambda+\alpha_\varphi,\alpha_\lambda+\alpha_\varphi\rangle}{2}w~,
\end{equation}
where we used the fact that $\{u_i\}_{i=1}^m$ is an orthonormal basis of $\mf{h}_1$.
Similarly
\begin{equation}
    \bar{L}^W(0)\cdot w=\frac{\langle\beta_\lambda+\beta_\varphi,\beta_\lambda+\beta_\varphi\rangle}{2}w \, .
\end{equation}
Thus the conformal weights of the whole space $\Omega_W(\lambda_\varphi+\lambda)$ is the same. 
Since the module satisfies the single-valuedness property \ref{item:M_singvalprop}, we must have 
\begin{equation}
\begin{split}
\frac{\langle\alpha_\lambda+\alpha_\varphi,\alpha_\lambda+\alpha_\varphi\rangle}{2}-\frac{\langle\beta_\lambda+\beta_\varphi,\beta_\lambda+\beta_\varphi\rangle}{2}&=\frac{\lambda_\varphi\circ\lambda_\varphi}{2}+\lambda_\varphi\circ\lambda+\frac{\lambda\circ\lambda}{2}\\&=\frac{(\lambda_\varphi+\lambda)\circ(\lambda_\varphi+\lambda)}{2}\in\Z~.
\end{split}
\end{equation}
In particular, for $\lambda=0$, we get $\lambda_\varphi\circ\lambda_\varphi\in 2\Z$ and for $\lambda=(\alpha,0),(0,\beta)\in\Lambda_0$ we obtain $\alpha_\varphi\in(\Lambda_1^0)^\star$ and $\beta_\varphi\in(\Lambda_2^0)^\star$. By Lemma \ref{thm:rankL0mn} we conclude that $\lambda_\varphi\in\Lambda'$. 
%Since $\lambda\circ\lambda\in 2\Z$ 
Thus we conclude that $\lambda_\varphi\in\Lambda_0^\circ$.
\end{proof}
\begin{comment}
Let $w_i\in \Omega_W(\lambda_\varphi+\lambda_i),~i=1,2$ for arbitrary $\lambda_1,\lambda_2\in\Lambda_0$. 

Since $W$ is irreducible, for any $\lambda,\lambda'\in\Lambda_0$, there exists $v\in V_\Lambda$ such that 
\begin{equation}
    x^W_{m,n}(v)\cdot w_1=w_2.
\end{equation}
Now for any $v\in V_\Lambda$, from \eqref{eq:L0commmodesgenver}, we see that for $m,n\in\Z$
\begin{equation}
\begin{split}
    &[L^W(0),x^W_{m,n}(v)]=-mx^W_{m,n}(v)\\&[\bar{L}^W(0),x^W_{m,n}(v)]=-nx^W_{m,n}(v).
\end{split}    
\end{equation}
Thus the conformal weights of $\Omega_W(\lambda_\varphi+\lambda_1),\Omega_W(\lambda_\varphi+\lambda_2)$ differ by an integer. This means that 
\begin{equation}
    \frac{\langle\alpha_\lambda+\alpha_\varphi,\alpha_\lambda+\alpha_\varphi\rangle}{2}-
\end{equation}
\end{comment}
We will now show that $(W,Y_W)\cong (V(\lambda_\varphi),Y)$ as $V(0)$-modules, (see definition \eqref{eq:equivalent class}). For $\lambda'\in\Lambda_0$, define the operators  $\alpha^W_{\lambda'}:=\alpha^W_{\lambda'}(0)$ and $\beta^W_{\lambda'}:=\beta^W_{\lambda'}(0)$ so that  $x^{\alpha^W_{\lambda'}},\Bar{x}^{\beta^W_{\lambda'}}$ acts on $\Omega_W$ by 
\begin{equation}\label{eq:xalponOW}
\begin{split}
&x^{\alpha^W_{\lambda'}}\cdot w=x^{\langle\alpha_{\lambda'},\alpha_\varphi+\alpha_\lambda\rangle}w, \\&\bar{x}^{\beta^W_{\lambda'}}\cdot w=\bar{x}^{\langle\beta_{\lambda'},\beta_\varphi+\beta_\lambda\rangle}w,\quad w\in\Omega_W(\lambda_\varphi+\lambda),~\lambda\in\Lambda_0 \, .     
\end{split}
\end{equation}
$x^{\alpha^W_{\lambda'}},\Bar{x}^{\beta^W_{\lambda'}}$ acts as identity on $S(\hat{\mf{h}}^-)$. Thus it is clear that these operator commute among themselves and also with $\alpha^W(n),\beta^W(n),~\alpha\in\mf{h}_1,\beta\in\mf{h}_2,~n\in \Z$. Define the operator 
\begin{equation}
    \label{eq: def-eW}
\mathrm{e}^W_\lambda(x,\bar{x})=Z(\mathrm{e}^\lambda,x,\bar{x})x^{-\alpha^W_{\lambda}}\Bar{x}^{-\beta^W_{\lambda}},\quad \lambda\in\Lambda_0 \, .
\end{equation}
This is clearly a well-defined operator on $W$. Using Lemma \ref{lemma:derofZ} we see that 
\begin{equation}
    \frac{\partial}{\partial x}\mathrm{e}^W_\lambda(x,\bar{x})=\frac{\partial}{\partial\bar{x}}\mathrm{e}^W_\lambda(x,\bar{x})=0 \, . 
\end{equation}
Thus $\mathrm{e}^W_\lambda(x,\bar{x})$ is independent of $x,\bar{x}$ and hence we will denote it simply by $\mathrm{e}^W_\lambda$. This gives us the following proposition.
\begin{prop}
The following statements are true: 
\begin{enumerate}
\item For $\lambda \in \Lambda_0$,  $Y_W(\mathrm{e}^\lambda, x,\bar{x})$ can be expressed as
\begin{equation}\label{eq:modveropelamb}
\begin{split}
    Y_{W}({\rm e}^{\lambda},x,\Bar{x})=&\exp\left(\int dx~\alpha^W_{\lambda}(x)^-\right)\exp\left(\int dx~\alpha^W_{\lambda}(x)^+\right)\\ \times &\exp\left(\int d\bar{x}~\beta^W_{\lambda}(\Bar{x})^-\right)\exp\left(\int d\bar{x}~\beta^W_{\lambda}(\Bar{x})^+\right){\rm e}^W_{\lambda} \, x^{\alpha^W_{\lambda}}\Bar{x}^{\beta^W_{\lambda}} \, ,
    \end{split}
\end{equation}
where the exponential operators are defined in \eqref{eq: alphabetapm} which act on $S(\hat{\mf{h}}^-)$ and ${\rm e}^W_{\lambda} \, x^{\alpha^W_{\lambda}}\Bar{x}^{\beta^W_{\lambda}}$ acts on $\Omega_W$.
\item The operator $\mathrm{e}^W_\lambda$ preserves $\Omega_W$, for $\lambda' \in \Lambda_0$ 
\begin{equation}\label{eq:eWonOW}
{\rm e}^W_{\lambda'} \cdot \Omega_W\left(\lambda_\varphi+\lambda\right) \subset \Omega_W\left(\lambda_\varphi+\lambda+\lambda'\right).
\end{equation} 
\end{enumerate}
\begin{proof}
(1) follows from the definition \eqref{eq:defZopera}, \eqref{eq: def-eW} and the commutators \eqref{eq:Zcommalphabeta}. For (2), let $w\in \Omega_W\left(\lambda_\varphi+\lambda\right)$. Since ${\rm e}^W_{\lambda'}$ commutes with $\alpha^W(n),\beta^W(n)$ for all $n\neq 0$ and $\alpha\in\mf{h}_1,\beta\in\mf{h}_2$
\begin{equation}
{\rm e}^W_{\lambda'} \cdot w \in\Omega_W.   
\end{equation}
The next statement follows from \eqref{eq:commalp0Zmn}.
\end{proof}
\end{prop}
\begin{lemma}
For $\lambda=(\alpha,\beta),\lambda'=(\alpha',\beta')\in\Lambda_0$, the following are true 
\begin{equation}\label{eq:commeW}
{\rm e}^W_{\lambda}{\rm e}^W_{\lambda'}=(-1)^{\lambda\circ\lambda'} {\rm e}^W_{\lambda'}{\rm e}^W_{\lambda}~,  
\end{equation}
\begin{equation}\label{eq:xealpW}
\begin{split}
    x^{\alpha} {\rm e}^W_{\lambda' } = x^{\langle \alpha, \alpha' \rangle}{\rm e}^W_{\lambda '}x^{\alpha}~, \\
    \bar{x}^{\beta} {\rm e}^W_{\lambda' } = \bar{x}^{\langle \beta, \beta' \rangle}{\rm e}^W_{\lambda '}\bar{x}^{\beta}~.
\end{split}
\end{equation}
\begin{proof}
Following the calculation in \cite[Section 3.3.1]{Singh:2023mom}, we obtain from \cite[Eq. (3.115)]{Singh:2023mom} \begin{equation}\label{eq:veropprodnormordexpplicit}
\begin{split}
    Y_W(\mathrm{e}^\lambda,x_1,\bar{x}_1)Y_W(\mathrm{e}^{\lambda'},x_2,\bar{x}_2)=\left(x_1-x_2\right)^{\langle\alpha,\alpha'\rangle}&\left(\Bar{x}_1-\Bar{x}_2\right)^{\langle\beta,\beta'\rangle}\\&\times\typecolon Y_W(\mathrm{e}^\lambda,x_1,\bar{x}_1)Y_W(\mathrm{e}^{\lambda'},x_2,\bar{x}_2)\typecolon
    \end{split}
\end{equation}    
where 
\begin{equation}
\begin{split}
     \typecolon Y_W(\mathrm{e}^\lambda,x_1,\bar{x}_1)&Y_W(\rm{e}^{\lambda'},x_2,\bar{x}_2)\typecolon \\& =\exp\left(\int\alpha(x_1)^-\right)\exp\left(\int\alpha'(x_2)^-\right)\exp\left(\int\alpha(x_1)^+\right)\\&\times\exp\left(\int\alpha'(x_2)^+\right)\exp\left(\int\beta(\bar{x}_1)^-\right)\exp\left(\int\beta'(\bar{x}_2)^-\right)\\&\times\exp\left(\int\beta(\bar{x}_1)^+\right)\exp\left(\int\beta'(\bar{x}_2)^+\right){\rm e}^W_\lambda {\rm e}^W_{\lambda'}x_1^{\alpha^W}\Bar{x}_1^{\beta^W} x_2^{\alpha'^W}\Bar{x}_2^{\beta'^W}. 
    \end{split}    
\end{equation}
Similarly
\begin{equation}
\begin{split}
    Y_W(\mathrm{e}^{\lambda'},x_2,\bar{x}_2)Y_W(\mathrm{e}^{\lambda},x_1,\bar{x}_1)=\left(x_2-x_1\right)^{\langle\alpha,\alpha'\rangle}&\left(\Bar{x}_2-\Bar{x}_1\right)^{\langle\beta,\beta'\rangle}\\&\times\typecolon Y_W(\mathrm{e}^{\lambda'},x_2,\bar{x}_2)Y_W(\mathrm{e}^{\lambda},x_1,\bar{x}_1)\typecolon
    \end{split}
\end{equation} 
\begin{equation}
 \begin{split}
     \typecolon Y_W(e^{\lambda'},x_2,\bar{x}_2)&Y_W(e^{\lambda},x_1,\bar{x}_1)\typecolon \\& =\exp\left(\int\alpha(x_1)^-\right)\exp\left(\int\alpha'(x_2)^-\right)\exp\left(\int\alpha(x_1)^+\right)\\&\times\exp\left(\int\alpha'(x_2)^+\right)\exp\left(\int\beta(\bar{x}_1)^-\right)\exp\left(\int\beta'(\bar{x}_2)^-\right)\\&\times\exp\left(\int\beta(\bar{x}_1)^+\right)\exp\left(\int\beta'(\bar{x}_2)^+\right){\rm e}^W_{\lambda'} {\rm e}^W_{\lambda}x_1^{\alpha^W}\Bar{x}_1^{\beta^W} x_2^{\alpha'^W}\Bar{x}_2^{\beta'^W}. 
    \end{split}   
\end{equation}
Now replacing formal variables with complex variables and using the locality property \ref{item:M_locprop}, we obtain the first relation.
The proof of the second relation is exactly the same as the \cite[Eq. (3.111)]{Singh:2023mom} using \eqref{eq:xalponOW} and \eqref{eq:eWonOW}.
\end{proof}
\end{lemma}
\begin{remark}
From the form of the module vertex operator \eqref{eq:modveropelamb}, we see that \eqref{eq:xealpW} implies that 
\begin{equation}\label{eq:xYWalpW}
\begin{split}
    x_1^{\alpha^W} Y_W({\rm e}^{\lambda' },x_2,\bar{x}_2) = x_1^{\langle \alpha, \alpha' \rangle}Y_W({\rm e}^{\lambda '},x_2,\bar{x}_2)x_1^{\alpha^W}, \\
    \bar{x}_1^{\beta^W} Y_W({\rm e}^{\lambda' },x_2,\bar{x}_2) = \bar{x}_1^{\langle \beta, \beta' \rangle}Y_W({\rm e}^{\lambda '},x_2,\bar{x}_2)\bar{x}_1^{\beta^W}.
\end{split}    
\end{equation}
Moreover, from the definition \eqref{eq:defZopera}, we have 
\begin{equation}
\label{eq: xZcom}
    \begin{aligned}
x_1^{\alpha^W} Z(\mathrm{e}^\lambda,x_2,\bar x_2) &= x_1^{\langle\alpha,\alpha_\lambda\rangle}Z(\mathrm{e}^{\lambda},x_2,\bar x_2)x_1^{\alpha^W},\\
\bar x_1^{\beta^W} Z(\mathrm{e}^\lambda,x_2,\bar x_2) &= \bar x_1^{\langle\beta,\beta_\lambda\rangle}Z(\mathrm{e}^{\lambda},x_2,\bar x_2)\bar x_1^{\beta^W}.
\end{aligned}
\end{equation}
\end{remark}
\begin{prop}\label{prop:OWisoClambphi}
Let $\theta\in \mathbb{Z}$
and $\lambda \in \Lambda_0$.
The map $($see \eqref{eq:thetalambhat}$)$
\begin{equation}
    \begin{split}
        &\hat{\Lambda}_0\longrightarrow\mathrm{End}(\Omega_W)\\&\theta{\rm e}_{\lambda}\longmapsto \theta^W{\rm e}^W_{\lambda},
    \end{split}
\end{equation}
where $\theta^W$ acts on  $\Omega_W$ by multiplication with $\theta$,
defines a representation of the central extension $\hat{\Lambda}_0$ of $\Lambda_0$ on $\Omega_W$. Moreover, it is irreducible as both $\mf{h}$ and $\hat{\Lambda}_0$-module and 
\begin{equation}
\Omega_W\cong\C[\lambda_\varphi+\Lambda_0].    
\end{equation} 
\begin{proof}
To show that the map above is a group homomorphism, we need to show that 
\begin{equation}\label{eq:grouplaweWlam}
{\rm e}_{\lambda}{\rm e}_{\lambda'}\longmapsto {\rm e}^W_{\lambda} {\rm e}^W_{\lambda'},\quad \lambda,\lambda'\in\Lambda_0.   
\end{equation}
From \eqref{eq:elambmu}, we see that  
\begin{equation}
    {\rm e}_{\lambda}{\rm e}_{\lambda'}\longmapsto {\rm e}^W_{\lambda+\lambda'}(-1)^{\epsilon(\lambda,\lambda')}.
\end{equation}
So \eqref{eq:grouplaweWlam} is equivalent to showing that 
\begin{equation}\label{eq:eWlam+lamb'}
{\rm e}^W_{\lambda} {\rm e}^W_{\lambda'}={\rm e}^W_{\lambda+\lambda'}(-1)^{\epsilon(\lambda,\lambda')}    .
\end{equation}
We first show that 
\begin{equation}\label{eq:eWalph+alph'}
\begin{split}
&{\rm e}^W_{\alpha_\lambda} {\rm e}^W_{\beta_{\lambda'}}={\rm e}^W_{\alpha_\lambda+\beta_{\lambda'}}(-1)^{\epsilon(\alpha_\lambda,\beta_{\lambda'})}\\&{\rm e}^W_{\alpha_\lambda} {\rm e}^W_{\alpha_{\lambda'}}={\rm e}^W_{\alpha_\lambda+\alpha_{\lambda'}}(-1)^{\epsilon(\alpha_\lambda,\alpha_{\lambda'})}\\&{\rm e}^W_{\beta_\lambda} {\rm e}^W_{\beta_{\lambda'}}={\rm e}^W_{\beta_\lambda+\beta_{\lambda'}}(-1)^{\epsilon(\beta_\lambda,\beta_{\lambda'})}.
\end{split}
\end{equation}
Note that the operators ${\rm e}^W_{\alpha_\lambda},{\rm e}^W_{\beta_{\lambda'}}$ make sense because by definition $\alpha_\lambda\in\Lambda_1^0\subset\Lambda_0$ and 
$\beta_\lambda\in\Lambda_2^0\subset\Lambda_0$.
By definition 
\begin{equation}\label{eq:eWalp+bet}
{\rm e}^W_{\alpha_\lambda+\beta_{\lambda'}}=Z(\mathrm{e}^{\alpha_\lambda+\beta_{\lambda'}},x,\bar{x})x^{-\alpha^W_{\lambda}}\Bar{x}^{-\beta^W_{\lambda'}} \, , 
\end{equation}
and by \eqref{eq:eloweraction}
\begin{equation}
    Y_W({\rm e}^{\alpha_\lambda+\beta_{\lambda'}} ,x,\bar{x})= (-1)^{-\epsilon(\alpha_\lambda,\beta_{\lambda'})}Y_W({\rm e}_{\alpha_\lambda}\cdot{\rm e}^{\beta_{\lambda'}},x,\bar{x}) \, . 
\end{equation}
Now since $\mathrm{e}^{\alpha_\lambda}$ is a chiral vector with 
\begin{equation}
    \mathsf{wt}\,\mathrm{e}^{\alpha_\lambda}=\frac{\langle\alpha_\lambda,\alpha_\lambda\rangle}{2}\in\Z,\quad \overline{\mathsf{wt}}\,\mathrm{e}^{\alpha_\lambda}=0 \, ,
\end{equation}
we can use Lemma \ref{lemma:modwwithverop}. By \eqref{eq:YWpmYWfull} we have 
\begin{equation}\label{eq:Y-4.5-rel}
Y_W({\rm e}_{\alpha_\lambda}\cdot{\rm e}^{\beta_{\lambda'}},x,\bar{x})=Y_W\left({\rm e}^{\alpha_\lambda}, x\right)^{-} Y_W\left({\rm e}^{\beta_{\lambda'}},  \bar{x}\right)+Y_W\left({\rm e}^{\beta_{\lambda'}},  \bar{x}\right) Y_W\left({\rm e}^{\alpha_\lambda}, x\right)^{+} \, ,    
\end{equation}
where we used that 
\begin{equation}
    \mathrm{e}_{\alpha_\lambda}\cdot \mathrm{e}^{\beta_{\lambda'}} = 
x_{\bar{h} - h}(\mathrm{e}^{\alpha_\lambda})\cdot \mathrm{e}^{\beta_{\lambda'}} \, , 
\end{equation}
where $h =\mathsf{wt}\,\mathrm{e}^{\alpha_\lambda} $ and $\bar{h} = \overline{\mathsf{wt}}\,\mathrm{e}^{\alpha_\lambda}=0$. To see the above equation, we compare 
\begin{equation}
    Y_V\left(\mathrm{e}^{\alpha_\lambda}, x\right) \mathrm{e}^{\beta_{\lambda^{\prime}}}=\exp \left(\int d x \alpha_\lambda(x)^{-}\right) \mathrm{e}_{\alpha_\lambda} x^{\left\langle\alpha_\lambda, \beta_{\lambda^{\prime}}\right\rangle} \mathrm{e}^{\beta_{\lambda^{\prime}}}=\exp \left(\int d x \alpha_\lambda(x)^{-}\right) \mathrm{e}_{\alpha_\lambda}  \mathrm{e}^{\beta_{\lambda^{\prime}}} ~,
\end{equation}
with the mode expansion
\begin{equation}
    Y_V\left(\mathrm{e}^{\alpha_\lambda}, x\right) \mathrm{e}^{\beta_{\lambda^{\prime}}}=\sum_{m \in \mathbb{Z}} x_m\left(\mathrm{e}^{\alpha_\lambda}\right) \cdot \mathrm{e}^{\beta_{\lambda^{\prime}}} x^{-m-\left\langle\alpha_\lambda, \alpha_\lambda\right\rangle / 2}~,
\end{equation}
in particular we match the term in the two expansions with no $x$ dependence. Using \eqref{eq:Y-4.5-rel} we have 
\begin{equation}
\begin{split}
Z(\mathrm{e}^{\alpha_\lambda+\beta_{\lambda'}},x,\bar{x})&=(-1)^{-\epsilon(\alpha_\lambda,\beta_{\lambda'})}\left[\exp\left(\sum_{n<0}\frac{\alpha^W_\lambda(n)}{n}x^{-n}\right)\exp\left(\sum_{n<0}\frac{\beta^W_{\lambda'}(n)}{n}\bar{x}^{-n}\right)\right.\\&\quad Y_W({\rm e}_{\alpha_\lambda}\cdot{\rm e}^{\beta_{\lambda'}},x,\bar{x})\left.\exp\left(\sum_{n>0}\frac{\alpha^W_\lambda(n)}{n}x^{-n}\right)\exp\left(\sum_{n>0}\frac{\beta^W_{\lambda'}(n)}{n}\bar{x}^{-n}\right)\right]\\&=(-1)^{-\epsilon(\alpha_\lambda,\beta_{\lambda'})}\left[\exp\left(\sum_{n<0}\frac{\alpha^W_\lambda(n)}{n}x^{-n}\right)\exp\left(\sum_{n<0}\frac{\beta^W_{\lambda'}(n)}{n}\bar{x}^{-n}\right)\right.\\&\quad\left(Y_W\left({\rm e}^{\alpha_\lambda}, x\right)^{-} Y_W\left({\rm e}^{\beta_{\lambda'}},  \bar{x}\right)+Y_W\left({\rm e}^{\beta_{\lambda'}},  \bar{x}\right) Y_W\left({\rm e}^{\alpha_\lambda}, x\right)^{+}\right)\\&\quad\left.\exp\left(\sum_{n>0}\frac{\alpha^W_\lambda(n)}{n}x^{-n}\right)\exp\left(\sum_{n>0}\frac{\beta^W_{\lambda'}(n)}{n}\bar{x}^{-n}\right)\right].  
\end{split}
\end{equation}
Now since 
\begin{equation}
    [\alpha^W_\lambda(n),\beta^W_{\lambda'}(m)]=0 \text{ for all } m,n\in\Z \, ,
\end{equation}
we get 
\begin{equation}
\begin{split}
Z(\mathrm{e}^{\alpha_\lambda+\beta_{\lambda'}},x,\bar{x})&=(-1)^{-\epsilon(\alpha_\lambda,\beta_{\lambda'})}\exp\left(\sum_{n<0}\frac{\alpha^W_\lambda(n)}{n}x^{-n}\right)\left(Y_W\left({\rm e}^{\alpha_\lambda}, x\right)^{-} Z\left({\rm e}^{\beta_{\lambda'}},x,  \bar{x}\right)\right.\\&\quad\left.+Z\left({\rm e}^{\beta_{\lambda'}},  x,\bar{x}\right) Y_W\left({\rm e}^{\alpha_\lambda}, x\right)^{+}\right)\exp\left(\sum_{n>0}\frac{\alpha^W_\lambda(n)}{n}x^{-n}\right)\\&=(-1)^{-\epsilon(\alpha_\lambda,\beta_{\lambda'})}\exp\left(\sum_{n<0}\frac{\alpha^W_\lambda(n)}{n}x^{-n}\right)\left(Y_W\left({\rm e}^{\alpha_\lambda}, x\right)^{-} +Y_W\left({\rm e}^{\alpha_\lambda}, x\right)^{+}\right)\\&\quad\exp\left(\sum_{n>0}\frac{\alpha^W_\lambda(n)}{n}x^{-n}\right)Z\left({\rm e}^{\beta_{\lambda'}},x,  \bar{x}\right)\\&=(-1)^{-\epsilon(\alpha_\lambda,\beta_{\lambda'})}Z\left({\rm e}^{\alpha_{\lambda}},x,  \bar{x}\right)Z\left({\rm e}^{\beta_{\lambda'}},x,  \bar{x}\right) \, , 
\end{split}
\end{equation}
where in the second step we used 
 \eqref{eq:commalphnZmn}. The first relation of \eqref{eq:eWalph+alph'} now follows. To prove the second, first comparing 
\begin{equation}
    Y_V\left(\mathrm{e}^{\alpha_\lambda}, x\right) \mathrm{e}^{\alpha_{\lambda^{\prime}}}=\exp \left(\int d x \alpha_\lambda(x)^{-}\right) \mathrm{e}_{\alpha_\lambda} x^{\left\langle\alpha_\lambda, \alpha_{\lambda^{\prime}}\right\rangle} \mathrm{e}^{\alpha_{\lambda^{\prime}}}~,
\end{equation}
with the mode expansion
\begin{equation}
    Y_V\left(\mathrm{e}^{\alpha_\lambda}, x\right) \mathrm{e}^{\alpha_{\lambda^{\prime}}}=\sum_{m \in \mathbb{Z}} x_m\left(\mathrm{e}^{\alpha_\lambda}\right) \cdot \mathrm{e}^{\alpha_{\lambda^{\prime}}} x^{-m-\left\langle\alpha_\lambda, \alpha_\lambda\right\rangle / 2}~,
\end{equation}
we have 
\begin{equation}
    \mathrm{e}_{\alpha_\lambda}\cdot \mathrm{e}^{\alpha_{\lambda'}} = 
x_{m_0}(\mathrm{e}^{\alpha_\lambda})\cdot \mathrm{e}^{\alpha_{\lambda'}},
\end{equation}
with $m_0 = -\frac{\langle\alpha_\lambda,\alpha_\lambda\rangle}{2} - \langle\alpha_\lambda,\alpha_{\lambda'}\rangle$.
Therefore, using \eqref{eq:M_locality_fields} and using similar analysis as in \eqref{eq:contintYvonw}, we get 
\begin{equation}
  \begin{aligned}
Y_W(\mathrm{e}_{\alpha_\lambda}\cdot \mathrm{e}^{\alpha_{\lambda'}},z_2) &= 
Y_W(x_{m_0}(\mathrm{e}^{\alpha_\lambda})\cdot \mathrm{e}^{\alpha_{\lambda'}},z_2)\\
=&\oint_{C_1^\delta(z_2)} d z_1 \sum_{p\in \mathbb{Z}}
Y_W(x_p(\mathrm{e}^{\alpha_\lambda})\cdot \mathrm{e}^{\alpha_{\lambda'}},z_2)(z_1-z_2)^{-p-1 -\frac{\langle\alpha_\lambda,\alpha_\lambda\rangle}{2} - \langle\alpha_\lambda,\alpha_{\lambda'}\rangle}\\
=& 
\oint_{C^{r_1}_{z_1}} dz_1 Y_W(\mathrm{e}^{\alpha_\lambda},z_1)
Y_W(\mathrm{e}^{\alpha_{\lambda'}},z_2)(z_1 - z_2)^{-\langle \alpha_\lambda,\alpha_{\lambda'}\rangle-1}\\
&-\oint_{C^{r_2}_{z_1}} dz_1 Y_W(\mathrm{e}^{\alpha_{\lambda'}},z_2)
Y_W(\mathrm{e}^{\alpha_{\lambda}},z_1)(z_1 - z_2)^{-\langle \alpha_\lambda,\alpha_{\lambda'}\rangle-1}
\end{aligned}  
\end{equation}
From the definition \eqref{eq: def-eW}, we get
\begin{equation}
\label{eq:eWalalp}
    \begin{aligned}
\mathrm{e}^W_{\alpha_\lambda+\alpha_{\lambda'}}&=Z(\mathrm{e}^{\alpha_\lambda+\alpha_{\lambda'}},z_2,\bar{z}_2)z_2^{-\alpha^W_\lambda-\alpha^W_{\lambda'}}\\
&=(-1)^{-\epsilon(\alpha_\lambda,\alpha_{\lambda'})}\exp\left(\sum_{n<0}\frac{\alpha^W_\lambda(n)+\alpha^W_{\lambda'}(n)}{n}z_2^{-n}\right)\\
&\times Y_W({\rm e}_{\alpha_\lambda}\cdot {\rm e}^{\alpha_{\lambda'}},z_2)\exp\left(\sum_{n>0}\frac{\alpha^W_\lambda(n)+\alpha^W_{\lambda'}(n)}{n}z_2^{-n}\right)z_2^{-\alpha_\lambda^W-\alpha_{\lambda'}^W}\\
& = 
(-1)^{-\epsilon(\alpha_\lambda,\alpha_{\lambda'})}
\exp\left(\sum_{n<0}\frac{\alpha^W_\lambda(n)+\alpha^W_{\lambda'}(n)}{n}z_2^{-n}\right)\\
&\times\bigg[\oint_{C^{r_1}_{z_1}} dz_1 Y_W(\mathrm{e}^{\alpha_\lambda},z_1)
Y_W(\mathrm{e}^{\alpha_{\lambda'}},z_2)(z_1 - z_2)^{-\langle \alpha_\lambda,\alpha_{\lambda'}\rangle-1}\\
&\quad\quad-\oint_{C^{r_2}_{z_1}} dz_1 Y_W(\mathrm{e}^{\alpha_{\lambda'}},z_2)
Y_W(\mathrm{e}^{\alpha_{\lambda}},z_1)(z_1 - z_2)^{-\langle \alpha_\lambda,\alpha_{\lambda'}\rangle-1}\bigg]\\
&\times
\exp\left(\sum_{n>0}\frac{\alpha^W_\lambda(n)+\alpha^W_{\lambda'}(n)}{n}z_2^{-n}\right)z_2^{-\alpha_\lambda^W-\alpha_{\lambda'}^W} \, . 
\end{aligned}
\end{equation}
Let us compute the first part of \eqref{eq:eWalalp}.
\begin{equation}
    \begin{aligned}
     &\exp \left(\sum_{n<0} \frac{\alpha_\lambda^W(n)+\alpha_{\lambda^{\prime}}^W(n)}{n} z_2^{-n}\right)
    \oint_{C_{z_1}^{r_1}} d z_1 Y_W\left(\mathrm{e}^{\alpha_\lambda}, z_1\right) Y_W\left(\mathrm{e}^{\alpha_{\lambda^{\prime}}}, z_2\right)\left(z_1-z_2\right)^{-\left\langle\alpha_\lambda, \alpha_{\lambda^{\prime}}\right\rangle-1}\\
    &
    \times\exp \left(\sum_{n>0} \frac{\alpha_\lambda^W(n)+\alpha_{\lambda^{\prime}}^W(n)}{n} z_2^{-n}\right) z_2^{-\alpha_\lambda^W-\alpha_{\lambda^{\prime}}^W}\\
    =&
    \oint_{C_{z_1}^{r_1}}
     dz_1 
    \exp(\sum_{n<0} \frac{\alpha^W_\lambda
    (n)}{n}
    (z_2^{-n} -z_1^{-n}
    ))
    \exp(\sum_{n<0} 
    \frac{\alpha^W_{\lambda'}(n)}{n} z_2^{-n}
    )
    \exp(\sum_{n<0} 
    \frac{\alpha^W_{\lambda}(n)}{n} z_1^{-n}
    )\\
&\times Y_W(\mathrm{e}^{\alpha_\lambda},z_1)
    \exp(\sum_{n>0} 
    \frac{\alpha^W_{\lambda}(n)}{n} z_1^{-n}
    )
    \exp(-\sum_{n>0} 
    \frac{\alpha^W_{\lambda}(n)}{n} z_1^{-n}
    )
    \exp(-\sum_{n<0} 
    \frac{\alpha^W_{\lambda'}(n)}{n} z_2^{-n}
    )\\
    &\times
    \exp(\sum_{n<0} 
    \frac{\alpha^W_{\lambda'}(n)}{n} z_2^{-n}
    )
    Y_W\left(\mathrm{e}^{\alpha_{\lambda^{\prime}}}, z_2\right)
    \exp(\sum_{n>0} 
    \frac{\alpha^W_{\lambda'}(n)}{n} z_2^{-n}
    )
    \exp(\sum_{n>0} 
    \frac{\alpha^W_{\lambda}(n)}{n} z_2^{-n}
    )\\
    &\times \left(z_1-z_2\right)^{-\left\langle\alpha_\lambda, \alpha_{\lambda^{\prime}}\right\rangle-1} z_2^{-\alpha_\lambda^W-\alpha_{\lambda^{\prime}}^W}\\
    =& 
    \oint_{C_{z_1}^{r_1}} d z_1 \exp \left(\sum_{n<0} \frac{\alpha_\lambda^W(n)}{n}\left(z_2^{-n}-z_1^{-n}\right)\right) \exp \left(\sum_{n<0} \frac{\alpha_{\lambda^{\prime}}^W(n)}{n} z_2^{-n}\right)
    Z(\mathrm{e}^{\alpha_\lambda},z_1) \\
    &\times\exp \left(-\sum_{n>0} \frac{\alpha_\lambda^W(n)}{n} z_1^{-n}\right)\exp \left(-\sum_{n<0} \frac{\alpha_{\lambda^{\prime}}^W(n)}{n} z_2^{-n}\right)
    Z(\mathrm{e}^{\alpha_{\lambda'}},z_2)
    \exp \left(\sum_{n>0} \frac{\alpha_\lambda^W(n)}{n} z_2^{-n}\right)
    \\& \times 
    \left(z_1-z_2\right)^{-\left\langle\alpha_\lambda, \alpha_{\lambda^{\prime}}\right\rangle-1} z_2^{-\alpha_\lambda^W-\alpha_{\lambda^{\prime}}^W}\\
    =& \oint_{C_{z_1}^{r_1}} d z_1 \exp \left(\sum_{n<0} \frac{\alpha_\lambda^W(n)}{n}\left(z_2^{-n}-z_1^{-n}\right)\right)
    Z\left(\mathrm{e}^{\alpha_\lambda}, z_1\right)
    Z\left(\mathrm{e}^{\alpha_{\lambda^{\prime}}}, z_2\right)
    \left(1-\frac{z_2}{z_1}\right)^{\left\langle\alpha_\lambda, \alpha_{\lambda^{\prime}}\right\rangle} \\
    &\times 
    \exp \left(\sum_{n>0} \frac{\alpha_\lambda^W(n)}{n}\left(z_2^{-n}-z_1^{-n}\right)\right)\left(z_1-z_2\right)^{-\left\langle\alpha_\lambda, \alpha_{\lambda^{\prime}}\right\rangle-1} z_2^{-\alpha_\lambda^W-\alpha_{\lambda^{\prime}}^W}\\
    =& 
    \oint_{C_{z_1}^{r_1}} d z_1 \exp \left(\sum_{n<0} \frac{\alpha_\lambda^W(n)}{n}\left(z_2^{-n}-z_1^{-n}\right)\right)
    \exp \left(\sum_{n>0} \frac{\alpha_\lambda^W(n)}{n}\left(z_2^{-n}-z_1^{-n}\right)\right)\\
    &\times(z_1 -z_2)^{-1}
    \left(\frac{1}{z_1}\right)^{\left\langle\alpha_\lambda, \alpha_{\lambda^{\prime}}\right\rangle}Z\left(\mathrm{e}^{\alpha_\lambda}, z_1\right)
    z_1^{-\alpha^W_\lambda} 
    z_1^{\alpha^W_\lambda} Z\left(\mathrm{e}^{\alpha_{\lambda^{\prime}}}, z_2\right)
    z_2^{-\alpha_\lambda^W-\alpha_{\lambda^{\prime}}^W} \\
    =& 
    \mathrm{e}^W_{\alpha_\lambda}  \mathrm{e}^W_{\alpha_{\lambda'}}
    \oint_{C_{z_1}^{r_1}} d z_1 \exp \left(\sum_{n<0} \frac{\alpha_\lambda^W(n)}{n}\left(z_2^{-n}-z_1^{-n}\right)\right) \exp \left(\sum_{n>0} \frac{\alpha_\lambda^W(n)}{n}\left(z_2^{-n}-z_1^{-n}\right)\right)\\
    &\times\left(z_1-z_2\right)^{-1}
    z_1^{\alpha_\lambda^W}
    z_2^{-\alpha_\lambda^W} \, , 
    \end{aligned}
\end{equation}
where we used the definition of $Z$
 \label{eq:defZope} in the second step. The third step is obtained by using \eqref{eq:commalphnZmn} and (see \cite[Section 3.3.1]{Singh:2023mom} for proof) 
\begin{equation}\label{eq:exexalpalp+-}
\begin{split}
    \exp\left(-\int dx_2~\alpha'(x_2)^-\right)\exp\left(\int dx_1~\alpha(x_1)^+\right)&\exp\left(\int dx_2~\alpha'(x_2)^-\right)\\&=\left(1-\frac{x_2}{x_1}\right)^{\langle\alpha,\alpha'\rangle}\exp\left(\int dx_1~\alpha(x_1)^+\right)~.
\end{split}
\end{equation}
The final step is given by using 
\eqref{eq: def-eW}, \eqref{eq: xZcom} 
and the fact that $\mathrm{e}_{\lambda^{\prime}}^W$ commutes with $\alpha^W(n), \beta^W(n)$ for all $n \neq 0$.

Similarly we can write 
\begin{equation}
\begin{split}
&\exp\left(\sum_{n<0}\frac{\alpha^W_\lambda(n)}{n}z_2^{-n}\right)\exp\left(\sum_{n<0}\frac{\alpha^W_{\lambda'}(n)}{n}z_2^{-n}\right)\oint_{C_{z_1}^{r_2}} d z_1 Y_W({\rm e}^{\alpha_\lambda'}, z_2) Y_W\left({\rm e}^{\alpha_{\lambda}}, z_1\right)\left(z_1-z_2\right)^{-1-\langle\alpha_\lambda,\alpha_{\lambda'}\rangle}\\&\times\exp\left(\sum_{n>0}\frac{\alpha^W_\lambda(n)}{n}z_2^{-n}\right)\exp\left(\sum_{n>0}\frac{\alpha^W_{\lambda'}(n)}{n}z_2^{-n}\right)z_2^{-\alpha^W_\lambda-\alpha^W_{\lambda'}}\\&={\rm e}^W_{\alpha_{\lambda'}}{\rm e}^W_{\alpha_{\lambda}}\oint_{C_{z_1}^{r_2}} d z_1~\exp\left(\sum_{n<0}\frac{\alpha^W_{\lambda}(n)}{n}(z_2^{-n}-z_1^{-n})\right)\exp\left(\sum_{n>0}\frac{\alpha^W_{\lambda}(n)}{n}(z_2^{-n}-z_1^{-n})\right)\\&(-1)^{\langle\alpha_\lambda,\alpha_{\lambda'}\rangle}\left(z_1-z_2\right)^{-1} z_1 ^{\alpha^W_{\lambda}} z_2^{-\alpha^W_{\lambda}}.
\end{split}    
\end{equation}
Now using \eqref{eq:commeW} we get
\begin{equation}
\begin{split}
{\rm e}^W_{\alpha_\lambda+\alpha_{\lambda'}}=(-1)^{-\epsilon(\alpha_\lambda,\alpha_{\lambda'})}{\rm e}^W_{\alpha_{\lambda}}{\rm e}^W_{\alpha_{\lambda'}}&\left[\oint_{C_{z_1}^{r_1}} d z_1\right.\exp\left(\sum_{n<0}\frac{\alpha^W_{\lambda}(n)}{n}(z_2^{-n}-z_1^{-n})\right)\\&\exp\left(\sum_{n>0}\frac{\alpha^W_{\lambda}(n)}{n}(z_2^{-n}-z_1^{-n})\right)\left(z_1-z_2\right)^{-1}\left(\frac{z_1}{z_2}\right)^{\alpha^W_{\lambda}}\\-&\oint_{C_{z_1}^{r_2}} d z_1~\exp\left(\sum_{n<0}\frac{\alpha^W_{\lambda}(n)}{n}(z_2^{-n}-z_1^{-n})\right)\\&\left.\exp\left(\sum_{n>0}\frac{\alpha^W_{\lambda}(n)}{n}(z_2^{-n}-z_1^{-n})\right)\left(z_1-z_2\right)^{-1}\left(\frac{z_1}{z_2}\right)^{\alpha^W_{\lambda}}\right]\\=(-1)^{-\epsilon(\alpha_\lambda,\alpha_{\lambda'})}{\rm e}^W_{\alpha_{\lambda}}{\rm e}^W_{\alpha_{\lambda'}}&\left[\oint_{C_{z_1}^{\delta}(z_2)} d z_1~\exp\left(\sum_{n<0}\frac{\alpha^W_{\lambda}(n)}{n}(z_2^{-n}-z_1^{-n})\right)\right.\\&\left.\exp\left(\sum_{n>0}\frac{\alpha^W_{\lambda}(n)}{n}(z_2^{-n}-z_1^{-n})\right)\left(z_1-z_2\right)^{-1}\left(\frac{z_1}{z_2}\right)^{\alpha^W_{\lambda}}\right]\\=(-1)^{-\epsilon(\alpha_\lambda,\alpha_{\lambda'})}{\rm e}^W_{\alpha_{\lambda}}{\rm e}^W_{\alpha_{\lambda'}}.&    
\end{split}    
\end{equation}
Similarly one can prove the other relation in \eqref{eq:eWalph+alph'}. 
Finally to prove \eqref{eq:eWlam+lamb'}, using \eqref{eq:eWalph+alph'} we calculate 
\begin{equation}
\begin{split}
    {\rm e}^W_{\lambda} {\rm e}^W_{\lambda'}&={\rm e}^W_{\alpha_\lambda+\beta_\lambda} {\rm e}^W_{\alpha_{\lambda'}+\beta_{\lambda'}}\\&=   (-1)^{-\epsilon(\alpha_\lambda,\beta_{\lambda})}(-1)^{-\epsilon(\alpha_{\lambda'},\beta_{\lambda'})}{\rm e}^W_{\alpha_\lambda} {\rm e}^W_{\beta_{\lambda}} {\rm e}^W_{\alpha_{\lambda'}} {\rm e}^W_{\beta_{\lambda'}}\\&=(-1)^{-\epsilon(\alpha_\lambda,\beta_{\lambda})}(-1)^{-\epsilon(\alpha_{\lambda'},\beta_{\lambda'})}{\rm e}^W_{\alpha_\lambda}{\rm e}^W_{\alpha_{\lambda'}} {\rm e}^W_{\beta_{\lambda}}  {\rm e}^W_{\beta_{\lambda'}}\\&=(-1)^{-\epsilon(\alpha_\lambda,\beta_{\lambda})}(-1)^{-\epsilon(\alpha_{\lambda'},\beta_{\lambda'})}(-1)^{\epsilon(\alpha_\lambda,\alpha_{\lambda'})}(-1)^{\epsilon(\beta_{\lambda},\beta_{\lambda'})}{\rm e}^W_{\alpha_\lambda+\alpha_{\lambda'}}{\rm e}^W_{\beta_\lambda+\beta_{\lambda'}}\\&=(-1)^{-\epsilon(\alpha_\lambda,\beta_{\lambda})}(-1)^{-\epsilon(\alpha_{\lambda'},\beta_{\lambda'})}(-1)^{\epsilon(\alpha_\lambda,\alpha_{\lambda'})}(-1)^{\epsilon(\beta_{\lambda},\beta_{\lambda'})}(-1)^{\epsilon(\alpha_{\lambda}+\alpha_{\lambda'},\beta_{\lambda}+\beta_{\lambda'})}{\rm e}^W_{\lambda+\lambda'}\\&=(-1)^{\epsilon(\lambda,{\lambda'})}{\rm e}^W_{\lambda+\lambda'} \, ,
\end{split}  
\end{equation}
where we used the fact that ${\rm e}^W_{\alpha_{\lambda'}} {\rm e}^W_{\beta_{\lambda}}={\rm e}^W_{\beta_{\lambda}}{\rm e}^W_{\alpha_{\lambda'}}$ which is true since $\alpha_{\lambda'}\circ\beta_{\lambda}=0$, using \eqref{eq:commeW}. This identity, along with \eqref{eq:eWalph+alph'}, also implies that $(-1)^{\epsilon(\alpha_{\lambda '},\beta_{\lambda} )}  = (-1)^{\epsilon(\beta_{\lambda}, \alpha_{\lambda '} )} $ We also used the $\Z$-bilinearity of $\epsilon$. The irreducibility of $\Omega_W$ and the isomorphism $\Omega_W\cong\C[\lambda_\varphi+\Lambda_0]$ can be proved using the arguments in the proof of \cite[Theorem 5.2]{Singh:2023mom}. 
\end{proof}
\end{prop}
\begin{prop}\label{prop:vergenratedbyalp}
All vertex operators $Y_W(u,x,\Bar{x})$ are generated by $\alpha^W(x),\beta^W(\bar{x})$ and $Y_W({\rm e}^\lambda,x,\bar{x})$ for $(\alpha,\beta)\in\mf{h},\lambda\in\Lambda_0$. More precisely, $u^W_{m,n}\in\mathrm{End}(W)$ is generated by $\alpha^W(m),\beta^W(m),(\mathrm{e}^{\lambda})^W_{m,n}$ for $m,n\in\Z,(\alpha,\beta)\in\mf{h},\lambda\in\Lambda_0.$ 
\begin{proof}
This follows from the formula \eqref{eq:M_locality_fields} and the fact that $V(0)$ is spanned by elements of the form \eqref{eq:genvect}.  
\end{proof}
\end{prop}
\noindent\textit{Proof of Theorem \ref{thm:classllvoamod}:} By Lemma \ref{lemma:Wisoheisrep} and Proposition \ref{prop:OWisoClambphi}, we see that $W\cong V(\lambda_\varphi)$  where $\lambda_\varphi\in\Lambda_0^\circ$ (see \eqref{eq:lambphi} and Subsection \ref{sec:moduleduallattice} for notations). Moreover the action of $\alpha^W(x),\beta^W(\bar{x})$ and $Y_W({\rm e}^\lambda,x,\bar{x})$ for $(\alpha,\beta)\in\mf{h},\lambda\in\Lambda_0$ on $W$ is same as the action of $\alpha(x),\beta(\bar{x})$ and $Y({\rm e}^\lambda,x,\bar{x})$ on $V(\lambda_\varphi)$. Thus by Proposition \ref{prop:vergenratedbyalp}, the action of $Y_{W}(u,x,\bar{x})$  on $W$ is same as the action of $Y(u,x,\bar{x})$ on $V(\lambda_\varphi)$ for every $u\in V(0)$.\\\\
\section{Intertwining Operators Between Modules Of The LLVOA}\label{sec:intertwiners}
In this section, we classify the intertwining operators between various irreducible modules of the LLVOA. 
\begin{thm}\label{thm:interfusion}
Let $[\mu],[\nu],[\rho] \in \Lambda_0^\circ / \Lambda_0$ be three equivalence classes and suppose that $\mu+\nu\in\Lambda_0^\circ$. If $[\rho] \neq[\mu+\nu]$, then $V_{\mu \nu}^\rho=0$, or equivalently $\mathcal{N}_{\mu \nu}^\rho=0$.
\end{thm}
\begin{proof}
We prove the theorem in 2 steps.\\
\textit{Step 1}: Let $\s{Y} \in \s{V}_{\mu \nu}^\rho$ be an intertwining operator 
for $[\rho] \neq [\mu+\nu]$. We first prove that 
\begin{equation}
\s{Y}\left({\rm e}^\mu, x, \bar{x}\right) {\rm e}^\nu=0 \text {. }    
\end{equation}
% The locality and duality of $\s{Y}(w_{(\mu)}, z_2, \bar{z}_2)$ with the vertex operator $\s{Y}\left(\alpha(-1) \otimes \textbf{1}, z_1\right),$ for $ w_{(\mu)}\in V(\mu),\alpha \in \mf{h}_1$ implies that 
% \begin{equation}
% \begin{aligned}
% Y\left(\alpha(-1) \otimes \textbf{1}, z_1\right) \s{Y}\left(w_{(\mu)}, z_2, \bar{z}_2\right) & =\s{Y}\left(w_{(\mu)}, z_2, \bar{z}_2\right) Y\left(\alpha(-1) \otimes \textbf{1}, z_1\right) \\
% & =\s{Y}\left(Y\left(\alpha(-1) \otimes \textbf{1}, z_1-z_2\right) w_{(\mu)}, z_2, \bar{z}_2\right)~.
% \end{aligned}
% \end{equation}
% \fixme{The equalities hold in different regions. $Y$ represents different module vertex operator.}
Following the proof of Lemma \ref{lemma:modwwithverop}, we get
\begin{equation}
\left[\alpha(0), \s{Y}\left(w_{(\mu)}, z, \bar{z}\right)\right]=\s{Y}\left(\alpha(0) \cdot w_{(\mu)}, z, \bar{z}\right), \quad \alpha \in \mf{h}_1, \quad w_{(\mu)} \in V(\mu)~.
\end{equation}
Similarly
\begin{equation}
\left[\beta(0), \s{Y}\left(w_{(\mu)}, z, \bar{z}\right)\right]=\s{Y}\left(\beta(0) \cdot w_{(\mu)}, z, \bar{z}\right), \quad \beta \in \mf{h}_2,\quad w_{(\mu)} \in V(\mu)~ .
\end{equation}
Using \eqref{eq:alp0act}, we get
\begin{equation}
\begin{aligned}
\alpha(0) \cdot \s{Y}\left({\rm e}^\mu, z, \bar{z}\right) {\rm e}^\nu & =\s{Y}\left({\rm e}^\mu, z, \bar{z}\right) \alpha(0) \cdot {\rm e}^\nu+\s{Y}\left(\alpha(0) \cdot {\rm e}^\mu, z, \bar{z}\right) {\rm e}^\nu \\
& =\left\langle\alpha, \alpha_\nu\right\rangle \s{Y}\left({\rm e}^\mu, z, \bar{z}\right) {\rm e}^\nu+\left\langle\alpha, \alpha_\mu\right\rangle \s{Y}\left({\rm e}^\mu, z, \bar{z}\right) {\rm e}^\nu \\
& =\left\langle\alpha, \alpha_{\mu+\nu}\right\rangle \s{Y}\left({\rm e}^\mu, z, \bar{z}\right) {\rm e}^\nu .
\end{aligned}
\end{equation}
Similarly
\begin{equation}
\beta(0)\cdot \s{Y}\left({\rm e}^\mu, z, \bar{z}\right) {\rm e}^\nu=\left\langle\beta, \beta_{\mu+\nu}\right\rangle \s{Y}\left({\rm e}^\mu, z, \bar{z}\right) {\rm e}^\nu \text {. }
\end{equation}
This implies that
\begin{equation}
\s{Y}\left({\rm e}^\mu, x, \bar{x}\right) {\rm e}^\nu \subset \left(S(\hat{\mf{h}}^{-}) \otimes \mathbb{C}\left[\mu+\nu+\Lambda_0\right]\right)\{x, \bar{x}\}~ .
\end{equation}
This implies that
\begin{equation}
[\rho]=[\mu+\nu]~,
\end{equation}
contradicting the assumption. This implies that
\begin{equation}
\s{Y}\left({\rm e}^\mu, x, \bar{x}\right) {\rm e}^\nu=0 \text {. }
\end{equation}
\textit{Step 2}: If $\s{Y}\left({\rm e}^\mu, x, \bar{x}\right) {\rm e}^\nu=0$, then we have
\begin{equation}
   \s{Y}(\cdot,x,\bar{x})\equiv 0. 
\end{equation}
Using locality and duality of $\s{Y}$, we get
\begin{equation}
\s{Y}\left({\rm e}^\mu, x, \bar{x}\right) Y\left(v_1, x_1,\bar{x}_1\right)\cdots
Y\left(v_n, x_n,\bar{x}_n\right)
{\rm e}^\nu=0 \quad \text { for any } v_i \in V(0)~,
\end{equation}
which, using irreducibility of $V(\nu)$ implies that
\begin{equation}
\s{Y}\left({\rm e}^\mu, x, \bar{x}\right)=0 \text {. }
\end{equation}
Similarly
\begin{equation}
\begin{aligned}
\s{Y}\left(Y\left(v, z_1-z_2, \bar{z}_1-\bar{z}_2\right) \cdot {\rm e}^\mu, z_2, \bar{z}_2\right) {\rm e}^\nu & =Y(v, z_1, \bar{z}_1)\s{Y}\left({\rm e}^\mu, z_2, \bar{z}_2\right)  {\rm e}^\nu=0~,
\end{aligned}
\end{equation}
which, again using irreducibility of $V(\mu)$ implies
\begin{equation}
\s{Y}(\cdot, x, \bar{x}) \equiv 0 \text {. }
\end{equation}
This proves that $\s{V}_{\mu \nu}^\rho=0$.
\end{proof}
\begin{thm}\label{thm:Nmnp<=1}
Suppose $\mu, \nu \in \Lambda_0^\circ$ are such that $\mu+\nu \in \Lambda_0^\circ$. Then $\mathcal{N}_{\mu \nu}^{\mu+\nu} \leq 1$.
\end{thm}
\begin{proof}
We prove this theorem in three steps.\\
\textit{Step 1}: An intertwining operator $\s{Y} \in \s{V}_{\mu \nu}^{\mu+\nu}$ has an expansion of the form 
\begin{equation}\label{eq:intexp}
\s{Y}(w_{(\mu)}, x, \bar{x})=\sum_{m, n \in \mathbb{Z}} \s{Y}_{m, n}(w_{(\mu)}) x^{-m+\left\langle\alpha_\mu, \alpha_\nu\right\rangle} \bar{x}^{-n+\left\langle\beta_\mu, \beta_\nu\right\rangle}~,\quad w_{(\mu)} \in V(\mu)_{(h, \bar{h})}~.
\end{equation}
To prove this, we first start with the expansion 
\begin{equation}
\s{Y}(w_{(\mu)}, x, \bar{x})=\sum_{m, n \in \mathbb{C}} \s{Y}_{m, n}(w_{(\mu)}) x^{-m+\left\langle\alpha_\mu, \alpha_\nu\right\rangle} \bar{x}^{-n+\left\langle\beta_\mu, \beta_\nu\right\rangle}~,\quad w_{(\mu)} \in V(\mu)_{(h, \bar{h})}~.
\end{equation}
The $L(0)$-property now implies that
\begin{equation}\label{eq:wtsYmn}
\begin{split}
\mathsf { wt }~ \s{Y}_{m, n}\left(w_{(\mu)}\right)=\mathsf { wt }~w_{(\mu)}-m+\left\langle\alpha_\mu, \alpha_\nu\right\rangle, \\ \overline{\mathsf{w t}}~ \s{Y}_{m, n}\left(w_{(\mu)}\right)=\overline{\mathsf{w t}}~w_{(\mu)} -n+\left\langle\beta_\mu, \beta_\nu\right\rangle \text {. }
\end{split}
\end{equation}
Note that since $\mu+\nu \in \Lambda_0^\circ$,
\begin{equation}
\mu\circ \nu=\left\langle\alpha_\mu, \alpha_\nu\right\rangle-\left\langle\beta_\mu, \beta_\nu\right\rangle \in \mathbb{Z}~,
\end{equation}
and hence single-valuedness of $V(\mu+\nu)$ implies that
\begin{equation}
\s{Y}_{m, n}\left(w_{(\mu)}\right)=0 \text {, if } m-n \notin \mathbb{Z} \, .
\end{equation}
Notice that for any $[\rho]\in\Lambda^\circ_0$, the module $V(\rho)$ is actually graded by 
\begin{equation}\label{eq:gradVrho}
    \left(\frac{\langle\alpha_\rho,\alpha_\rho\rangle}{2},\frac{\langle\beta_\rho,\beta_\rho\rangle}{2}\right)+(\Z \times \Z)~.
\end{equation} 
Take arbitrary homogeneous vectors $w_{(\mu)} \in V(\mu)$ 
and 
$w_{(\nu)} \in V(\nu)$,
we have
\begin{equation}
\begin{split}
\mathsf { wt }~w_{(\mu)}=h_\mu+\frac{\left\langle\alpha_\mu, \alpha_\mu\right\rangle}{2},\quad \overline{\mathsf{w t}}~w_{(\mu)}=\bar{h}_{\mu}+\frac{\left\langle\beta_\mu, \beta_\mu\right\rangle}{2}~,\\
\mathsf { wt }~w_{(\nu)}=h_\nu+\frac{\left\langle\alpha_\nu, \alpha_\nu\right\rangle}{2},\quad \overline{\mathsf{w t}}~w_{(\nu)}=\bar{h}_{\nu}+\frac{\left\langle\beta_\nu, \beta_\nu\right\rangle}{2}~,
\end{split}
\end{equation}
for some $h_\mu,\bar{h}_\mu,h_\nu,\bar{h}_\nu\in\Z$. Then since  
\begin{equation}
    \s{Y}_{m, n}\left(w_{(\mu)}\right) \cdot w_{(\nu)} \in V(\mu+\nu)~,
\end{equation}
we see that 
\begin{equation}
\label{eq: weight Y.w}
\begin{aligned}
& \mathsf { wt }~ \s{Y}_{m, n}\left(w_{(\mu)}\right) \cdot w_{(\nu)}=h_\mu+h_\nu+\frac{\left\langle\alpha_\nu, \alpha_\nu\right\rangle}{2}+\frac{\left\langle\alpha_\mu, \alpha_\mu\right\rangle}{2}+\left\langle\alpha_\mu, \alpha_\nu\right\rangle-m\in\frac{\langle\alpha_{\mu+\nu},\alpha_{\mu+\nu}\rangle}{2}+\Z, \\
& \overline{\mathsf{w t}}~ \s{Y}_{m, n}\left(w_{(\mu)}\right) \cdot w_{(\nu)}=\bar{h}_\mu+\bar{h}_{\nu}+\frac{\left\langle\beta_\nu, \beta_\nu\right\rangle}{2}+\frac{\left\langle\beta_\mu, \beta_\mu\right\rangle}{2}+\left\langle\beta_\mu, \beta_\nu\right\rangle-n \in \frac{\langle\beta_{\mu+\nu},\beta_{\mu+\nu}\rangle}{2}+\Z~.
\end{aligned}
\end{equation}

This implies that $m,n\in\Z$.\\\\
\textit{Step 2}: For an intertwining operator $\s{Y} \in \s{V}_{\mu \nu}^{\mu+\nu}$, if $\s{Y}_{0,0}({\rm e}^\mu)\cdot {\rm e}^\nu=0$ then $\s{Y}(\cdot,x,\bar{x})\equiv 0$.\\\\
By Step 1, $\s{Y} \in \s{V}_{\mu \nu}^{\mu+\nu}$ can be expanded as in \eqref{eq:intexp}. 
%From \eqref{eq:gradVrho} we see that
%\begin{equation}
%\mathsf { wt }~ w_{(\mu+\nu)} \geq \frac{\left\langle\alpha_\mu+\alpha_\nu, \alpha_\mu+\alpha_\nu\right\rangle}{2} \text {, } \quad\overline{\mathsf{w t}}~ w_{(\mu+\nu)} \geq \frac{\left\langle\beta_\mu+\beta_\nu, \beta_\mu+\beta_\gamma\right\rangle}{2}~.
%\end{equation}
From 
\eqref{eq: weight Y.w}
we see that 
\begin{equation}\label{eq:wtymnemuenu}
\begin{aligned}
& \mathsf { wt }~ \s{Y}_{m, n}({\rm e}^{\mu}) \cdot {\rm e}^{\nu}=\frac{\left\langle\alpha_\nu, \alpha_\nu\right\rangle}{2}+\frac{\left\langle\alpha_\mu, \alpha_\mu\right\rangle}{2}+\left\langle\alpha_\mu, \alpha_\nu\right\rangle-m~,\\
& \overline{\mathsf{w t}}~ \s{Y}_{m, n}({\rm e}^{\mu}) \cdot {\rm e}^{\nu}=\frac{\left\langle\beta_\nu, \beta_\nu\right\rangle}{2}+\frac{\left\langle\beta_\mu, \beta_\mu\right\rangle}{2}+\left\langle\beta_\mu, \beta_\nu\right\rangle-n ~.
\end{aligned}
\end{equation}
But since 
\begin{equation}
    \begin{aligned}
& \mathsf { wt }~ \s{Y}_{m, n}({\rm e}^{\mu}) \cdot {\rm e}^{\nu}\geq\frac{\left\langle\alpha_\nu, \alpha_\nu\right\rangle}{2}+\frac{\left\langle\alpha_\mu, \alpha_\mu\right\rangle}{2}+\left\langle\alpha_\mu, \alpha_\nu\right\rangle~,\\
& \overline{\mathsf{w t}}~ \s{Y}_{m, n}({\rm e}^{\mu}) \cdot {\rm e}^{\nu}\geq\frac{\left\langle\beta_\nu, \beta_\nu\right\rangle}{2}+\frac{\left\langle\beta_\mu, \beta_\mu\right\rangle}{2}+\left\langle\beta_\mu, \beta_\nu\right\rangle ~.
\end{aligned}
\end{equation}
Thus we conclude that
\begin{equation}\label{eq:ymnemuenumn>0}
\s{Y}_{m, n}\left({\rm e}^\mu\right)\cdot {\rm e}^\nu=0, \quad m>0,\text{ or } n>0~.
\end{equation}
We now show that if
\begin{equation}
\s{Y}_{0,0}\left({\rm e}^\mu\right) \cdot {\rm e}^\nu=0 \text {, }
\end{equation}
then 
\begin{equation}\label{eq:ymn0mn<0}
    \s{Y}_{m, n}\left({\rm e}^\mu\right)\cdot {\rm e}^\nu=0, \quad m, n<0~.
\end{equation}
We will show, using induction\footnote{We are using the following version of induction: suppose we want to prove a statement $P(m,n)$ for all $m,n\in\Z_{\geq 0}$. Then proof by induction goes as follows: 
\begin{itemize}
    \item $P(0,0)$ is true.
    \item If $P(p,q)$ is true for all $p<m,q\leq n$, then $P(m,n)$ is true.
    \item If $P(p,q)$ is true for all $p\leq m,q<n$, then $P(m,n)$ is true.
\end{itemize}
Then $P(m,n)$ is true for all $m,n \in\Z_{\geq 0}$.} on $m,n$. Suppose that for some $m,n \leq 0$
\begin{equation}
\s{Y}_{r,s}\left({\rm e}^\mu\right) \cdot {\rm e}^\nu=0 \text {, for all } r>m,s\geq n~.    
\end{equation}
We want to show that 
\begin{equation}
\s{Y}_{m,n}\left({\rm e}^\mu\right) \cdot {\rm e}^\nu=0~.    
\end{equation}
Using contour integrals as in the proof of Lemma \ref{lemma:modwwithverop}, we have, using locality and duality of intertwining operators, 
\begin{equation}
[\alpha(p), \s{Y}(w_{(\mu)}, x, \bar{x})]=\sum_{r=0}^p\binom{p}{r} \s{Y}(\alpha(r) \cdot w_{(\mu)}, x, \bar{x}) x^{p-r}~,\quad\alpha\in\mf{h}_1~,
\end{equation}
which implies for $p>0$,
\begin{equation}
[\alpha(p), \s{Y}({\rm e}^\mu, x, \bar{x})]=\langle\alpha,\alpha_\mu\rangle\s{Y}({\rm e}^\mu, x, \bar{x})x^p~,\quad\alpha\in\mf{h}_1~.    
\end{equation}
Thus we get 
\begin{equation}\label{eq:commaqsYmn}
 [\alpha(p), \s{Y}_{m, n}\left({\rm e}^\mu\right)]=\langle\alpha,\alpha_\mu\rangle\s{Y}_{m+p, n}\left({\rm e}^\mu\right)~,\quad\alpha\in\mf{h}_1~.   
\end{equation}
Similarly we have 
\begin{equation}\label{eq:commbqsYmn}
 [\beta(q), \s{Y}_{m, n}\left({\rm e}^\mu\right)]=\langle\beta,\beta_\mu\rangle\s{Y}_{m, n+q}\left({\rm e}^\mu\right)~,\quad\beta\in\mf{h}_2~, \end{equation}
 Using \eqref{eq:commaqsYmn},
 % \eqref{eq:commbqsYmn} and the fact that
 % \begin{equation}
 %     [\alpha(p),\beta(q)]=0~
 % \end{equation}
 % , 
 we get 
\begin{equation}
\begin{split}
\alpha(p)\s{Y}_{m, n}({\rm e}^\mu)\cdot{\rm e}^\nu = \s{Y}_{m, n}({\rm e}^\mu)\alpha(p) \cdot{\rm e}^\nu+\langle\alpha,\alpha_\mu\rangle\s{Y}_{m+p, n}\left({\rm e}^\mu\right)\cdot{\rm e}^\nu~,
\\
\beta(q)\s{Y}_{m, n}({\rm e}^\mu)\cdot{\rm e}^\nu = \s{Y}_{m, n}({\rm e}^\mu)\beta(q) \cdot{\rm e}^\nu+\langle\beta,\beta_\mu\rangle\s{Y}_{m, n+q}\left({\rm e}^\mu\right)\cdot{\rm e}^\nu~.
\end{split}
\end{equation}
In particular,
\begin{equation}\label{eq:aoboymn}
\begin{split}
     \alpha(0)\s{Y}_{m, n}({\rm e}^\mu)\cdot{\rm e}^\nu =\langle\alpha,\alpha_{\mu+\nu}\rangle\s{Y}_{m, n}({\rm e}^\mu)\cdot{\rm e}^\nu,\quad\alpha\in\mf{h}_1~, \\\beta(0)\s{Y}_{m, n}({\rm e}^\mu)\cdot{\rm e}^\nu =\langle\beta,\beta_{\mu+\nu}\rangle\s{Y}_{m, n}({\rm e}^\mu)\cdot{\rm e}^\nu,\quad \beta\in\mf{h}_2 ~.
\end{split}    
\end{equation}
By induction hypothesis, we get 
\begin{equation}
\label{eq: apYe}
\begin{split}
  \alpha(p)\s{Y}_{m, n}({\rm e}^\mu)\cdot{\rm e}^\nu = 0~,\quad p>0~,\quad\alpha\in\mf{h}_1~,\\  \beta(q)\s{Y}_{m, n}({\rm e}^\mu)\cdot{\rm e}^\nu = 0~,\quad q\geq0~,\quad \beta\in\mf{h}_2~.
\end{split}   
\end{equation}
This means that 
\begin{equation}
   \s{Y}_{m, n}({\rm e}^\mu)\cdot{\rm e}^\nu\in\C\cdot{\rm e}^{\mu+\nu}~.
\end{equation}
But this contradicts \eqref{eq:wtymnemuenu}. Thus we conclude that $\s{Y}_{m, n}({\rm e}^\mu)\cdot{\rm e}^\nu=0$.
% Next, assuming that for some
% \fixme{From (5.39) to (5.40) is redundant}
% $m,n<0$, 
% \begin{equation}
% \s{Y}_{r,s}\left({\rm e}^\mu\right) \cdot {\rm e}^\nu=0 \text {, for all } r>m,s\geq n~.    
% \end{equation}
% We can prove that 
% \begin{equation}
% \s{Y}_{m,n}\left({\rm e}^\mu\right) \cdot {\rm e}^\nu=0~.    
% \end{equation}
% Thus by induction, 
% \begin{equation}\label{eq:ymnemuenumn<0}
% \s{Y}_{m, n}({\rm e}^\mu)\cdot{\rm e}^\nu=0,\quad \text{for all } m,n<0~.    
% \end{equation}
Combining \eqref{eq:ymnemuenumn>0} and \eqref{eq:ymn0mn<0}, we get 
\begin{equation}
\s{Y}({\rm e}^\mu,x,\bar{x})\cdot{\rm e}^\nu=0.    
\end{equation}
Then applying Step 2 of Theorem \ref{thm:interfusion} completes the proof of this step. \\\\
\textit{Step 3}: If $\s{Y},\s{Y}'\in \s{V}_{\mu \nu}^{\mu+\nu}$, then there exists $c\in\C$ such that $\s{Y}=c\s{Y}'$.\\\\ 
Expand $\s{Y},\s{Y}'$ as in \eqref{eq:intexp}. Then by 
\eqref{eq:aoboymn} and \eqref{eq: apYe} for $m, n = 0$, we see that 
\begin{equation}
\s{Y}_{0, 0}({\rm e}^\mu)\cdot{\rm e}^\nu,\quad \s{Y}'_{0, 0}({\rm e}^\mu)\cdot{\rm e}^\nu\in\C\cdot{\rm e}^{\mu+\nu}~.    
\end{equation}
Thus there exists $c\in\C$ such that
\begin{equation}
\s{Y}_{0, 0}({\rm e}^\mu)\cdot{\rm e}^\nu=c\s{Y}'_{0, 0}({\rm e}^\mu)\cdot{\rm e}^\nu~.
\end{equation}
for some constant $c \in \mathbb{C}$.
Then
\begin{equation}
\left(\s{Y}-c \s{Y}^{\prime}\right)_{0,0}\left({\rm e}^\mu\right) \cdot {\rm e}^\nu=\s{Y}_{0,0}\left({\rm e}^\mu\right) {\rm e}^\nu-c \s{Y}_{0,0}^{\prime}\left({\rm e}^\mu\right) \cdot {\rm e}^\nu=0 ~.
\end{equation}
Thus by step 2, we have 
\begin{equation}
\s{Y}\equiv c  \, \s{Y}^{\prime} \, .
\end{equation}
Thus
\begin{equation}
\s{N}_{\mu \nu}^{\mu+\nu}=\mathsf{dim} \s{V}_{\mu \nu}^{\mu+\nu} \leq 1 .
\end{equation}
\end{proof}
We finally have the following theorem.
\begin{thm}\label{thm: interclass}
$\s{N}_{\mu\nu}^{\rho}=\delta_{\mu+\nu}^\rho$.    
\end{thm}
\begin{proof}
The operator
\begin{equation}
    \s{Y}(w_{(\mu)},x,\bar{x})=\typecolon \prod_{r=1}^{k}\prod_{s=1}^{\bar{k}}\left(\frac{1}{(m_r-1)!}\frac{d^{m_r - 1 }\alpha_{r}(x)}{dx^{m_r-1}}\right)\left(\frac{1}{(\bar{m}_s-1)!}\frac{d^{\bar{m}_s-1}\beta_{s}(\bar{x})}{d\bar{x}^{\bar{m}_s-1}}\right) \s{Y}({\rm e}^\lambda,x,\bar{x})\typecolon~,
\end{equation}
for 
\begin{equation}
   w_{(\mu)}=\left( \alpha_{1}(-m_1)\cdot\alpha_{2}(-m_2)\cdots\alpha_{k}(-m_k) \cdot \beta_{1}(-\bar{m}_1)\cdot\beta_{2}(-\bar{m}_2)\cdots\beta_{\bar{k}}(-\bar{m}_{\bar{k}}) \right)  \otimes {\rm e}^{\lambda},\quad\lambda\in\mu+\Lambda_0~,  
\end{equation}
with $\s{Y}({\rm e}^\lambda,x,\bar{x})$ defined similar to \eqref{eq:veropelamb}
is a non-zero intertwining operator of type ${\mu+\nu\choose\mu~\nu}$. The proof of axioms of intertwining operators is a straightforward computation using the calculations in \cite{Singh:2023mom}. From Theorem \ref{thm:interfusion} and Theorem \ref{thm:Nmnp<=1} we conclude the result.    
\end{proof}
\section{MTC Associated To Rational Modular Invariant LLCFT}\label{sec:mtc}
In this section, we find the MTC corresponding to the left and right moving sectors of a rational modular invariant LLCFT (see Section \ref{sec:rationality_LLCFT}). We also give some explicit examples.

\subsection{Construction Of The MTC}\label{sec:mtc_left-right_LLCFT}

Let $\{(V(\mu),Y_{\mu})\}_{\mu\in\Lambda/\Lambda_0}$ 
be a modular invariant rational LLCFT. The left and right moving chiral CFTs of the LLCFT are the VOAs based on Euclidean lattices $\Lambda_1^0$ and $\Lambda_2^0$ along with the category of modules. We will denote the left and right moving VOAs by $V^L$ and $V^R$ respectively. Let $\mathcal{C}_{V^L},\mathcal{C}_{V^R}$ denote the category with objects being $V^L,V^R$-modules respectively and morphisms being VOA module homomorphisms . It is easy to see that these categories are abelian categories: the proof of this fact is entirely analogous to the proof for the category of vector spaces and linear maps. In fact these categories can be endowed with a tensor bifunctor \cite{huang1995theory1,huang1995theory2,huang1995theory3,huang1995theory4,HLZ1,HLZ2,HLZ3,HLZ4,HLZ5,HLZ6,HLZ7} and rigid structure, braiding structure and twists making it into a modular tensor category \cite{huang2005vertex1,huang2008rigidity,huang2008vertex2,HLZ8}. 
\par Let us briefly recall the notion of modular tensor category, we refer the reader to \cite{turaev1994quantum,bakalov2001lectures} for the detailed definition. A modular tensor category is the following data:
\begin{enumerate}
\item A category $\mathcal{C}$ endowed with a tensor bifunctor 
\begin{equation}
    \boxtimes:\mathcal{C}\times\mathcal{C}\longrightarrow\mathcal{C}~,
\end{equation}
a \textit{unit object} $V$, three natural isomorphisms: for objects $W,W_1,W_2,W_3$ in the category 
\begin{equation}
\begin{split}
    \mathcal{A}:W_1\boxtimes(W_2\boxtimes W_3)&\longrightarrow (W_1\boxtimes W_2)\boxtimes W_3~,\\
    \ell_W:V\boxtimes W&\longrightarrow W~,\\
    r_W:W\boxtimes V&\longrightarrow W~,
    \end{split}
\end{equation}
called respectively the \textit{associativity, left unit} and \textit{right unit} isomorphisms, 
which satisfy some coherence conditions. This is the data for a \textit{monoidal category}. 
\item A natural isomorphism called \textit{braiding} isomorphism: for objects $W_1,W_2$
\begin{equation}
    \mathcal{B}_{W_1,W_2}:W_1\boxtimes W_2\longrightarrow W_2\boxtimes W_1~,
\end{equation}
compatible with the associativity, left and right unit isomorphism. This is the additional data for a \textit{braided tensor category}.
\item For every object $W$, there is a unique \textit{left dual} \underline{or} \textit{right dual} object $W^*,{}^*W$ respectively along with two natural morphisms called \textit{creation} and \textit{annihilation} morphisms:
\begin{equation}
 \begin{split}
e_W: W^* \boxtimes W \longrightarrow V,&\quad i_W: V \longrightarrow W \boxtimes W^* ~,\\ 
e_W^{\prime}: W \boxtimes{ }^* W \longrightarrow V,&\quad i_W^{\prime}: V \longrightarrow{ }^* W \boxtimes W~,
 \end{split}   
\end{equation}
compatible with associativity isomorphism. This is called the called \textit{rigid} structure and is the additional data for a \textit{rigid braided tensor category}. 
%One can define trace of morphisms in a rigid braided tensor category.
\item A natural isomorphism is called the \textit{twist map}: for $W$ an object, the twist map is an isomorphism $\theta_W:W\longrightarrow W$ satifying the \textit{balancing axioms}, i.e., compatible with the braiding isomorphism and (right \underline{or} left dual) rigid structure. This is the additional data for a \textit{ribbon category}. 
\item $\mathcal{C}$ is an $\textbf{Vect}_{\C}$-enriched \footnote{In general, one defines $\textbf{Vect}_{\mathbb{F}}$-enriched category where the set of morphisms are $\mathbb{F}$-vector spaces. For the application to VOAs, we restrict to $\mathbb{F}=\C$.} category, i.e., the set of morphisms between objects has a complex vector space structure bilinear with respect to composition, every morphism has a kernal and cokernal, every monomorphism is the kernel of some morphism and every epimorphism is the co-kernel of some morphism and the direct sum of objects is defined. In an abelian category, one can define the notion of \textit{simple objects}. Then $\mathcal{C}$ is assumed to be a \textit{semisimple} category, i.e., every object is isomorphic to the direct sum of a finitely many simple objects. 
\item $\mathcal{C}$ is a \textit{fusion category}, i.e., it is a semisimple $\textbf{Vect}_{\C}$-enriched rigid braided category with ribbon structure, with finitely many isomorphism classes of simple objects (the unit object being one of them), the addition of morphisms bilinear with respect to the tensor bifunctor and the vector space of morphisms from unit object to itself is isomorphic to $\C$. 
\item Let $\{W_a\}_{a\in\mathscr{A}}$ be a complete list of representatives of isomorphism class of simple objects in $\mathcal{C}$ with $W_0=V$. Then the $S$-matrix defined by 
\begin{equation}
\begin{aligned}
[S_{\mathsf{MTC}}]_{ab}:=&\bigg[ \ell_{V} \circ (e_{W_a} \otimes e_{W_b^{*}}) \circ (\mathds{1}_{W_a^{*}} \otimes \mathcal{B}_{W_aW_b} \otimes \mathds{1}_{W_b^{*}}) \\&\quad\circ(\mathds{1}_{W_a^{*}} \otimes B_{W_aW_b} \otimes \mathds{1}_{W_b^{*}}) \circ (i_{W_a^{*}} \otimes i_{W_a}) \circ \ell_{V}^{-1}\bigg]_{ab}~,    
\end{aligned}
\end{equation}
where $\mathds{1}_{W_a}:W_a\longrightarrow W_a$ is the identity morphism, is invertible. This condition on a fusion category defines a modular tensor category. The number of isomorphism classes of simple objects is called the rank of the MTC.   
\end{enumerate}
Using the twist map, we define the $T_{\mathsf{MTC}}$ of the MTC by the matrix elements \footnote{In a fusion category over an algebraically closed field (in particular $\C$), the space of morphisms from a simple object to itself is isomorphic to $\C$ and hence the twist map $\theta_W$ for a simple object $W$ can be identified with a complex number.}
\begin{equation}
[T_{\mathsf{MTC}}]_{ab}:=\theta_{W_a}\delta_{ab}~,    
\end{equation}
the numbers $\theta_{W_a}$ are called \textit{twists} of the MTC.
By semi-simplicity of the category, we can decompose the tensor product of simple objects as
\begin{equation}
    W_a\boxtimes W_b\cong N_{ab}^cW_c~,
\end{equation}
for some positive integers $N_{ab}^c\in\Z_{\geq 0}$. The triple $(N,S_{\mathsf{MTC}},T_{\mathsf{MTC}})$
is called the \textit{modular symbol} of the MTC. The modular symbol of an MTC satisfies the \textit{Verlinde formula} \cite[Theorem 4.5.1]{turaev1994quantum} 
\begin{equation}
    N_{a_1a_2}^{a_3}=\sum_{a\in\mathscr{A}} \frac{[S_{\mathsf{MTC}}]_{a_1a} [S_{\mathsf{MTC}}]_{a_2a} [S_{\mathsf{MTC}}]_{a a_3}^{\star}}{[S_{\mathsf{MTC}}]_{0a}}~.
\end{equation}
Moreover, there exists a phase $\Theta$ such that\footnote{This phase can be determined in terms of the twists and the quantum dimension of the modular symbol, see \cite{rowell2009classification}.} such that 
\begin{equation}
    (\mathcal{S}_{\mathsf{MTC}} \mathcal{T}_{\mathsf{MTC}})^3=\Theta \mathcal{S}_{\mathsf{MTC}}^2,\quad \mathcal{S}_{\mathsf{MTC}}^4=\mathds{1}~,
\end{equation}
see \cite{rowell2009classification} for a more elaborate definition. The topological central charge $c_{\mathsf{MTC}}$ of the MTC is defined by 
\begin{equation}
    \Theta=:e^{\frac{2\pi ic_{\mathsf{MTC}}}{8}}~,
\end{equation}
and the (topological) conformal weights $\{\Delta_a\}_{a\in\mathscr{A}}$ for the simple objects are defined by 
\begin{equation}
   \theta_{W_a}=: e^{2\pi i\Delta_a}~.
\end{equation}
Note that $c_{\mathsf{MTC}}$ is unique modulo 8 and $\Delta_a$ is unique modulo 1. 
\\\\
We will now use the standard terminology of VOA literature, for which we refer to \cite{Frenkel:1988xz,Frenkel:1993xz,lepowsky2012introduction}. 
Let $(V,Y_V)$ be a VOA. Let $\mathcal{C}_V$ be the category of $V$-modules with morphisms being $V$-module homomorphims. Then $\mathcal{C}_V$ is naturally an abelian category with the abelian structure inherited from the vector space of module homomorphisms. A construction of tensor bifunctor in this category was given by Huang-Lepowsky \cite{huang1995theory1,huang1995theory2,huang1995theory3,huang1995theory4}. When the VOA is ``nice'' \footnote{\label{foot:nicevoa}We say that a VOA is ``nice'' if it satisfies the following conditions:
\begin{enumerate}
\item $V_{(n)}=0$ for $n<0, V_{(0)}=\mathbb{C} \mathbf{1}$ and $V^{\prime}$ is isomorphic to $V$ as a $V$-module, where $V'$ is the contragradient of $V$ (see \cite{Frenkel:1993xz} for the definition).
\item  Every $\mathbb{N}$-gradable weak $V$-module is completely reducible (see \cite{abe2002rationality} for relevant definitions). 
\item $V$ is $C_2$-cofinite.
\end{enumerate}
A ``nice'' VOA admits the structure of a braided tensor category, see \cite{huang2008rigidity}. Moreover one can show that a ``nice'' VOA has finitely many isomorphism classes of irreducible modules  \cite{dong1998twisted} and all the fusion rules VOA are finite \cite{abe2003finiteness,Gaberdiel:2000qn,Huang:2002yb,li1999some}.} then it was shown by Huang that $\mathcal{C}_V$ equipped with the tensor bifunctor is a braided tensor category with $V$ being the unit object. In fact more is true: let $\{W_a\}_{a\in\mathscr{A}}$ with $W_0=V$ be a complete list of representatives of (finitely many) isomorphism class of irreducible $V$-modules. Then $\mathcal{C}_V$ is a modular tensor category with simple objects $\{W_a\}_{a\in\mathscr{A}}$ \cite{huang2005vertex1,huang2008vertex2,huang2008rigidity}.
\par
In subsequent sections, we will describe each of the data for the category of modules of a VOA, and explicitly calculate it for the chiral parts of the LLCFT. 

\subsubsection{The Tensor Bifunctor}
Let us briefly recall the notion of $P(z)$-tensor product. Let $(V,Y)$ be a VOA.  
For every $z\in\mathbb{C}^\times$, the $P(z)$-\textit{product} of two modules $W_1,W_2$ is a pair $(W_3,I_{P(z)})$ where $W_3$ is a module and \footnote{$\ov{W}_3$ denotes the algebraic completion of $W_3$ defined by 
\begin{equation}
    \ov{W}_3=\prod_{n\in\C}(W_3)_{(n)}~.
\end{equation}
} $I_{P(z)}:W_1\otimes W_2\longrightarrow \ov{W}_3$ is an \textit{intertwining map}, i.e., 
\begin{equation}
    I_{P(z)}:=\mathcal{Y}(\cdot,z)~,
\end{equation}
for some intertwining operator $\mathcal{Y}$ of type ${W_3\choose W_1 W_2}$. The $P(z)$-\textit{tensor product} of $W_1,W_2$ is then defined by the following universal property: the $P(z)$-tensor product of $W_1,W_2$, denoted by $(W_1\boxtimes_{P(z)}W_2,\boxtimes_{P(z)})$ is a $P(z)$-product such that for any other $P(z)$-product $(W_3,I_{P(z)})$, there exists a module map $f : W_1 \boxtimes_{P(z)}W_2 \longrightarrow W_3 $, with the unique extension $\overline{f} : \overline{W_1 \boxtimes_{P(z)}W_2} \longrightarrow \overline{W}_3 $, such that $I_{P(z)} = \overline{f} \circ \boxtimes_{P(z)} $, i.e., the following diagram commutes:

\begin{center}
    \begin{tikzcd}
	{\overline{W_1 \boxtimes_{P(z)} W_2}} && {\overline{W_3}} \\
	\\
	&& {W_1 \otimes W_2}
	\arrow["{\bar f}", from=1-1, to=1-3]
	\arrow["{\boxtimes_{P(z)}}", from=3-3, to=1-1]
	\arrow["I"', from=3-3, to=1-3]
\end{tikzcd}
\end{center}
The explicit construction of $P(z)$-tensor product in the category of modules of a VOA was given in \cite{huang1995theory1,huang1995theory2,huang1995theory3,huang1995theory4,HLZ1,HLZ2,HLZ3,HLZ4,HLZ5,HLZ6,HLZ7}.
The tensor bifunctor in the category of modules is then taken to be the $P(1)$-tensor product. 
One can then show that if $\{W_a\}_{a\in\mathcal{A}}$ is a complete list of representatives of the (finitely many) isomorphism classes of irreducible modules of a ``nice''  (see Footnote \ref{foot:nicevoa}) VOA, then 
\begin{equation}\label{eq:tens_prod_decomp}
    W_a\boxtimes_{P(1)} W_b\cong N_{ab}^cW_c~,
\end{equation}
where $N_{ab}^c$ are fusion rules. The tensor product is then extended to arbitrary modules using bilinearity. The associativity isomorphism, left and right unitors and braiding isomorphism were constructed and checked to satisfy the coherence conditions in \cite{huang2005vertex1,huang2008rigidity,huang2008vertex2,HLZ8}. With $V$ (as a module for itself) being the unit object in the category, the module category is a braided tensor category. \par   
We will focus on the category $\mathcal{C}_{V^L}$ of modules of the left-moving VOA $V^L$ of the LLCFT, the right-moving category is similar. The following theorem is crucial. 
\begin{thm}\cite{Dong1993VertexAA,lepowsky2012introduction}
Let $\left(L_0,\langle\cdot, \cdot\rangle\right)$ be a positive definite even lattice of full rank and let $(V_{L_0},Y^L_0)$ be the VOA based on $L_0$. Then, the irreducible modules, up to isomorphism, are in one-to-one correspondence with the cosets $\left(L_0\right)^{\star} / L_0$. Moreover, any $V_{L_0}$-module is completely reducible.  We will denote a choice of representatives of the isomorphism class of irreducible $V_{L_0}$-module by $\{(V^L(\alpha),Y^L_\alpha)\}_{[\alpha]\in \left(L_0\right)^{\star} / L_0}$ with the identification $(V^L(0),Y^L_0)\cong (V_{L_0},Y^L_0)$. The fusion rules are given by $N_{\alpha_1\alpha_2}^{\alpha_3}=\delta^{\alpha_3}_{\alpha_1+\alpha_2}$.   
\end{thm}
This theorem for the case $L_0=\Lambda_1^0,\Lambda_2^0$ implies that the category $\mathcal{C}_{V^L},\mathcal{C}_{V^R}$ is a braided tensor category
\footnote{One can show that the VOAs $V^L,V^R$ satisfy all ``niceness'' assumptions, see for example \cite{lepowsky2012introduction}.}.In view of part (3) of Theorem \ref{thm:rankL0mn}, we can label the modules of $V^L$ by $\Lambda/\Lambda_0$. So we will denote a choice of representatives of the isomorphism class of $V^{L}$-modules by $\{(V^L(\mu),Y^L)\}_{[\mu]\in \Lambda/ \Lambda_0}$. The fusion rules are then given by 
\begin{equation}\label{eq:fusrulegenllcft}
    N_{\mu\nu}^{\rho}=\delta_{\mu+\nu}^{\rho}~,
\end{equation}
which is identical to the fusion rules for the modules of non-chiral VOA $V(0)$.
%By Dong's classification of modules of lattice VOAs \cite{Dong1993VertexAA}, the irreducible modules of $V^L,V^R$ are in one-to-one correspondence with the cosets $(\Lambda_1^0)^*/\Lambda_1^0$ and $(\Lambda_2^0)^*/\Lambda_2^0$ respectively. To construct the MTC corresponding to $\{(V(0),Y_{0})\}_{\mu\in\Lambda/\Lambda_0}$, we consider the left-moving part and the right-moving part of the LLCFT. By definition  We will see that the MTC obtained from the left and right moving VOAs are the same. We begin by reviewing the construction of the tensor bifunctor on the category of modules of the VOA based on $\Lambda_1^0$ and $\Lambda_2^0$.  
Then the tensor product of irreducible modules $V^L(\mu),V^L(\nu)$ is given by 
\begin{equation}
    V^L(\mu)\boxtimes_{P(1)}V^L(\nu)\cong V^L(\mu+\nu)~,
\end{equation}
see \cite{dantan} for a nice exposition on the construction of the tensor bifunctor for the category of modules of lattice VOAs.  
\subsubsection{Braiding And Fusing Matrix For LLCFTs}
The associativity and braiding isomorphisms along with the complete reducibility of modules imply the existence of braiding and fusing matrices \cite{bakalov2001lectures,turaev1994quantum}. Since the $P(z)$-tensor product is defined using intertwining operators in the category of modules of a ``nice'' VOA, the braiding and fusing matrices are given by the associativity and commutativity of intertwining operators, see \cite{huang2005vertex1,huang2008rigidity,huang2008vertex2} for the precise definition. The coherence condition for the associativity and braiding isomorphism implies and is equivalent to \textit{hexagon} and \textit{pentagon} relations for fusing and braiding matrices \cite{Moore:1988qv,Moore:1988uz,Moore:1989vd,Blumenhagen:2009zz}. \par In this section, we will compute the braiding and fusing matrices for the category $\mathcal{C}_{V^L},\mathcal{C}_{V^R}$ and show that they satisfy the pentagon and hexagon relations. 
The systematic way to calculate these is to use the intertwining operators of the VOA $V^L,V^R$ as described in \cite{dantan}. But equivalently, we can take the full intertwining operator of the non-chiral VOA $V(0)$ and remove the (anti-)holomoprhic part. We choose the latter to avoid introducing new symbols. To this end, we define the (anti-)holomorphic part of the usual intertwining operator as follows: for a vector 
\begin{equation}
    w :=  \left( \alpha_{1}(-m_1)\cdot\alpha_{2}(-m_2)\cdots\alpha_{k}(-m_k) \, \beta_{1}(-\bar{m}_1)\cdot\beta_{2}(-\bar{m}_2)\cdots\beta_{\bar{k}}(-\bar{m}_{\bar{k}}) \right)  \otimes {\rm e}^{\tilde{\mu}}\in V(\mu)    \, , 
\end{equation}
with ${\tilde{\mu}} - \mu \in  \Lambda_0$, we define
%\fixme{In the big parentheses in (4.1) $x \to z$? RKS: Yes}
\begin{equation}
    \begin{split}
       & \mathcal{Y}^{\rho}_{\mu \nu} (w, z) :=  \typecolon\prod_{r=1}^k \left(\frac{1}{(m_r-1)!}\frac{d^{m_r-1}\alpha_{r}(z)}{dz^{m_r-1}}\right) \mathcal{Y}^{\rho}_{\mu \nu} ({\rm e}^{\tilde{\mu}}, z) \typecolon  \, , 
    \end{split}
\end{equation}
where 
\begin{equation}
\mathcal{Y}^{\rho}_{\mu \nu} ({\rm e}^{\tilde{\mu}}, z) = \exp\left(\int dz~\alpha^{\tilde{\mu}}(z)^-\right) \exp\left(\int dz~\alpha^{\tilde{\mu}}(z)^+\right)  {\rm e}_{\tilde{\mu}} z^{\alpha^{\tilde{\mu}}} ~.
\end{equation} 

\begin{comment}
In \cite{Moore:1989vd}, intertwining operators are defined as maps
\begin{equation}
    \Phi_{i k}^{j \beta}(z,\bar z) : \mathcal{H}_k \to \mathcal{H}_i \, , 
\end{equation}
where $\beta$ is a highest weight vector in $\mathcal{H}_j$ , i.e. $\beta \in W_j$ - the ground state representation. We now need to define their action on highest weight vectors of $\mathcal{H}_i$ and $\mathcal{H}_k$, i.e. 
\begin{equation}
    \bra{\alpha}  \Phi_{i k}^{j \beta}(z,\bar z) \ket{\gamma} = t_{\beta \gamma}^{\alpha} z^{-{h_i + h_j - h_k}} \bar{z}^{-{\bar{h}_i + \bar{h}_j - \bar{h}_k}} \, .
\end{equation}

In \cite{Singh:2023mom}, the intertwining operator are defined in - .

\begin{equation}
    \bra{\text{e}^{\mu_1 + \mu_2}}\mathcal{Y}^{\mu_1 + \mu_2}_{\mu_1, \mu_2}({\rm e}^{\mu_1},z,\bar{z}) \ket{\text{e}^{\mu_2 }} = \bra{\text{e}^{\mu_1 + \mu_2}} z^{\langle \alpha_{\mu_1}, \alpha_{\mu_2}\rangle} \bar{z}^{\langle \beta{\mu_1}, \beta{\mu_2}\rangle}\ket{\text{e}^{\mu_1 + \mu_2}} =  z^{\langle \alpha_{\mu_1}, \alpha_{\mu_2}\rangle} \bar{z}^{\langle \beta{\mu_1}, \beta{\mu_2}\rangle}
 \end{equation}
\begin{itemize}
    \item Find $t$ of \cite{Moore:1989vd} for Intertwining operators of \cite{Singh:2023mom}.

\item Find the Braiding and Fusing Matrices and check that they satisfy the Pentagon and Hexagon Equations. 
    \item Check 2.3 of \cite{Moore:1988uz}.
\end{itemize}
\end{comment}
Similarly, we define the (anti-)holomorphic intertwining operator as 
%\fixme{In the big parentheses in (4.4) $x \to \bar z$? RKS: Yes.}
\begin{equation}
    \begin{split}
       & \mathcal{Y}^{\rho}_{\mu \nu} (w, \bar{z}) :=  \typecolon\prod_{s=1}^{\bar{k}} \left(\frac{1}{(\bar{m}_s-1)!}\frac{d^{\bar{m}_s-1}\beta_{s}(\bar{z})}{d\bar{z}^{\bar{m}_s-1}}\right) \mathcal{Y}^{\rho}_{\mu \nu} ({\rm e}^{\tilde{\mu}}, \bar{z}) \typecolon \, , 
    \end{split}
\end{equation}
where
\begin{equation}
\mathcal{Y}^{\rho}_{\mu \nu} ({\rm e}^{\tilde{\mu}}, \bar{z}) = \exp\left(\int d\bar{z}~\beta^{\tilde{\mu}}(\bar{z})^-\right) \exp\left(\int d\bar{z}~\beta^{\tilde{\mu}}(\bar{z})^+\right)  {\rm e}_{\tilde{\mu}} \bar{z}^{\alpha^{\tilde{\mu}}} ~,
\end{equation}
Let $\s{Y}_{\mu\nu}^\rho,\s{Y}_{\lambda\rho}^\sigma$ be (holomorphic) intertwining operators of type ${\rho\choose \mu~ \nu}$ and ${\sigma\choose \lambda~\rho}$ respectively.
The left-moving braiding matrices are defined by 
\begin{equation}
\label{eq: def-braid-M}
\mathcal{Y}^{\sigma}_{\lambda\rho} (\mathrm {\rm e}^\lambda,z_1) \mathcal{Y}^\rho_{\mu\nu}
(\mathrm {\rm e}^\mu, z_2) = 
    \sum_{\delta} \s{B}^L_{\rho\delta}
    \left[\begin{array}{ll}
\lambda & \mu \\
\sigma & \nu
\end{array}\right] 
\mathcal{Y}^\sigma_{\mu\delta} 
    (\mathrm {\rm e}^\mu, z_2)\mathcal{Y}^\delta_{\lambda\nu} 
    (\mathrm {\rm e}^\lambda, z_1) \, .
\end{equation}
The right-moving braiding matrices $\s{B}^R_{\rho\delta}
    \left[\begin{array}{ll}
\lambda & \mu \\
\sigma & \nu
\end{array}\right]$ are defined analogously.
The fusion rule \eqref{eq:fusrulegenllcft} implies that the the braiding matrix must be proportional to $\delta_{\rho, \mu+\nu} 
\delta_{\delta, \lambda+\nu} \delta_{\sigma,\lambda+\rho}$. To compute the coefficient, we calculate
%\fixme{the first line of (4.7) should be $z_1$ while the second line should be $z_2$and $dx \to d z$ ? RKS: Yes}
\begin{equation}
\begin{split}
\mathcal{Y}^{\sigma}_{\lambda\rho} (\mathrm {\rm e}^\lambda,z_1) \mathcal{Y}^\rho_{\mu\nu}
(\mathrm {\rm e}^\mu, z_2){\rm e}^{\nu}&=\exp\left(\int dz_1~\alpha^{\lambda}(z_1)^-\right) \exp\left(\int dz_1~\alpha^{\lambda}(z_1)^+\right)  {\rm e}_{\lambda} z_1^{\alpha^{\lambda}}\\&~~~\exp\left(\int dz_2~\alpha^{\mu}(z_2)^-\right) \exp\left(\int dz_2~\alpha^{\mu}(z_2)^+\right)  {\rm e}_{\mu} z_2^{\alpha^{\mu}}{\rm e}^{\nu}~.
\end{split}
\end{equation}
Following the steps in the proof of \cite[Proposition 3.1]{Singh:2023mom}, we get 
\begin{equation}
\begin{split}
\mathcal{Y}^{\sigma}_{\lambda\rho} (\mathrm {\rm e}^\lambda,z_1) \mathcal{Y}^\rho_{\mu\nu}
(\mathrm {\rm e}^\mu, z_2){\rm e}^{\nu}&=(z_1-z_2)^{\langle\alpha^\lambda,\alpha^\mu\rangle}z_1^{\langle\alpha^\lambda,\alpha^\nu\rangle}z_2^{\langle\alpha^\mu,\alpha^\nu\rangle}(-1)^{\epsilon(\lambda,\mu+\nu)}(-1)^{\epsilon(\mu,\nu)}\\&\delta_{\rho, \mu+\nu} 
%\delta_{\delta, \lambda+\nu} 
\delta_{\sigma,\lambda+\rho}\exp\left(\int d z_1~\alpha^{\lambda}(z_1)^-\right)\exp\left(\int dz_2~\alpha^{\mu}(z_2)^-\right){\rm e}^{\mu+\nu+\lambda}~.     
\end{split}  
\end{equation}
Similarly, we have 
\begin{equation}
\begin{split}
    \mathcal{Y}^\sigma_{\mu\delta} 
    (\mathrm {\rm e}^\mu, z_2)\mathcal{Y}^\delta_{\lambda\nu} 
    (\mathrm {\rm e}^\lambda, z_1) \mathrm {\rm e}^\nu&=(z_2-z_1)^{\langle\alpha^\lambda,\alpha^\mu\rangle}z_1^{\langle\alpha^\lambda,\alpha^\nu\rangle}z_2^{\langle\alpha^\mu,\alpha^\nu\rangle}(-1)^{\epsilon(\mu,\lambda+\nu)}(-1)^{\epsilon(\lambda,\nu)}\\&\delta_{\sigma, \mu+\delta} 
%\delta_{\delta, \lambda+\nu} 
\delta_{\delta,\lambda+\nu}\exp\left(\int dz_1~\alpha^{\lambda}(z_1)^-\right)\exp\left(\int dz_2~\alpha^{\mu}(z_2)^-\right){\rm e}^{\mu+\nu+\lambda}~.       
\end{split} 
\end{equation}
Using the cocycle relation \cite[Eq. (A.12)]{Singh:2023mom}, we find the left-moving braiding matrix 
\begin{equation}\label{eq:braidL}
    \s{B}^L_{\rho\delta}
    \left[\begin{array}{ll}
\lambda & \mu \\
\sigma & \nu
\end{array}\right]
= {\rm e}^{-\pi i\langle\alpha^\mu,\alpha^\lambda\rangle}{\rm e}^{\pi i(\lambda \circ \mu)}\delta_{\rho, \mu+\nu} 
\delta_{\delta, \lambda+\nu} \delta_{\sigma,\lambda+\rho}~.
\end{equation}
Similarly, the right-moving braiding matrix is given by 

\begin{equation}\label{eq:braidR}
    \s{B}^R_{\rho\delta}
    \left[\begin{array}{ll}
\lambda & \mu \\
\sigma & \nu
\end{array}\right]
= {\rm e}^{\pi i\langle\beta^\mu,\beta^\lambda\rangle}{\rm e}^{\pi i(\lambda \circ \mu)}\delta_{\rho, \mu+\nu} 
\delta_{\delta, \lambda+\nu} \delta_{\sigma,\lambda+\rho}~.
\end{equation}
Both the left and right-moving braiding matrices satisfy the hexagon relation 
\begin{equation}
    \sum_p \s{B}_{r p}\left[\begin{array}{ll}j & k \\ i & s\end{array}\right] \s{B}_{s t} \left[\begin{array}{ll}j & l \\ p & m\end{array}\right]\s{B}_{p u}\left[\begin{array}{ll}k & l \\ i & t\end{array}\right]=\sum_q \s{B}_{s q}\left[\begin{array}{ll}k & l \\ r & m\end{array}\right] \s{B}_{r u} \left[\begin{array}{ll}j & l \\ i & q\end{array}\right]\s{B}_{q t}\left[\begin{array}{ll}j & k \\ u & m\end{array}\right] \, . 
\end{equation}
%\subsection{Fusing Matrix}
We now compute the fusing matrices.
The left-moving fusing matrices are defined as 
\begin{equation}
\mathcal{Y}^{\sigma}_{\lambda\rho} (\mathrm {\rm e}^\lambda,z_1)\mathcal{Y}^\rho_{\mu\nu}
(\mathrm {\rm e}^\mu, z_2) = \sum_{\delta} \s{F}^L_{\rho\delta}
    \left[\begin{array}{ll}
\lambda & \mu \\
\sigma & \nu
\end{array}\right] 
\mathcal{Y}^\sigma_{\delta\nu} \left(\mathcal{Y}^\delta_{\lambda\mu} 
    ({\rm e}^\lambda,z_1- z_2){\rm e}^\mu, z_2\right)\,,
        \end{equation}
and the left-moving fusing matrices $\s{F}^R_{\rho\delta}
    \left[\begin{array}{ll}
\lambda & \mu \\
\sigma & \nu
\end{array}\right] $ are defined analogously.        \\\\
The non-trivial fusing matrices can be computed as follows: let $\textbf{1}$ denote the vacuum of the LLCFT. Then for $z_3 \in \mathbb{C}^\times$
\begin{equation*}
\begin{aligned}
\mathcal{Y}^{\mu+\nu+\lambda}_{\lambda,\mu+\nu} ({\rm e}^\lambda, z_1)&
\mathcal{Y}^{\mu+\nu}_{\nu,\mu} ({\rm e}^\nu, z_2) {\rm e}^{z_3L(-1)} {\rm e}^\mu \\
&= 
\mathcal{Y}^{\mu+\nu+\lambda}_{\lambda,\mu+\nu} ({\rm e}^\lambda, z_1)
\mathcal{Y}^{\mu+\nu}_{\nu,\mu} ({\rm e}^\nu, z_2) 
\mathcal{Y}^{\mu}_{\mu,0}({\rm e}^\mu,z_3) \mathbf{1}
\\
&\sim {\rm e}^{-\pi i \langle \alpha^{\lambda+\nu},\alpha^\mu\rangle}
{\rm e}^{\pi i (\lambda+\nu) \circ \mu}
\mathcal{Y}^{\mu+\nu+\lambda}_{\mu,\nu+\lambda}({\rm e}^\mu,z_3)
\mathcal{Y}^{\nu +\lambda}_{\lambda,\nu}({\rm e}^\lambda,z_1)
\mathcal{Y}^{\nu}_{\nu,0}({\rm e}^\nu,z_2) \mathbf{1} \\
&= {\rm e}^{-\pi i \langle \alpha^{\lambda+\nu},\alpha^\mu\rangle}
{\rm e}^{\pi i (\lambda+\nu) \circ \mu}
\mathcal{Y}^{\mu+\nu+\lambda}_{\mu,\nu+\lambda}({\rm e}^\mu,z_3)
\mathcal{Y}^{\nu +\lambda}_{\lambda,\nu}({\rm e}^\lambda,z_1)
{\rm e}^{z_2L(-1)} {\rm e}^\nu \\
&= 
{\rm e}^{-\pi i \langle \alpha^{\lambda+\nu},\alpha^\mu\rangle}
{\rm e}^{\pi i (\lambda+\nu) \circ \mu}
\mathcal{Y}^{\mu+\nu+\lambda}_{\mu,\nu+\lambda}({\rm e}^\mu,z_3)
{\rm e}^{z_2 L(-1)} 
\mathcal{Y}^{\nu+\lambda}_{\lambda,\nu}({\rm e}^\lambda,z_1-z_2){\rm e}^\nu
\\
&=
{\rm e}^{-\pi i \langle \alpha^{\lambda+\nu},\alpha^\mu\rangle}
{\rm e}^{\pi i (\lambda+\nu) \circ \mu}
\mathcal{Y}^{\mu+\nu+\lambda}_{\mu,\nu+\lambda}({\rm e}^\mu,z_3)
\mathcal{Y}^{\nu+\lambda}_{\nu+\lambda,0} 
(\mathcal{Y}^{\nu+\lambda}_{\lambda,\nu}({\rm e}^\lambda,z_1-z_2){\rm e}^\nu,
z_2) \mathbf{1} \\
&\sim 
{\rm e}^{-2\pi i \langle \alpha^{\lambda+\nu},\alpha^\mu\rangle}
\mathcal{Y}^{\mu+\nu+\lambda}_{\nu+\lambda,\mu}
(\mathcal{Y}^{\nu+\lambda}_{\lambda,\nu}({\rm e}^\lambda,z_1-z_2)
{\rm e}^\nu,z_2) \mathcal{Y}^\mu_{\mu,0}({\rm e}^\mu,z_3) \mathbf{1}\,,
\end{aligned}
\end{equation*}
where we used  
\cite[Eq. (4.35)]{Singh:2023mom},
the definition of braiding matrix \eqref{eq: def-braid-M} 
and the translation property\footnote{This follows from the translation property for the full intertwining operator.} 
\begin{equation}
{\rm e}^{z_2L(-1)} \mathcal{Y}^{\nu+\lambda}_{\lambda,\nu}({\rm e}^\lambda,z_1)
{\rm e}^{-z_2L(-1)} = 
\mathcal{Y}^{\nu+\lambda}_{\lambda,\nu}({\rm e}^\lambda,z_1+z_2)~.
\end{equation}
Taking $z_3\to 0$, we get 
\begin{equation}
\mathcal{Y}^{\mu+\nu+\lambda}_{\lambda,\mu+\nu} ({\rm e}^\lambda, z_1)
\mathcal{Y}^{\mu+\nu}_{\nu,\mu} ({\rm e}^\nu, z_2) {\rm e}^\mu =
{\rm e}^{-2\pi i \langle \alpha^{\lambda+\nu},\alpha^\mu\rangle}
\mathcal{Y}^{\mu+\nu+\lambda}_{\nu+\lambda,\mu}
(\mathcal{Y}^{\nu+\lambda}_{\lambda,\nu}({\rm e}^\lambda,z_1-z_2)
{\rm e}^\nu,z_2) {\rm e}^\mu~.
\end{equation}
Thus, the left-moving fusing matrices are 
\begin{equation}\label{eq:fuseL}
\s{F}^L_{\rho \delta}\left[\begin{array}{cc}
\lambda & \mu \\
\sigma & \nu 
\end{array}\right] =
{\rm e}^{-2\pi i \langle\alpha^\mu+\alpha^\lambda, \alpha^\nu\rangle}
\delta_{\sigma,\lambda+\mu+\nu} 
\delta_{\rho, \mu+\nu}
\delta_{\delta,\lambda+\mu}\,.
\end{equation}
Similarly, the right-moving fusing matrices can be derived 
\begin{equation}\label{eq:fuseR}
\s{F}^R_{\rho \delta}\left[\begin{array}{cc}
\lambda & \mu \\
\sigma & \nu 
\end{array}\right] =
{\rm e}^{2\pi i \langle\beta^\mu+\beta^\lambda, \beta^\nu\rangle}
\delta_{\sigma,\lambda+\mu+\nu} 
\delta_{\rho, \mu+\nu}
\delta_{\delta,\lambda+\mu}\,.
\end{equation}
The fusing matrices $\s{F}^L$ and $\s{F}^R$ satisfy the Pentagon equation 
\begin{equation}
\s{F}_{r t}\left[\begin{array}{ll}j & k \\ i & s\end{array}\right] \s{F}_{s u}\left[\begin{array}{ll}t & l \\ i & m\end{array}\right]=\sum_p \s{F}_{s p}\left[\begin{array}{ll}k & l \\ r & m\end{array}\right] \s{F}_{r u}\left[\begin{array}{ll}j & p \\ i & m\end{array}\right] \s{F}_{p t}\left[\begin{array}{ll}j & k \\ u & l\end{array}\right] \, . 
\end{equation}

\subsubsection{Modular Matrices And Verlinde Formula}
\label{sec:modular_Verlinde}
For an object $W$ in the category $\mathcal{C}_V$, the twist map $\theta_W$ is given by \cite{huang2005vertex1}
\begin{equation}
    \theta_W=e^{2\pi iL^W(0)}~.
\end{equation}
By irreducibility it is enough to determine the twist map for the irreducible modules. The twist map as well as the $\mathcal{S}$-matrix is captured in the modular transformation of the graded dimensions of the modules of the VOA \cite{huang2005vertex1,huang2008vertex2,huang2008rigidity}: let $\tau$ be a variable in the upper half plane $\mathbb{H}:=\{x+iy\in\C:y>0\}$ and $q:=e^{2\pi i\tau}$. Then by Zhu's theorem \cite{Zhu1995ModularIO}, the graded dimensions of the modules defined by 
\begin{equation}
\chi_a(\tau):=\mathsf{Tr}_{W_a}q^{L^W_a(0)-\frac{c}{24}}~,    
\end{equation}
furnishes a representation of $\mathrm{SL}(2,\Z)$:
\begin{align}
\label{eq: ST_trans}
T:&\quad \chi_{a}(\tau+1)=\sum_{b\in\mathscr{A}}T_{ab} \chi_{b}(\tau)~,\\\label{eq: S_trans}
S:&\quad    \chi_{a}(-1 / \tau)=\sum_{b\in\mathscr{A}}S_{ab}\chi_{b}(\tau)~.
\end{align} 
Then $(T_{\mathsf{MTC}},S_{\mathsf{MTC}})$ is given by 
\begin{equation}
    T_{\mathsf{MTC}}=e^{\frac{2\pi ic}{24}}T,\quad S_{\mathsf{MTC}}=S~.
\end{equation}
We now restrict our attention to Rational modular invariant LLCFTs. The partition function \eqref{eq:partZlamb} can be written as
\begin{equation}\label{eq:partfunc}
\begin{split}
    Z_{\Lambda}(\tau,\bar{\tau})&=\sum_{[\mu]\in\Lambda/\Lambda_0}\chi_{\mu}(\tau,\bar{\tau})\\&=\frac{1}{\eta(\tau)^{m}\overline{\eta(\tau)}^{n}}\sum_{(\alpha,\beta)\in\Lambda}q^{\frac{\langle\alpha,\alpha\rangle}{2}}\bar{q}^{\frac{\langle\beta,\beta\rangle}{2}} \, . 
\end{split}    
\end{equation}
% where the graded dimension (character) $\chi_{W_\mu}(\tau,\Bar{\tau})$ for the module $W_\mu$ is given by
% \begin{equation}\label{eq:gradcharllvoa_mu}
%     \chi_{\mu}(\tau,\bar{\tau})=\frac{1}{\eta(\tau)^{m}\overline{\eta(\tau)}^{n}}\sum_{(\alpha,\beta)\in\mu+\Lambda_0}q^{\frac{\langle\alpha,\alpha\rangle}{2}}\bar{q}^{\frac{\langle\beta,\beta\rangle}{2}} \, .
% \end{equation}
For an even, Euclidean lattice $L\subset\R^d$ of dimension $d$ and any coset $[\alpha]\in L^{\star}/L$, introduce the theta functions $\Theta_{L+\alpha}(\tau)$ 
\begin{equation}
\Theta_{L+\alpha}(\tau):=\sum_{\alpha'\in L+\alpha}q^{\frac{\langle\alpha',\alpha'\rangle}{2}}~.%~q={\rm e}^{2\pi i\tau}~. 
\end{equation}
Here $L^\star$ denotes the dual of $L$.
These theta functions furnish a representation of $\mathrm{SL}(2,\Z)$ \cite{elkies}:
\begin{align}
\label{eq: ST trans}
T:&\quad \Theta_{L+\alpha}(\tau+1)={\rm e}^{\pi i\langle\alpha,\alpha\rangle} \Theta_{L+\alpha}(\tau)~,\\\label{eq: Strans}
S:&\quad    \Theta_{L+\alpha}(-1 / \tau)=\frac{1}{\sqrt{|L^\star/L|}}(\tau / i)^{d / 2} \sum_{\alpha^{\prime} \in L^\star / L} {\rm e}^{-2 \pi i\langle\alpha, \alpha^{\prime}\rangle} \Theta_{L+\alpha'}(\tau)~.
\end{align}    
We denote the theta functions for $\Lambda_1^0$ and $\Lambda_2^0$ by 
\begin{equation}
    \Theta_{\Lambda_1^0+\alpha}(\tau),\quad\Theta_{\Lambda_2^0+\beta}(\tau),\quad \alpha\in (\Lambda_1^0)^\star/\Lambda_1^0,~~ \beta\in (\Lambda_2^0)^\star/\Lambda_2^0~.
\end{equation}
Then by Theorem \ref{thm:rcft},
the graded dimension for the module $V(\mu)$, due to  \eqref{eq:Lambi0}, decomposes as 
\begin{equation}
\begin{split}
 \chi_{\mu}(\tau,\bar{\tau}) &= \frac{1}{\eta(\tau)^{m}\overline{\eta(\tau)}^{n}}\Theta_{\Lambda_1^0+\alpha^\mu}(\tau)\overline{\Theta_{\Lambda_2^0+\beta^\mu}(\tau)}\\&=:\chi^L_{\mu}(\tau)\overline{\chi^R_{\mu}(\tau)}~,
 \end{split}
\end{equation}
where $\mu=(\alpha^\mu,\beta^\mu)$, $[\mu] \in \Lambda/\Lambda_0$,  $\eta$ is the Dedekind $\eta$-function and we have introduced the notations 
\begin{equation}
\chi^L_{\mu}(\tau)\equiv\frac{1}{\eta(\tau)^{m}} \Theta_{\Lambda_1^0+\alpha^\mu}(\tau)~,\quad  \chi^R_{\mu}(\tau)\equiv\frac{1}{\eta(\tau)^{n}} \Theta_{\Lambda_2^0+\beta^\mu}(\tau)~,
\end{equation}
for the left and right-moving characters. Recall, that the Dedekind $\eta$-function under $T$ and $S$ transforms as
\begin{equation}\label{eq:Ded-eta-transf}
    \eta(\tau + 1) = {\rm e}^{\frac{\pi i}{12}} \, \eta(\tau) \,, \quad \eta \left(-\frac{1}{\tau} \right) = \sqrt{- i \tau} \,  \eta(\tau) \, .
\end{equation}
%\fixme{Comment some argument here since we introduced (2.28) and (2.29) now}
\begin{comment}
As we discuss in Lemma \ref{eq:sublattice_determinant}, we have an injective map from $\Lambda/\Lambda_0 \to (\Lambda_1^0)^{\star}/(\Lambda_1^0)$,  $[\mu] \mapsto [\mu_1]$, hence $\chi^{L}_{\mu_1}$ is well-defined, one can similarly argue for $\chi^{R}_{\mu_2}$.
\end{comment}
When the LLCFT is modular invariant and rational, then
\begin{equation}
|\Lambda/\Lambda_0|=
|(\Lambda_1^0)^\star/\Lambda_1^0| = 
|(\Lambda_2^0)^\star/\Lambda_2^0|~.    
\end{equation}
Choose $\{\mu=(\alpha^\mu,\beta^\mu)\in \Lambda\}$ to be a set of representatives of the cosets. Then we have the following theorem.
\begin{thm}
Let $\left\{(V(\mu),Y)\right\}_{[\mu]\in\Lambda/\Lambda_0}$ be the rational modular invariant LLCFT based on $\Lambda$. Then there exist $|\Lambda/\Lambda_0|\times |\Lambda/\Lambda_0|$ matrices $T^L,T^R$ and $S^L,S^R$ such that 
\begin{equation}
\label{eqf: def-S-T}
\begin{split}
    &\chi^L_{\mu}(\tau+1)= \sum_{\nu \in\Lambda/\Lambda_0 }T^L_{\mu\nu}\chi^L_{\nu}(\tau),\quad \chi^R _{\mu}(\tau+1)= \sum_{\nu \in\Lambda/\Lambda_0 }T^R_{\mu\nu}\chi^R_{\nu}(\tau) \,, \\&\chi^L_{\mu}(-1/\tau)= \sum_{\nu \in\Lambda/\Lambda_0 }S^L_{\mu\nu}\chi^L_{\nu}(\tau),\quad \chi^{R}_{\mu}(-1/\tau)= \sum_{\nu \in\Lambda/\Lambda_0 }S^{R}_{\mu\nu}\chi^R_{\nu}(\tau)\,.
\end{split}    
\end{equation}
We also have the property that 
\begin{equation}
    T^L=T^R =: T,\quad S^L=S^R=: S\,,
\end{equation}
%Further, for the Rational Lorentzian Lattice CFT corresponding to $\Lambda$, it can be shown that the $S$ matrices 
when $m-n \equiv 0\!  \mod 24$ and $S$ satisfies the Verlinde formula, i.e.
\begin{equation}
 \label{eq:verlform}   \delta_{\mu+\nu}^{\rho}=\sum_{\sigma\in\Lambda/\Lambda_0} \frac{S_{\mu\sigma} S_{\nu\sigma} S_{\sigma \rho}^{\star}}{S_{0\sigma}}~. 
\end{equation}
%where $N_{\ell_1\ell_2}^{\ell_3}$ are the fusion rules \eqref{eq:fusrulegenllcft} for the rational LLCFT.
\end{thm}
\begin{proof}
From  
\eqref{eq: ST trans} and \eqref{eq:Ded-eta-transf}, we get the T-transformation matrices $T^L$ and $T^R$ 
\begin{equation}
    T^L_{\mu\nu} = 
    {\rm e}^{\pi i \left(\langle
    \alpha^\mu, \alpha^\mu 
    \rangle - \frac{m}{12}\right)} \delta_{\mu\nu}, 
    \quad 
    T^R_{\mu\nu} = 
    {\rm e}^{\pi i \left(\langle
    \beta^\mu, \beta^\mu 
    \rangle - \frac{n}{12}\right)} \delta_{\mu\nu}~.
\end{equation}
When $m-n=0 \;\text{mod}\; 24$, 
$T^L = T^R$ since $\Lambda$ is even. Using \eqref{eq: Strans} and \eqref{eq:Ded-eta-transf}, we have the S-transformations $S^L$, $S^{R}$: 
\begin{equation}
    S^L_{\mu\nu} =
    \frac{1}{\sqrt{|\Lambda/\Lambda_0|}} 
    {\rm e}^{-2\pi i \langle \alpha^\mu, \alpha^\nu\rangle}, 
    \quad 
    S^R_{\mu\nu} =
    \frac{1}{\sqrt{|\Lambda/\Lambda_0|}} 
    {\rm e}^{-2\pi i \langle \beta^\mu, \beta^\nu\rangle}~.
\end{equation}
Since $\Lambda$ is an integral lattice
\begin{equation}
 (\alpha^\mu ,\beta^\mu)\circ(\alpha^\nu, \beta^\nu)
=
\langle \alpha^\mu, \alpha^\nu\rangle
-
\langle \beta^\mu, \beta^\nu\rangle
\in   \mathbb{Z}\,,
\end{equation}
we have 
\begin{equation}
    S^L = S^R \,.
\end{equation}
The Verlinde formula for the left-moving sector is easy to check:
\begin{equation}
    \sum_{\sigma\in\Lambda/\Lambda_0} \frac{S_{\mu\sigma} S_{\nu\sigma} S_{\sigma \rho}^{\star}}{S_{0\sigma}}
    =\frac{1}{|\Lambda/\Lambda_0|} \sum_{\sigma\in\Lambda/\Lambda_0} 
    \exp(-2\pi i \langle 
    \alpha^{\mu} + 
    \alpha^{\nu} - 
    \alpha^{\rho}, \alpha^\sigma 
    \rangle)
    = \delta_{[\alpha^\mu+\alpha^\nu]}^{[\alpha^\rho]}\,.
\end{equation}
Similarly one can check the Verlinde formula for the right-moving sector.
\end{proof}
\begin{comment}
The $\mathcal{T}$ and $\mathcal{S}$ matrices satisfy the following equation ((7),\cite{delaney2019lecture})
\begin{equation}
    (\mathcal{S} \mathcal{T})^3=\Theta \mathcal{S}^2,
\end{equation}
where $\Theta$ determines the topological charge of the LLCFT by 
\begin{equation}
    \Theta = {\rm e}^{2\pi i c/8}.
\end{equation}
$c$ is well-defined modulo $8$ and 
$c^L = c^R \mod 8$.
\end{comment}
Note that the $S$ and $T$ matrices also satisfy the following equation
\begin{equation}
    (S T)^3= S^2~,
\end{equation}
and are related to the braiding and fusing matrices by 
(see for example \cite[Eq. (5.15)]{huang2008vertex2})
\begin{equation}
    \frac{S_{\mu\nu}}{S_{00}}=\frac{\s{B}_{\mu\nu}^2}{\s{F}_\mu \s{F}_\nu}~,
\end{equation}
where 
\begin{equation}
    \s{B}_{\mu \nu}:=\s{B}_{00}\left[\begin{array}{cc}\mu & \nu \\ \mu & -\nu\end{array}\right],\quad \s{F}_\mu:=\s{F}_{00}\left[\begin{array}{cc}\mu & -\mu \\ \mu & \mu\end{array}\right]\,,
\end{equation}
The T-matrix of the MTC is given by
\begin{equation}
[T_{\mathsf{MTC}}]_{\mu\nu} = {\rm e}^{2\pi i \Delta^{L,R}_\mu} \delta_{\mu\nu}\,,
\end{equation}
where $\Delta^{L,R}_\mu$, well-defined modulo $\mathbb{Z}$, is given by
\begin{equation}
    \Delta^L_\mu = \frac{1}{2}
    \langle \alpha^\mu, \alpha^\mu \rangle 
    \bmod 1, \quad 
    \Delta^R_\mu = \frac{1}{2}
    \langle \beta^\mu, \beta^\mu \rangle 
    \bmod 1~.
\end{equation}
Since $\Lambda$ is even,
we have 
$\Delta^L_\ell = 
\Delta^R_\ell  \mod 1$.
The twists of MTC are given by 
\begin{equation}
     \theta_\mu
    ={\rm e}^{2 \pi i \Delta_\mu^{L,R} }~.
\end{equation}
The topological central charge is given by 
\begin{equation}
    c_{\mathsf{MTC}}\equiv m\equiv n\bmod 8~.
\end{equation}
With braiding and fusing matrices given in \eqref{eq:braidL}, \eqref{eq:braidR}, \eqref{eq:fuseL} and \eqref{eq:fuseR}, this completes the calculation of the MTC data associated to the left and right moving sectors of the LLCFT. We see that the left and right moving sector determine the same MTC.\par 
Let us comment on the relation of our discussion above with the code construction of Narain CFTs in \cite{Angelinos:2022umf,Aharony:2023zit}. In signature $(n,n)$, given any even, self-dual code $\mathcal{C}\subset \Lambda_0^\star/\Lambda_0$ (see \cite[Eq. (2.16), Eq. (2.17)]{Aharony:2023zit} for the precise definition), one can construct an even self-dual  Lorentzian lattice $\Lambda_{\mathcal{C}}\subset\R^{2n}$ of signature $(n,n)$ as follows:
\begin{equation}
    \Lambda_{\mathcal{C}}:=\{\lambda\in\Lambda_0^\star:[\lambda]\in\mathcal{C}\}~.
\end{equation}
Then there is a Narain CFT based on $\Lambda_{\mathcal{C}}$. The characters (called ``codeword blocks'' in \cite{Aharony:2023zit}) of this CFT are then in one-to-one correspondence with the elements of the code $\mathcal{C}$ and given by \eqref{eq:chichar}. The partition function of this CFT is given by the enumerator polynomial of the code $\mathcal{C}$. The LLCFT is then described by the code $\mathcal{C}=\Lambda/\Lambda_0\subset\Lambda_0^\star/\Lambda_0$.
\subsection{Examples}
\label{sec: examples}
In this section, we will work out two explicit examples of the MTCs coming from LLCFTs. 
\subsubsection{Compact Boson At $R=\sqrt{2k}$}
The $c=1$ compact boson at the radius $R=\sqrt{2k}$ for $k\in\N$ is a rational LLCFT based on the lattice \cite[Section 4]{Dymarsky:2020qom} 
\begin{equation}
    \Lambda:=\left\{\left(\frac{a}{\sqrt{2k}}+\frac{b\sqrt{2k}}{2},\frac{a}{\sqrt{2k}}-\frac{b\sqrt{2k}}{2}\right):a,b\in\Z\right\}~.
\end{equation}
A generator matrix for the lattice is 
\begin{equation}
    \s{G}_{\Lambda}=\frac{1}{\sqrt{2k}}\left(\begin{array}{cc}
        1&1\\k&-k
    \end{array}\right)\,.
\end{equation}
It is easy to check that the sublattice $\Lambda_0$ is given by 
\begin{equation}
    \Lambda_0=\{(\sqrt{2k}a,\sqrt{2k}b):a,b\in\Z\}\,.
\end{equation}
A generator matrix for this sublattice is 
\begin{equation}
\s{G}_{\Lambda_0}=\sqrt{2k}\begin{pmatrix}
        1&0\\0&1
    \end{pmatrix}\,.    
\end{equation}
Using the fact that 
\begin{equation}\label{eq:no_cosets}
    |\Z^dA/\Z^dA|=\left|\frac{\text{det}B}{\text{det}A}\right|\quad\text{if}\quad \text{det}A,\text{det}B\neq 0~,
\end{equation}
we get 
\begin{equation}
    |\Lambda/\Lambda_0|=|\Z^2\s{G}_{\Lambda}/\Z^2\s{G}_{\Lambda_0}|=2k~.
\end{equation}
Thus there are $2k$ irreducible modules up to isomorphism which make up the modular invariant LLCFT. From the general formula \eqref{eq:chichar}, the characters of the modules can easily be computed. We make the following choice of representatives
\begin{equation}
    \Lambda/\Lambda_0=\left\{\left[\left(\frac{\ell}{\sqrt{2k}},\frac{\ell}{\sqrt{2k}}\right)\right]:-k+1\leq \ell\leq k\right\}\,,
\end{equation}
with $[(0,0)]$ being the LLVOA $V(0)$ itself considered as a module for itself. The character $\chi_\ell$ corresponding to the coset $[(\ell/\sqrt{2k},\ell/\sqrt{2k})]$ is given by 
\begin{equation}
\begin{split}
    \chi_\ell(\tau,\Bar{\tau})&=\frac{1}{\eta(\tau)\overline{\eta}(\Bar{\tau})}\sum_{(\alpha,\beta)\in\Lambda_0+(\frac{\ell}{\sqrt{2k}},\frac{\ell}{\sqrt{2k}})}q^{\frac{\langle\alpha,\alpha\rangle}{2}}\bar{q}^{\frac{\langle\beta,\beta\rangle}{2}}\\&=\frac{1}{|\eta(\tau)|^2}\sum_{a,b\in\Z+\frac{\ell}{2k}}q^{ka^2}\bar{q}^{kb^2}\\&=\frac{1}{|\eta(\tau)|^2}|\Theta_{k,\ell}(\tau)|^2~,
\end{split}    
\end{equation}
where 
\begin{equation}
\Theta_{k,\ell}(\tau)=\sum_{a\in\Z+\frac{\ell}{2k}}q^{ka^2},\quad -k+1\leq \ell\leq k~,    
\end{equation}
is the modified Jacobi theta function. The partition function for the theory is given by 
\begin{equation}
    Z_\Lambda(\tau,\bar{\tau})=\frac{1}{|\eta(\tau)|^2}\sum_{\ell=-k+1}^k|\Theta_{k,\ell}(\tau)|^2~.
\end{equation}
Both left and right-moving $S$ and $T$ matrices are
\begin{equation}\label{eq:compbosST}
\begin{split}
    &T_{\ell\ell'}={\rm e}^{\pi i (\frac{\ell^2}{2 k}
    -\frac{1}{12})}\delta_{\ell,\ell'}\,,\\&S_{\ell\ell^{\prime}}=\frac{1}{\sqrt{2 k}} \exp \left(-\pi i \frac{\ell \ell^{\prime}}{k}\right)\,.
\end{split}
\end{equation}
The fusion rules are given by 
\begin{equation}
\label{eq:fusrulcompbos}    N_{\ell_1\ell_2}^{\ell_3}=\delta_{\ell_1+\ell_2}^ {\ell_3}~.
\end{equation}
The Verlinde formula is given by 
\begin{equation}
 \label{eq:verlform2}   N_{\ell_1\ell_2}^{\ell_3}=\sum_\ell \frac{S_{\ell_1\ell} S_{\ell_2\ell} S_{\ell \ell_3}^\star}{S_{0\ell}}~.
\end{equation}
From \eqref{eq:compbosST}, we see that the RHS of \eqref{eq:verlform2} is given by 
\begin{equation}
\begin{split}
    \frac{1}{2k}\sum_{\ell=-k+1}^{k}\exp\left(-\frac{2\pi i\ell(\ell_1+\ell_2-\ell_3)}{2k}\right)&= \frac{1}{2k}\sum_{\ell = -k + 1}^ k
    \exp\left(-\frac{2\pi i\ell(\ell_1+\ell_2-\ell_3)}{2k}\right)\\&=\delta_{\ell_1+\ell_2}^ {\ell_3}=N_{\ell_1\ell_2}^{\ell_3}~.
\end{split}    
\end{equation}
Thus we have verified the Verlinde formula. The braiding and fusing matrices are straightforward to compute using the general formula in \eqref{eq:braidL} and \eqref{eq:fuseL}. Thus we obtain an MTC of rank $2k$. For $k=$ 1 and 2, the modular symbol can be matched with the semion MTC and the $\Z_4$ MTC respectively \cite{rowell2009classification}. Note that the braiding and fusing matrices give one particular realization of the MTC. Another realization of the compact boson MTC is given in \cite{Fuchs:2002cm}. 
\subsubsection{LLCFT for $\mathrm{II}_{m,n}$ with $m+n\in 4\Z$}
An even self-dual lattice of signature $(m,n)$ exists if and only if $(m-n)\equiv 0\bmod 8$. A choice of an even, self-dual lattice of signature $(m,n)$, denoted by $\mathrm{II}_{m,n}$, is 
\begin{equation}\label{eq:def-set-lattice}
\mathrm{II}_{m,n}=\left\{(a_1,\dots,a_{m+n})\in\R^{m,n}:\text{ all }a_i\in\Z\text{ or all }a_i\in\Z+\frac{1}{2},~\sum_{i=1}^{m+n}a_i\in2\Z\right\}.  \end{equation}
For $m+n\in 4\Z$, a generator matrix for this lattice is given by (see Appendix \ref{app:genmat} for proof)
\begin{equation}\label{eq:genmatIImn}
    \mathcal{G}_{\mathrm{II}_{m,n}}=\begin{pmatrix}[ccccccc]
        1&0&\cdots&0&0&-1\\0&1&\cdots&0&0&-1\\\vdots&\vdots&\cdots&\vdots&\vdots&\vdots\\0&0&\cdots&1&0&-1\\0&0&\cdots&0&0&2\\\frac{1}{2}&\frac{1}{2}&\cdots&\frac{1}{2}&\frac{1}{2}&\frac{1}{2}
    \end{pmatrix}\,.
\end{equation}
%Then by Theorem \ref{thm:classlorlat}, given any Lorentzian lattice $\Lambda\subset\R^{m,n}$ there exists $\mathcal{O}_\Lambda\in\mathrm{O}(m,n,\R)$ such that $\mathcal{G}_\Lambda:=\mathcal{G}_{\mathrm{II}_{m,n}}\mathcal{O}_\Lambda$ is a generator matrix for $\Lambda$. 
The even integral sublattice $(\mathrm{II}_{m,n})_0$ of $\mathrm{II}_{m,n}$ is given by
\begin{equation}\label{eq:genmatIImn0}
\mathcal{G}_{(\mathrm{II}_{m,n})_0}:=\left(\begin{array}{cc}
\mathcal{G}_{m} & 0_{m\times n}\\0_{n\times m}&\mathcal{G}_{n}    
\end{array}\right)  \, ,  
\end{equation}
where 
\begin{equation}
    \mathcal{G}_{m}:=\begin{pmatrix}
    1 & 0 & 0 & \cdots & 0 & -1 \\ 0 & 1 & 0 & \cdots & 0 & -1 \\ \vdots & \vdots & \vdots & \cdots & \vdots & \vdots \\ 0 & 0 & 0 & \cdots & 1 & -1 \\ 0 & 0 & 0 & \cdots & 0 & 2
    \end{pmatrix}_{m\times m} 
    \, ,
\end{equation}
and $\mathcal{G}_{n}$ is defined similarly. \\
Using \eqref{eq:no_cosets}, one can show that $\left| \mathrm{II}_{m,n} / (\mathrm{II}_{m,n})_{0} \right| =4$, and hence the number of modules for the LLVOA constructed using II$_{m,n}$, is 4. We choose the following  representatives of  $\mathrm{II}_{m,n} / (\mathrm{II}_{m,n})_{0} \cong \mathbb{Z}_2\times \mathbb{Z}_2$: 
\begin{equation}\label{eq:rep_modulo_lattice}
        v_0 = [(0,\dots,0)]\,, \quad
        v_1 = [(1,0,\dots, -1)]\,, \quad
        v_2 = \bigg[\bigg(\frac{1}{2},\dots,\frac{1}{2}\bigg)\bigg]\,, \quad
        v_3 = v_1 + v_2\,,
\end{equation}
where $v_1$ and $v_2$ are the generators of two $\mathbb{Z}_2$ respectively. 
The $\mathcal{T}$ matrices are 
\begin{equation}
T^L =\left(\begin{array}{cccc}
{\rm e}^{-\frac{i m \pi}{12} } & 0 & 0 & 0 \\
0 & -{\rm e}^{-\frac{i m \pi}{12}} & 0 & 0 \\
0 & 0 & {\rm e}^{\frac{i m \pi}{6}} & 0 \\
0 & 0 & 0 & {\rm e}^{\frac{i m\pi}{6}}
\end{array}\right), 
\quad 
T^R =\left(\begin{array}{cccc}
{\rm e}^{-\frac{i n \pi}{12} } & 0 & 0 & 0 \\
0 & -{\rm e}^{-\frac{i n \pi}{12}} & 0 & 0 \\
0 & 0 & {\rm e}^{\frac{i n \pi}{6}} & 0 \\
0 & 0 & 0 & {\rm e}^{\frac{i n\pi}{6}}
\end{array}\right)\,.
\end{equation}
Given the left-moving central charge $c = m$, 
the conformal weight of each chiral primary 
can be read off from $T^L$, 
\begin{equation}
    \Delta_0^{L} = 0, 
    \quad 
    \Delta_1^{L} = \frac{1}{2}, \quad 
    \Delta_2^{L} = 
    \Delta_3^{L} = \frac{m}{8}\bmod 1~, 
\end{equation}
while
\begin{equation}
    \Delta_0^{R} = 0, 
    \quad 
    \Delta_1^{R} = \frac{1}{2}, \quad 
    \Delta_2^{R} = 
    \Delta_3^{R} = \frac{n}{8}\bmod 1~. 
\end{equation}
The corresponding twists are 
\begin{equation}
    \theta^L_0 = 1 , \quad 
    \theta^L_1 = -1 ,
    \quad 
    \theta^L_2 = 
    \theta^L_3 =
    {\rm e}^{\frac{\pi i m}{4}},
\end{equation}
\begin{equation}
    \theta^R_0 = 1 , \quad 
    \theta^R_1 = -1 ,
    \quad 
    \theta^R_2 = 
    \theta^R_3 =
    {\rm e}^{\frac{\pi i n}{4}}.
\end{equation}
\begin{comment}
We define the $S$ matrix for these non-chiral Narain CFTs as 
\begin{equation}
    \chi_{l}\left(-\frac{1}{\tau},-\frac{1}{\bar{\tau}}\right) = \sum_{l'}S_{l,l'}\chi_{l'}\left(\tau, \bar{\tau}\right) \, , 
\end{equation}
where both $l, l'$ are elements inside $\mathrm{II}_{m,n} / (\mathrm{II}_{m,n})_{0} \cong \mathbb{Z}_2\times \mathbb{Z}_2$, i.e. one of the four vectors in Equation \ref{eq:rep_modulo_lattice}. Each of the above chiral character, $\chi_{l}(\tau, \bar{\tau})$, is given by 
\begin{equation}
    \chi_{l}(\tau, \bar{\tau}) = \frac{1}{\eta(\tau)^{m}\overline{\eta(\tau)}^{n}}\sum_{(\alpha,\beta)\in l +\Lambda_0}q^{\frac{\langle\alpha,\alpha\rangle}{2}}\bar{q}^{\frac{\langle\beta,\beta\rangle}{2}} \, , 
\end{equation}
where $q = \mathrm{e}^{2 \pi {\rm i} \tau}$ and $ \bar{q} = \mathrm{e}^{ - 2 \pi {\rm i} \bar{\tau}} $ and $\tau$ and $\bar{\tau}$ are complex conjugates of each other. 
\end{comment}
We have $T^L = T^R $ since $m \equiv n \mod 24$.
The $S$ matrices are 
\begin{equation}
S^L=\left(\begin{array}{cccc}
\frac{1}{2} & \frac{1}{2} & \frac{1}{2} & \frac{1}{2} \\
\frac{1}{2} & \frac{1}{2} & -\frac{1}{2} & -\frac{1}{2} \\
\frac{1}{2} & -\frac{1}{2} & \frac{1}{2} {\rm e}^{-\frac{1}{2} i m \pi} & \frac{1}{2} {\rm e}^{-\frac{1}{2} i(2+m) \pi} \\
\frac{1}{2} & -\frac{1}{2} & \frac{1}{2} {\rm e}^{-\frac{1}{2} i(2+m) \pi} & \frac{1}{2} {\rm e}^{-\frac{1}{2} i(8+m) \pi}
\end{array}\right)\, ,
\end{equation}
\begin{equation}
S^R=\left(\begin{array}{cccc}
\frac{1}{2} & \frac{1}{2} & \frac{1}{2} & \frac{1}{2} \\
\frac{1}{2} & \frac{1}{2} & -\frac{1}{2} & -\frac{1}{2} \\
\frac{1}{2} & -\frac{1}{2} & \frac{1}{2} {\rm e}^{-\frac{1}{2} i n \pi} & \frac{1}{2} {\rm e}^{-\frac{1}{2} i(-2+n) \pi} \\
\frac{1}{2} & -\frac{1}{2} & \frac{1}{2} {\rm e}^{-\frac{1}{2} i(-2+n) \pi} & \frac{1}{2} {\rm e}^{-\frac{1}{2} i n \pi}
\end{array}\right)\,.
\end{equation}
Note that $S^L=S^R$ since $m\equiv n \mod 8$.
% We can verify the Verlinde formula \eqref{eq:verlform2}, while 
% the LHS is given by
% \begin{equation}
%    \mathcal{N}_{i,j}^k 
%    = \delta_{v_i+v_j, v_k}, 
%    \quad 
%    i,j,k = 0 \dots 3.
% \end{equation}
The braiding and fusing matrices can be computed using \eqref{eq:braidL} and \eqref{eq:fuseL}.

When $m\equiv 0 \bmod 8$, the LLCFT with modular data $(N,S,T)$ realizes the toric code MTC with topological central charge $c_{\mathsf{MTC}} = 0$ \cite{rowell2009classification}. 
When $m\equiv 4 \mod 8$, the LLCFT realizes the $D_4$ MTC with topological central charge $c_{\mathsf{MTC}} = 4$ \cite{rowell2009classification}. More generally,  
when $m \equiv 2, 4, 6, 8 \bmod 8$, the corresponding MTCs are the Kac Moody $D( m \mod 8)$ at level $1$ respectively \cite{Hoehn}. Thus we have identified the MTCs realised by a large class of LLCFTs based on $\mathrm{II}_{m,n}$. \\\\
One can also realise the MTC corresponding to the compact boson at radius $R=2$ from the LLCFT based on II${}_{m,1}$ with $m-1\equiv 0\bmod 24$. Again, we note that there can be other choices of braiding and fusing matrices which realize the modular symbol for the LLCFT. 

\section{Discussion And Future Directions}
In Section \ref{sec:mtc} we described the well-known mathematically rigorous construction of modular tensor category associated to a vertex operator algebra. This gives the construction of an MTC from a modular invariant rational CFT: the (left or right-moving) chiral algebra of a rational CFT is a ``nice'' VOA and the MTC associated to this VOA is the MTC of the rational CFT. We explicitly demonstrated the construction for the LLCFT. \par 
However, we seek a construction of the category of modules of a generic non-chiral VOA defined in \cite{Singh:2023mom}. In general, a non-chiral VOA is not a tensor product of left-moving and right-moving sectors and hence the understanding of a tensor bifunctor in a suitable subcategory of the category of modules of a non-chiral VOA will most likely prove to be fruitful.  
% Recall that, in the abelian category of $C_2$-cofinite rational VOA \cite{Zhu1995ModularIO} and its modules, there exists a tensor functor $\boxtimes$ which acts on irreducible modules $\{W_\mu\}_{\mu=1}^N$ as 
% \begin{equation}
%     W_\mu\boxtimes W_\nu\cong\bigoplus_\rho \s{N}_{\mu\nu}^\rho W_\rho~,
% \end{equation}
% where $\s{N}_{\mu\nu}^\rho$ are fusion rules. The existence of such a tensor functor was assumed in Moore-Seiberg \cite{Moore:1988uz,Moore:1988qv,Moore:1989vd} who used it to deduce that the category is a modular tensor category. The explicit construction of the tensor functor was given by Huang-Lepowsky in \cite{huang1995theory1,huang1995theory2,huang1995theory3,huang1995theory4} and it was shown by Huang that the category of suitable modules of a $C_2$-cofinite rational VOA is in fact a modular tensor category \cite{huang2005vertex1,huang2008vertex2,huang2008rigidity}. The construction of the so called $\s{P}(z)$-tensor product \cite{huang1995theory1,huang1995theory2,huang1995theory3,huang1995theory4}, which is used to construct the tensor product of modules, uses the analytic properties of the intertwining operators.

In particular, we would like to understand the category of modules of the LLVOA. We have already classified all irreducible modules of the LLVOA in Section \ref{sec:moduleduallattice}. Suppose that one can construct a tensor product of irreducible modules of the\footnote{One can hope to construct the tensor product of irreducible modules of the LLVOA following the approach of Huang-Lepowsky \cite{huang1995theory1,huang1995theory2,huang1995theory3,huang1995theory4} since one only uses the analytic properties of the intertwining operators which for the irreducible modules of the LLVOA can be explicitly written down. We leave this for future investigation.} LLVOA with the property that
\begin{equation}\label{eq:tensorfunctor}
    V(\mu)\boxtimes V(\nu)\cong \s{N}_{\mu\nu}^{\rho}V(\rho)~.
\end{equation}
Then by Theorem \ref{thm:interfusion} and Theorem \ref{thm:Nmnp<=1}, we see that if we want the category of modules of the LLVOA to be closed under the tensor functor, then we need to restrict to the subset of irreducible modules (and their finite direct sums)  corresponding to the equivalence classes of $\widetilde{\Lambda}/\Lambda_0$ where $\widetilde{\Lambda}\subset\Lambda_0^\circ$ is an additive group group containing $\Lambda_0$. One obvious choice of $\widetilde{\Lambda}$ is $\Lambda$ which corresponds to the LLCFT. Then we can form a category with simple objects $\{V(\mu)\}_{[\mu]\in\Lambda/\Lambda_0}$. It would be interesting to see if there are other possibilities of such additive subgroups.   We will give the explicit construction of the tensor functor for this category in a future publication.
\par We remark that the direct sum of all irreducible modules of an LLVOA along with all vertex operators and intertwining operators constructed in Section \ref{sec:intertwiners} forms a full field algebra in the sense of \cite{Huang:2005gz}. Let us call this a Lorentzian lattice full field algebra (LLFFA). In fact, it is a full field algebra over the tensor product $V_{\Lambda_1^0}\otimes V_{\Lambda_2^0}$ of VOAs based on Euclidean lattices $\Lambda_1^0,\Lambda_2^0$. Now $V_{\Lambda_1^0}\otimes V_{\Lambda_2^0}$ is also a ``nice'' VOA and hence the category $\mathcal{C}_{V_{\Lambda_1^0}\otimes V_{\Lambda_2^0}}$ also has a braided tensor category structure. The explicit construction is given in \cite[Section 3]{kong2007full}. By a result in \cite{Huang:2005gz}, LLFFA can be identified with a suitable object in the category $\mathcal{C}_{V_{\Lambda_1^0}\otimes V_{\Lambda_2^0}}$. This object was then proven in \cite{kong2007full} to be a commutative associative algebra. What we are claiming in the above paragraph is that a suitable subcategory of modules of an LLVOA can be identified with a subcategory of $\mathcal{C}_{V_{\Lambda_1^0}\otimes V_{\Lambda_2^0}}$ and is itself a braided tensor category.  
% It was used in \cite{Singh2024c} that 
% \begin{equation}
%     \s{N}_{\mu\nu}^{\mu+\nu}=1~.
% \end{equation} 
% , i.e., $\widetilde{\Lambda}\subset\Lambda_0^\circ$ is the maximal subset such that  
% \begin{equation}
%     \widetilde{\Lambda}:=\{\lambda\in\Lambda_0^\circ:\lambda+\Lambda_0^\circ\subset\Lambda_0^\circ\}~.
% \end{equation}
%\begin{thm}
% Suppose $\Lambda$ is selfdual and $\left|\Lambda / \Lambda_0\right|<\infty$. Let $\mu, \nu \in \Lambda_0^\circ$ such that $\mu+\nu \in \Lambda_0^\circ$. Then $\mu, \nu \in \Lambda$. In particular $\widetilde{\Lambda}=\Lambda$.
% \end{thm}
% \begin{proof}
% Using the fact that $\mu=(\alpha_\mu,\beta_\mu),\nu=(\alpha_\nu,\beta_\nu)\in\Lambda_0^\circ$ and $\mu+\nu\in\Lambda_0^\circ$, it is easy to see that
% \begin{equation}
%    \langle\alpha_\mu,\alpha_\nu\rangle-\langle\beta_\mu,\beta_\nu\rangle\in\Z~.
% \end{equation}
% Since $\Lambda_0^\circ\subset\Lambda'$, there exists $\alpha_\mu',\alpha_\nu'\in\R^m,\beta_\mu',\beta_\nu'\in\R^n$ such that 
% \begin{equation}
%     (\alpha_\mu,\beta_\mu'),\quad(\alpha_\nu,\beta'_\nu),\quad(\alpha_\mu',\beta_\nu),\quad(\alpha'_\nu,\beta_\mu)\in\Lambda~.
% \end{equation}
% \end{proof}
\\\\
\textbf{Acknowledgments.} The authors would like to thank Masahito Yamazaki for discussions, several useful suggestions, including potential directions, and Anatoly Dymarsky for useful comments on a draft. We would also like to thank Yi-Zhi Huang for numerous useful discussions on modular tensor categories and vertex operator algebras and his comments on an earlier version of the paper. We would like to thank Anindya Banerjee and Spencer Stubbs for various discussions on chiral algebras. R.K.S would like to thank Ecole de Physique des Houches for hospitality where part of this work was done. The work of R.K.S and R.T was  supported by the US Department of Energy under grant DE-SC0010008. 
\begin{appendix}
\section{Properties Of The Sublattices $\Lambda_1^0$ and $\Lambda_2^0$ }\label{app:sublattices}
We begin with the following elementary result from linear algebra.
\begin{lemma}\label{thm:det_identity}
Let $A,B$ be invertible $r\times r$ matrices such that $\Z^rB\subseteq \Z^rA$. Then we have 
\begin{equation}
    \left|\Z^rA/\Z^rB\right|=\left|\frac{\mr{det}(B)}{\mr{det}(A)}\right|~.
\end{equation}
\begin{proof}
We can write  
\begin{equation}
    \Z^rA/\Z^rB=\Z^rA/\Z^rAM,\quad M=A^{-1}B~.
\end{equation}
Then using the isomorphism 
\begin{equation}
\begin{split}
     \Z^rA&\longrightarrow \Z^r \,, \\\vec{a}_i&\longmapsto \vec{e}_i,~i=1,\dots,r \,, 
\end{split}
\end{equation}
where $\{\vec{e}_i\}_{i=1}^r$ is the standard basis of $\Z^r$ and $\{\vec{a}_i\}_{i=1}^r$ are the rows of $A$, we see that 
\begin{equation}
\Z^rA/\Z^rAM\cong \Z^r/\Z^rM~.    
\end{equation}
The result now follows using $|\Z^r/\Z^rM|=|\mr{det}(M)|$ for any invertible integral matrix $M$. Note that $M=A^{-1}B$ is integral since $\Z^rB\subseteq \Z^rA$.
\end{proof}
\end{lemma}
%\fixme{I put this lemma out of Lemma A.2, because it only requires $\Lambda$ to be integral and Lorentzian. Seems to make the proof in Lemma A.2 more clear}
\begin{lemma}\label{le:injection_quotient_group }
Let $\Lambda$ be an integral Lorentzian lattice, with sublattices  $\Lambda_1^0$ and $\Lambda_2^0$ defined in \eqref{eq:Lambi0}. The following maps 
\begin{equation}\label{eq:inj_map_cosets}
 \begin{aligned}
     \varphi_1: \quad\Lambda/\Lambda_0 
     &\to (\Lambda_1^0)^\star/\Lambda_1^0 \, , \\
     [(\alpha,\beta)] &\mapsto [\alpha] \, ,
 \end{aligned}
 \end{equation}
  \begin{equation}
 \begin{aligned}
     \varphi_2: \quad\Lambda/\Lambda_0 
     &\to (\Lambda_2^0)^\star/\Lambda_2^0 \, , \\
     [(\alpha,\beta)] &\mapsto [\beta] \, .
 \end{aligned}
 \end{equation}
 are injective group homomorphisms. 
\end{lemma}
\begin{proof}
First, we show that these maps are well-defined homomorphisms: given $\mu = (\mu_1,\mu_2)\in \Lambda= \Lambda^\star$,
since $\Lambda_0\subseteq \Lambda$, for any element $\alpha =(\alpha,0) \in \Lambda_1^0$, we have
$
\mu \cdot \alpha = \langle \mu_1,\alpha \rangle \in \mathbb{Z} 
$.
This implies $\mu_1\in (\Lambda_1^0)^\star$, and similarly $\mu_2\in (\Lambda_2^0)^\star $.
Let $\mu, \tilde{\mu}$, such that $\mu-\tilde{\mu} \in \Lambda_0 = \Lambda_1^0 \oplus \Lambda_2^0 $, then
\begin{equation}
\mu_1 -\tilde{\mu}_1 \in \Lambda^1_0, \quad 
\mu_2 -\tilde{\mu}_2 \in \Lambda^2_0,
\end{equation}
which implies $[\mu_1] = [\tilde{\mu}_1] \in (\Lambda_1^0)^\star/\Lambda_1^0$ and $[\mu_2] = [\tilde{\mu}_2]\in (\Lambda_2^0)^\star/\Lambda_2^0$, therefore the maps $\varphi_1$ and $\varphi_2$ are well defined.
$\varphi_1$ and $\varphi_2$ are easily seen to be group homomorphisms under the natural group structure of their respective codomains.

Next, we show that $\varphi_1$ and $\varphi_2$ are injective.
Given
\begin{equation}
\varphi_1:[(\alpha_1,\beta_1)]\mapsto [\alpha_1], \quad 
\varphi_1:
[(\alpha_2,\beta_2)]\mapsto [\alpha_2],
\end{equation}
If $[\alpha_1] = [\alpha_2]\in (\Lambda_1^0)^\star/\Lambda_1^0$, then $\alpha_1-\alpha_2 \in \Lambda^1_0$, which implies $(\alpha_1-\alpha_2,0) \in \Lambda_0 \subseteq \Lambda$.
Then
\begin{equation}
(0, \beta_1-\beta_2)=(\alpha_1-\alpha_2, \beta_1-\beta_2) - (\alpha_1-\alpha_2, 0)\in \Lambda \, . 
\end{equation}
So we have $(0,\beta_1-\beta_2)\in \Lambda_2^0$. This implies $(\alpha_1-\alpha_2, \beta_1-\beta_2)\in \Lambda_0$ and hence $[(\alpha_1, \beta_1)] = [(\alpha_2,\beta_2)]$. Therefore, the map $\varphi_1$ is injective, and similarly, we can show that $\varphi_2$ is injective. 
\end{proof}

\begin{lemma}\label{eq:sublattice_determinant}
An even, self-dual Lorentzian lattice $\Lambda$, with sublattices $\Lambda_1^0$ and $\Lambda_2^0$, and  
$\left| \Lambda / \Lambda_0 \right| < \infty $
satisfies 
\begin{equation}
\label{eq:det_sublattice}
   \left| \rm{det}(\s{G}_{\Lambda_1^0} ) \right| = \left| \rm{det}(\s{G}_{\Lambda_2^0} ) \right| \, ,
\end{equation}
or equivalently we have 
\begin{equation}
    \left| \Lambda/ \Lambda_0 \right| = \left| (\Lambda_1^0)^\star/\Lambda_1^0 \right| 
    = 
    \left| (\Lambda_2^0)^\star/\Lambda_2^0 \right|,
\end{equation}
so that the homomorphisms defined in Lemma \ref{le:injection_quotient_group } are both isomorphisms.
Moreover, $\Lambda_1^0$ is full rank if and only if $\Lambda_2^0$ is full rank. 
\end{lemma}
\begin{proof}
 From Lemma \ref{thm:det_identity} we know 
\begin{equation}
\begin{split}
    \left| \Lambda/ \Lambda_0 \right|  = \left| \rm{det}(\s{G}_{\Lambda_1^0}) \right| \left| \rm{det}(\s{G}_{\Lambda_2^0}) \right|   \, , \\
    \left| (\Lambda_1^0)^\star/\Lambda_1^0 \right|  = \rm{det}(\s{G}_{\Lambda_1^0})^2  \, ,  \\
    \left| (\Lambda_2^0)^\star/\Lambda_2^0  \right| = \rm{det}(\s{G}_{\Lambda_2^0})^2 \, . 
\end{split}
\end{equation}
Hence, we have the inequalities 
\begin{equation}
\begin{split}
       \left| \rm{det}(\s{G}_{\Lambda_1^0}) \right| \left| \rm{det}(\s{G}_{\Lambda_2^0}) \right|  \geq  \rm{det}(\s{G}_{\Lambda_1^0})^2 \, ,\\ 
        \left| \rm{det}(\s{G}_{\Lambda_1^0}) \right| \left| \rm{det}(\s{G}_{\Lambda_2^0}) \right|  \geq  \rm{det}(\s{G}_{\Lambda_2^0})^2 \, ,  \\  
\end{split}
\end{equation}
where we used the fact in Lemma \ref{le:injection_quotient_group } that $\varphi_1$ and $\varphi_2$ are injective maps. Using these two inequalities, \eqref{eq:det_sublattice} easily follows.

Let us assume that $\Lambda_1^0$ is full rank. We then know that 
\begin{equation}
    \left| \left( \Lambda_0^{1} \right)^{\star} /\Lambda_1^0 \right|  = \rm{det}\left( \s{G}_{\Lambda_1^0}\right)^2< \infty \, , 
\end{equation}
Between the sets $\Lambda/\Lambda_0 $ and $(\Lambda_1^0)^\star/\Lambda_1^0$, we constructed an injective map in  \eqref{eq:inj_map_cosets}, hence 
\begin{equation}\label{eq:fullrk}
    \left| \Lambda/\Lambda_0 \right| \leq  \left|(\Lambda_1^0)^\star/\Lambda_1^0 \right|  < \infty \, . 
\end{equation}
We now show that \eqref{eq:fullrk} implies that $\Lambda_0$ is full rank. Let us assume the contrary that $\Lambda_0$ is not full rank. Suppose $\{ v_i \}_{i = 1}^{m+n}$ is a basis for $\Lambda$, and $\{u_i \}_{i = 1}^{k}$ is a basis for $\Lambda_0$, then $k < m +n$. Hence, there exists a vector $v_l\in \{ v_i \}_{i = 1}^{m+n}$ such that $v_l\not\in\text{Span}_{\R}\{u_1,\dots,u_k\}$. This implies that each $[nv_k ]$ with $n\in \mathbb{Z}$ is a distinct coset in $\Lambda / \Lambda_0$ so that the cardinality of $\Lambda/\Lambda_0$ is not finite. Therefore, $\Lambda_0$  is full rank, hence $\Lambda_2^0$ is also full rank. 
\end{proof}
%\begin{lemma}Let $M$ be a free $\Z$-module of rank $r$ and $N\subseteq M$ be a free submodule of rank $r$. Then given a basis $\{x_1,\dots,x_r\}$ of $M$, there exists positive integers $d_1,\dots,d_r$ such that $\{d_1x_1,\dots,d_rx_r\}$ is a basis of $N$.    \end{lemma}\begin{proof}Obvioussine$N$ has rank $r$.    \end{proof}
%\newpage
Recall the notation in \eqref{eq:Lambi0}: 
\begin{equation}
\begin{split}
     &\Lambda_1=\{\alpha^\lambda \,  \lvert \,  \lambda=(\alpha^\lambda,\beta^\lambda)\in\Lambda\text{ for some }\beta^\lambda\in\R^n\}\subset \R^m,\\&\Lambda_2=\{\beta^\lambda \,  \lvert \,  \lambda=(\alpha^\lambda,\beta^\lambda)\in\Lambda\text{ for some }\alpha^\lambda\in\R^m\}\subset \R^n\\&\Lambda':=\Lambda_1\oplus\Lambda_2~.
\end{split}
\end{equation}
We have the following lemma.
\begin{lemma}
Suppose $\mathsf{rank}(\Lambda_0)=m+n$. 
If $\Lambda$ is self-dual then $\Lambda'=(\Lambda_1^0)^\star\oplus (\Lambda_2^0)^\star$, where $(\Lambda_i^0)^\star$ is the dual of $\Lambda_1^0$.
\begin{proof}
We observe that $\Lambda_i\subseteq (\Lambda_i^0)^\star$ so that $|\Lambda_i/\Lambda_i^0|\leq |(\Lambda_i^0)^\star/\Lambda_i^0|$ for $i=1,2$. By Lemma \ref{le:injection_quotient_group }, Lemma \ref{eq:sublattice_determinant}, we get $|(\Lambda_i^0)^\star/\Lambda_i^0|=|\Lambda/\Lambda_0|$. Next, the maps
\begin{equation}
\begin{split}
    \varphi_1 :\Lambda/\Lambda_0 \longrightarrow \Lambda_1/\Lambda^0_1,\quad &\varphi_2 :\Lambda/\Lambda_0 \longrightarrow \Lambda_2 /\Lambda^0_2\\ [(\alpha,\beta)]\longmapsto [(\alpha,0)],\quad &[(\alpha,\beta)]\longmapsto [(0,\beta)]~,
\end{split}    
\end{equation}
are injective. Thus $|\Lambda/\Lambda_0| \leq |\Lambda_i / \Lambda^0_i|$. So we conclude that 
\begin{equation}
|(\Lambda_i^0)^\star/\Lambda_i^0|=|\Lambda/\Lambda_0|\leq |\Lambda_i/\Lambda_i^0|~.    
\end{equation}
Thus we get 
\begin{equation}
    |(\Lambda_i^0)^\star/\Lambda_i^0|=|\Lambda_i/\Lambda_i^0|~,\quad i=1,2.
\end{equation}
Since $\Lambda_i\subseteq (\Lambda_i^0)^\star$ we conclude 
\begin{equation}
  \Lambda_i= (\Lambda_i^0)^\star\implies \Lambda'=(\Lambda_1^0)^\star\oplus (\Lambda_2^0)^\star~. 
\end{equation}
\end{proof}
\end{lemma}
\section{Generator Matrix Of II${}_{m,n}$ For $m+n\in 4\Z$}
\label{app:genmat}    

Let us take the lattice in  \eqref{eq:genmatIImn}, where recall $m + n \in 4 \Z$. Let us take a vector $v = (a_1,a_2, \ldots, a_{m+n})$, such that $a_i \in \Z$ and $\sum a_i \in 2 \Z$, according to   \eqref{eq:def-set-lattice}, this vector lies in II$_{m,n}$. Now, it can be shown that 
\begin{equation}\label{eq:linear_sum}
 v = \sum_{i = 1}^{m + n -2} (a_i - a_{m+n-1} )v_i + 2 a_{m + n - 1} v_{m + n} + \frac{1}{2}\left(\sum_{i = 1}^{m+n}a_i - (m +n) a_{m + n - 1}\right)  v_{m + n - 1} \, , 
\end{equation}
where $v_i$'s are the rows of the matrix in   \eqref{eq:genmatIImn}. Further, note that the term in parenthesis in   \eqref{eq:linear_sum} lies in $2 \Z$ as $\sum_{i = 1}^{m+n}a_i \in 2 \Z$ and $m+n \in 4\Z$, hence the coefficient of $v_{m + n -1}$ is an integer. Now, the other type of vector that can lie in II$_{m,n}$ is $v = (a_1,a_2, \ldots, a_{m+n})$, such that $a_i \in \Z + \frac{1}{2}$ and $\sum a_i \in 2 \Z$. Now, let 
\begin{equation}
    (a_1, a_2, \ldots, a_{m +n}) = \left(\frac{1}{2},\frac{1}{2}, \ldots, \frac{1}{2} \right) + (b_1, b_2, \ldots, b_{m +n})\, , 
\end{equation}
then $b_i \in \Z$ for all $i$ and further $\sum_{i = 1}^{m + n} b _i \in 2 \Z$ as $\sum_{i = 1}^{m + n} a_i \in 2 \Z$ and $(m + n)/2 \in 2 \Z$. Hence, we have shown that the lattice II$_{m,n}$ as defined in   \eqref{eq:def-set-lattice} is contained inside the lattice with generator matrix in   \eqref{eq:genmatIImn}. 

However, the lattice with generator matrix in   \eqref{eq:def-set-lattice} lies inside II$_{m,n}$ since all the rows of the generator matrix lie in II$_{m,n}$. Hence, the generator matrix of II$_{m,n}$ is indeed the matrix in   \eqref{eq:genmatIImn}.
\section{Comparison With Wendland's Characterization Of Narain RCFTs}\label{app:wend}
In \cite{Wendland:2000ye}, several equivalent conditions for a Narain CFT based on the lattice $\Lambda\subset\R^{m,m}$ to be rational were presented. In this appendix, we show that our characterization of rational Narain CFTs is equivalent to \cite{Wendland:2000ye}.
From \cite{Dymarsky:2020qom}, 
 for any Lorentzian lattice $\Lambda$, there exists an $\mathrm{O}(m, \mathbb{R}) \times$ $\mathrm{O}(m, \mathbb{R})$ transformation 
 which relates $\Lambda$ to a Lorentzian lattice $\Lambda_S$
 such that the generator matrix is of the following form:
$$
\mathcal{G}_{\Lambda_S}=
\frac{1}{\sqrt{2}}
\left(\begin{array}{cc}
\gamma^{\star} & \gamma^{\star} \\
\gamma (B+\mathds{1}) & \gamma (B-\mathds{1})
\end{array}\right),
$$
where $B$ is an anti-symmetric matrix, $\gamma$ is the generator matrix for a lattice $\Gamma$ in $\mathbb{R}^m$, and $\gamma^{\star}$ is the generator matrix for the dual lattice $\Gamma^{\star}$.
Therefore, we can express lattice $\Lambda_S$ as
\begin{equation}
    \Lambda_S = \left\{\frac{1}{\sqrt{2}}\left(\mu + \lambda(\mathds{1}+B), \mu - \lambda(\mathds{1}-B)
    \right): (\mu,\lambda ) \in \Gamma^*\oplus \Gamma\right\}\,.
\end{equation}
The sublattices \eqref{eq:Lambi0} of $\Lambda_S$ are defined as
\begin{equation}
    \begin{aligned}
        \Lambda_1^0  &= 
        \{(\sqrt{2}\lambda,0): \lambda \in \Gamma
        \}\,,
        \\
        \Lambda_2^0  &= 
        \{(0,-\sqrt{2}\lambda): \lambda \in \Gamma
        \}\,.
    \end{aligned}
\end{equation}
The rationality conditions in \cite{Wendland:2000ye}
are 
\begin{equation}
    \mathsf{rank}(\Lambda^0_l) = m \quad \Leftrightarrow
    \quad \mathsf{rank}(\Lambda^0_r) = m,
\end{equation}
where $\Lambda^0_l$ and 
$\Lambda^0_r$ are sublattice of $\Gamma^* \oplus \Gamma$ defined as
\begin{equation}
    \begin{aligned}
        \Lambda^0_l &= 
        \{ (\mu,\lambda) \in \Gamma^* \oplus 
        \Gamma: \mu= \lambda(-B+\mathds{1}) \}, \\
        \Lambda^0_r &= 
        \{ (\mu,\lambda) \in \Gamma^* \oplus 
        \Gamma: \mu= \lambda(-B-\mathds{1}) \}.
    \end{aligned}
\end{equation}
Note that 
\begin{equation}
    \Lambda_1^0 =  \Lambda^0_l \mathcal{G}'_{\Lambda_S}, \quad 
    \Lambda_2^0 = 
    \Lambda^0_r\mathcal{G}'_{\Lambda_S} \,,
\end{equation}
where 
\begin{equation}
\mathcal{G}'_{\Lambda_S}=
\frac{1}{\sqrt{2}}
\left(\begin{array}{cc}
\mathds{1} & \mathds{1} \\
\mathds{1}+B &  -\mathds{1}+ B
\end{array}\right),
\end{equation}
is a (invertible) linear transformation.
This implies 
    \begin{equation}
    \mathsf{rank}(\Lambda_1^0)
    =\mathsf{rank}(\Lambda_2^0) = m~,
\end{equation}
which is the condition for rationality in Theorem \ref{thm:rcft}.

\end{appendix}
\bibliography{mtc}
\end{document}